\documentclass[11pt]{article}
\usepackage{fullpage,comment,algorithm,algorithmic}
\usepackage{color}
\usepackage{float}
\usepackage{amsthm}
\usepackage{amsmath}
\usepackage{amssymb}
\usepackage{graphicx}

\newtheorem{theorem}{Theorem}
\newtheorem{lemma}{Lemma}
\newtheorem{corollary}[lemma]{Corollary}
\newtheorem{remark}{Remark}
\newtheorem{claim}[lemma]{Claim}

\newtheorem{definition}{Definition}

\newtheorem{fact}{Fact}

\newcommand{\namedref}[2]{\hyperref[#2]{#1~\ref*{#2}}}
\newcommand{\sectionref}[1]{\namedref{Section}{#1}}
\newcommand{\appendixref}[1]{\namedref{Appendix}{#1}}

\newcommand{\theoremref}[1]{\namedref{Theorem}{#1}}
\newcommand{\defref}[1]{\namedref{Definition}{#1}}
\newcommand{\figureref}[1]{\namedref{Figure}{#1}}
\newcommand{\claimref}[1]{\namedref{Claim}{#1}}
\newcommand{\lemmaref}[1]{\namedref{Lemma}{#1}}

\newcommand{\corollaryref}[1]{\namedref{Corollary}{#1}}

\newcommand{\algref}[1]{\namedref{Algorithm}{#1}}
\newcommand{\factref}[1]{\namedref{Fact}{#1}}

\newcommand{\supp}{{\rm supp}}
\newcommand{\E}{{\mathbb{E}}}
\newcommand{\N}{\mathbb{N}}
\newcommand{\R}{\mathbb{R}}
\newcommand{\diam}{{\rm diam}}
\newcommand{\1}{{\bf 1}}

\newcommand{\str}{{\rm distortion}}
\newcommand{\con}{{\rm cong}}

\usepackage[colorlinks,linkcolor=black,filecolor=black,citecolor=black,urlco
lor=black,pdfstartview=FitH]{hyperref}

\newcommand{\alert}[1]{\textbf{\color{red}
[[[#1]]]}\marginpar{\textbf{\color{red}**}}\typeout{ALERT:
\the\inputlineno: #1}}

\floatstyle{ruled}
\newfloat{algorithm}{tbp}{loa}
\providecommand{\algorithmname}{Algorithm}
\floatname{algorithm}{\protect\algorithmname}

\makeatother

\usepackage{authblk}
\usepackage{pdfsync}

\begin{document}
\title{Terminal Embeddings\thanks{A preliminary version of this paper appeared in APPROX'15.}}

\author{Michael Elkin\thanks{Supported in part by ISF grant No. (724/15).},
	Arnold Filtser\thanks{Supported in part by ISF grant No. (523/12) and by the European Union Seventh Framework Programme (FP7/2007-2013) under grant agreement $n^\circ 303809$. },
	Ofer Neiman 
	\thanks{Supported in part by ISF grant No. (523/12) and by the European Union Seventh Framework Programme (FP7/2007-2013) under grant agreement $n^\circ 303809$. }
}

\affil{Department of Computer Science, Ben-Gurion University of the Negev,
	Beer-Sheva, Israel. Email: \texttt{\{elkinm,arnoldf,neimano\}@cs.bgu.ac.il}}

\date{}
\maketitle

\begin{abstract}
In this paper we study {\em terminal embeddings}, in which one is given a finite metric $(X,d_X)$ (or a graph $G=(V,E)$) and a subset $K \subseteq X$ of its points are designated as {\em terminals}. The objective is to embed the metric into a normed space, while approximately preserving all distances among pairs that contain a terminal. We devise such embeddings in various settings, and conclude that even though we have to preserve $\approx|K|\cdot |X|$ pairs, the distortion depends only on $|K|$, rather than on $|X|$.

We also strengthen this notion, and consider embeddings that approximately preserve the distances between {\em all} pairs, but provide improved distortion  for pairs containing a terminal. Surprisingly, we show that such embeddings exist  in many settings, and have optimal distortion bounds both with respect to $X \times  X$ and with respect to $K \times X$.

Moreover, our embeddings have implications to the areas of Approximation and Online Algorithms. In particular, \cite{ALN08} devised an $\tilde{O}(\sqrt{\log r})$-approximation algorithm for sparsest-cut instances with $r$ demands. Building on their framework, we provide an $\tilde{O}(\sqrt{\log |K|})$-approximation for sparsest-cut instances in which each demand is incident on one of the vertices of $K$ (aka, terminals). Since $|K| \le r$, our bound generalizes that of \cite{ALN08}.
\end{abstract}

\section{Introduction}

Embedding of finite metric spaces is a very successful area of research, due to both its algorithmic applications and its natural geometric appeal. Given two metric space $(X,d_X)$, $(Y,d_Y)$, we say that $X$ embeds into $Y$ with distortion $\alpha$ if there is a map $f:X\to Y$ and a constant $c>0$, such that for all $u,v\in X$,
\[
d_X(u,v)\le c\cdot d_Y(f(u),f(v))\le \alpha\cdot d_X(u,v)~.
\]
Some of the basic results in the field of metric embedding are: a theorem of \cite{B85}, asserting that any metric space on $n$ points embeds with distortion $O(\log n)$ into Euclidean space (which was shown to be tight by \cite{LLR95}), and probabilistic embedding into a distribution of trees, or ultrametrics,\footnote{An ultrametric $(U,\rho)$ is a metric space satisfying $\rho(x,z)\le\max\{\rho(x,y),\rho(y,z)\}$ for all $x,y,z\in U$.} with expected distortion $O(\log n)$ \cite{FRT03}, or expected congestion $O(\log n)$ \cite{R08} (which are also tight \cite{B96}).

In this paper we study a natural variant of embedding, in which the input consists of a finite metric space or a graph, and in addition, a subset of the points are designated as {\em terminals}. The objective is to embed the metric into a simpler metric (e.g., Euclidean metric), or into a simpler graph (e.g.,  a tree), while approximately preserving the distances between the terminals to {\em all other points}. We show that such embeddings, which we call {\em terminal embeddings}, can have improved parameters compared to embeddings that must preserve all pairwise distances. In particular, the distortion (and the dimension in embedding to normed spaces) depends only on the number of terminals, regardless of the cardinality of the metric space.

We also consider a strengthening of this notion, which we call  {\em strong terminal embedding}. Here we want a distortion bound on {\em all pairs}, and in addition an improved distortion bound on pairs that contain a terminal. Such strong terminal embeddings enhance the classical embedding results, essentially saying that one can obtain the same distortion for  all pairs, with the option to select some of the points, and obtain improved approximation of the distances between any selected point to any other point.

As a possible motivation for studying such embeddings, consider a scenario in which a certain network of clients and servers is given as a weighted graph (where edges correspond to links, weights to communication/travel time). It is  conceivable that one only cares about distances between clients and servers, and that there are few servers. We would like to have a simple structure, such as a tree spanning the network, so that the client-server distances in the tree are approximately preserved.

We show that there exists a general phenomenon; essentially any known metric embedding into an  $\ell_p$ space or a graph family can be transformed
via a general transformation
into a terminal embedding, while paying only a constant blow-up in the distortion.
In particular, we obtain  a terminal embedding of any finite metric into any $\ell_p$ space with terminal distortion $O(\log k)$, using only $O(\log k)$ dimensions. We show that a similar general phenomenon\footnote{Though this transformation is somewhat less general than for the case of ordinary, i.e., not strong, terminal embeddings.} holds also for {\em strong} terminal embeddings, i.e., that many embeddings can be transformed into strong terminal ones via a general transformation. Many other embeddings into normed spaces, probabilistic embedding into ultrametrics (including capacity preserving ones), and into spanning trees, have their {\em strong} terminal embedding counterparts, which are constructed directly, that is, not through our general transformation. Our results are tight in most settings.\footnote{All  our terminal embeddings are tight, except for the probabilistic spanning trees, where they match the state-of-the-art \cite{AN12}, and except for our terminal spanners.}

It is well known that embedding a graph into a  single tree may cause (worst-case) distortion $\Omega(n)$ \cite{RR98}. However, if one only cares about client-server distances, we show that it is possible to obtain distortion $2k-1$, where $k$ is the number of servers, and that this is tight. Furthermore, we study possible tradeoffs between the distortion and the total weight of the obtained tree. This generalizes the notion of {\em shallow light trees} \cite{KRY95,ABP92,ES11}, which provides a tradeoff between the distortion with respect to a single designated server and the weight of the tree.

We then address probabilistic approximation of metric spaces and graphs by ultrametrics and spanning trees. This line of work started with the results of \cite{AKPW95,B96}, and culminated in the $O(\log n)$ expected distortion for ultrametrics by \cite{FRT03}, and $\tilde{O}(\log n)$ for spanning trees by \cite{AN12}. These embeddings found numerous algorithmic applications, in various settings, see \cite{FRT03,EEST05,AN12} and the references therein for details. In their work on Ramsey partitions, \cite{MN06} implicitly showed that there exists a probabilistic embedding into ultrametrics with expected terminal distortion $O(\log k)$ (see \sectionref{sec:prel} for the formal definitions).
We generalize this result by showing a {\em strong} terminal embedding with the same expected $O(\log k)$ distortion guarantee for all pairs containing a terminal, and $O(\log n)$ for all other pairs.
We also show a similar result that extends the embedding of \cite{AN12} into spanning trees, with $\tilde{O}(\log k)$ expected distortion for pairs containing a terminal, and $\tilde{O}(\log n)$ for all pairs. A slightly different notion, introduced by \cite{R02}, is that of trees which approximate the congestion (rather than the distortion), and  \cite{R08} showed a distribution over trees with expected congestion $O(\log n)$. We provide a strong terminal version of this result, and show expected congestion of $O(\log k)$ for all edges incident on a terminal, and $O(\log n)$ for the rest.

We also consider spanners, with a stretch requirement only for pairs containing a terminal. Our general transformation produces for any $t\ge 1$ a $(4t-1)$-terminal stretch spanner with $O(k^{1+ 1/t}+n)$ edges. The drawback is that this is a metric spanner, not a subgraph of the input graph. We alleviate this issue by constructing a graph spanner with the same stretch and $O(\sqrt{n}\cdot k^{1+ 1/t}+n)$ edges.\footnote{Note that the number of edges is linear whenever $k\le n^{1/(2(1+ 1/t))}$.} A result of \cite{RTZ05} implicitly provides a terminal graph spanner with $(2t-1)$ stretch and $O(t\cdot n\cdot k^{1/t})$ edges. Our graph terminal spanner is sparser than that of \cite{RTZ05} as long as $k \le t \cdot n^{1/2(1+1/t)}$.

\subsection{Algorithmic Applications}\label{sec:app}

We overview a few of the applications of our results to approximation and online algorithms. Some of the most striking applications of metric embeddings are to various cut problems, such as the sparsest-cut, min-bisection, and also to several online problems. Our method provides improved guarantees when the input graph has a small set of "important" vertices. Specifically, these vertices can be considered as terminals, and we obtain approximation factors that depend on the cardinality of the terminal set, rather than on the input size. The exact meaning of importance is problem specific; e.g. in the cut problems, we require that the set of important vertices touches every demand pair, or every edge (that is, forms a vertex cover).

For instance, consider the (general) {\em sparsest-cut} problem \cite{LR99,AR98,LLR95}. We are given a graph $G=\left(V,E\right)$
with capacities on the edges $c:E\rightarrow\mathbb{R}_{+}$, and a collection of pairs $(s_1,t_1),\dots,(s_r,t_r)$ along with their demands $D_1,\dots,D_r$. The goal is to find a cut $S\subseteq V$ that minimizes the ratio between capacity and demand across the cut:
\[
\phi(S)=\frac{\sum_{\{u,v\}\in E}c(u,v)|\1_S(u)-\1_S(v)|}{\sum_{i=1}^rD_i|\1_S(s_i)-\1_S(t_i)|}~,
\]
where $\1_S(\cdot)$ is the indicator for membership in $S$. Following the breakthrough result of \cite{ARV09}, which showed $O(\sqrt{\log n})$ approximation for the uniform demand case, \cite{ALN08} devised an $\tilde{O}(\sqrt{\log r})$ approximation for the general case. If there is a set of $k$ important vertices, such that every demand pair contains an important vertex, we  obtain an $\tilde{O}(\sqrt{\log k})$ approximation using the terminal embedding of negative-type metrics to $\ell_1$.
Observe that $k \le r$, and so our result subsumes the result of \cite{ALN08}. Our bound
 is particularly useful for instances with many demand pairs but few distinct sources $s_i$ (or few targets $t_i$).

We also consider other cut problems, and show a similar phenomenon: the $O(\log n)$ approximation for the {\em min-bisection} problem \cite{R08}, can be improved to an approximation of only $O(\log k)$, where $k$ is the size of the minimum vertex cover of the input graph. For this result we employ our terminal variant of R\"{a}cke's result \cite{R08} on capacity-preserving probabilistic embedding into trees.

We then focus on one application of probabilistic embedding into ultrametrics \cite{B96,FRT03}, and illustrate the usefulness of our terminal embedding result by the
(online)
{\em constrained file migration} problem \cite{BFR95}. Given a graph $G=(V,E)$ representing a network, each node $v\in V$ has a memory capacity $m_v$, and there is a set of files that reside at the nodes, at most $m_v$ files may be stored at node $v$ at any given time. The cost of accessing a file is the distance in the graph to the node storing it (no copies are allowed). Files can also be migrated from one node to another. This costs $D$ times the distance, for a given parameter $D\ge 1$. When a sequence of file requests from nodes arrives online, the goal is to minimize the cost of serving all requests. \cite{B96} showed a algorithm with $O(\log m\cdot\log n)$ competitive ratio for graphs on $n$ nodes, where $m=\sum_{v\in V}m_v$ is the total memory available.\footnote{The original paper shows $O(\log m\cdot\log^2n)$, the improved factor is obtained by using the embedding of \cite{FRT03}.}
A setting which seems particularly  natural is one where there is a small set of nodes who can store files (servers), and the rest of the nodes can only access files but not store them (clients). We employ our probabilistic terminal embedding into ultrametrics to provide a $O(\log m\cdot\log k)$ competitive ratio, for the case where there are $k$ servers. (Note that this ratio is independent of $n$.)

\subsection{Overview of Techniques}

The weak variant of our terminal embedding into $\ell_2$ maps every terminal $x$ into its image $f(x)$ under an original black-box (e.g., Bourgain's) embedding of $K$ into $\ell_2$.
This embedding is then appended with one additional coordinate. Terminals are assigned 0 value in this coordinate, while each non-terminal point $y$ is mapped to $(f(x),d(x,y))$, where $x$ is the closest terminal to $y$. It is not hard to see that this embedding guarantees terminal distortion $O(\gamma(k))$, where $\gamma(k)$ is the distortion of the original black-box embedding, i.e., $O(\log k)$ in the case of Bourgain's embedding. On the other hand, the new embedding employs only $\beta(k)+1$ dimensions, where $\beta(k)$ is the dimension of the original blackbox embedding (i.e., $O(\log^2 k)$ in the case of Bourgain's embedding).
\footnote{We can also get dimension $O(\log k)$ for terminal embeddings into $\ell_2$ by replacing Bourgain's embedding with that of \cite{ABN06}.} This idea easily generalizes to a number of quite general scenarios, and under mild assumptions (see \theoremref{thm:lp-strong}) it can be modified to produce strong terminal embeddings.

This framework, however, does not apply in many important settings, such as embedding into subgraphs, and does not provide strong terminal guarantees in others. Therefore we devise embeddings tailored to each particular setting in a non-black-box manner. For instance, our probabilistic embedding into trees with strong terminal congestion  requires an adaptation of a theorem of \cite{AF09}, about the equivalence of distance-preserving and capacity-preserving random tree embeddings, to the terminal setting. Perhaps the most technically involved is our probabilistic embedding into spanning trees with strong terminal distortion. This result requires a set of modifications to the recent algorithm of \cite{AN12}, which is based on a certain hierarchical decomposition of graphs. We adapt this algorithm by giving preference to the terminals in the decomposition (they are the first to be chosen as cluster centers), and the crux is assuring that the distortion of any pair containing a terminal is essentially not affected by choices made for non-terminals. Furthermore, one has to guarantee that each such pair can be separated in at most $O(\log k)$ levels of the hierarchy.

The basic technical idea that we use for constructing $(4t-1)$-terminal subgraph spanners with $O(\sqrt{n} k^{1+1/t} + n)$ edges is the following one.
As was mentioned above, our general transformation constructs metric (i.e., non-subgraph) $(4t-1)$-terminal spanners  with $O(n + k^{1+1/t})$ edges. The latter spanners employ some edges which do not belong to the original graph. We provide these edges as an input to a pairwise preserver. A pairwise preserver \cite{CE05} is a sparse subgraph that preserves exactly all distances between a designated set of vertex pairs. We use these preservers to fill in the gaps in the non-subgraph terminal spanner constructed via our general transformation. As a result we obtain a subgraph terminal spanner which outperforms previously existing terminal spanners of \cite{RTZ05} in a wide range of parameters.

\subsection{Related Work}

Already in the pioneering work of \cite{LLR95}, an embedding that has to provide a distortion guarantee for a subset of the pairs is presented. Specifically, in the context of the sparsest-cut problem, \cite{LLR95} devised a non-expansive embedding of an arbitrary metric into $\ell_1$, with distortion at most $O(\log r)$ for a set of $r$ specified demand pairs.

Terminal distance oracles were studied by \cite{RTZ05}, who called them {\em source restricted} distance oracles. In their paper, \cite{RTZ05} show $(2t-1)$-terminal stretch using $O(t\cdot n\cdot k^{1/t})$ space.
Implicit in our companion paper \cite{EFN15}  is  a distance oracle with $(4t-1)$-terminal stretch, $O(t\cdot k^{1/t}+n)$ space and $O(1)$ query time.
Terminal spanners with additive stretch for unweighted graphs were recently constructed in \cite{KV13}. Specifically,  they showed a spanner with $\tilde{O}(n^{5/4}\cdot k^{1/4})$ edges and additive stretch 2 for pairs containing a terminal. Another line of work introduced {\em distance preservers} \cite{CE05}; these are spanners which preserve exactly distances for a given collection of pairs.

In the context of preserving distances just between the terminals, \cite{G01,CXKR06,EGKRTT14,KKN14} studied embeddings of a graph into a minor over the terminals, while approximately preserving distances.
In their work on the requirement cut problem, among other results, \cite{GNR10} obtain for any metric with $k$ specified terminals, a distribution over trees with expected expansion $O(\log k)$ for all pairs, and which is non-contractive for terminal pairs. (Note that this is different from our setting, as the extra guarantee holds for terminals only, not for pairs containing a terminal.)

Another line of research \cite{M09,CLLM10,MM10,EGKRTT14} studied cut and vertex sparsifiers. A {\em cut sparisifier} of a graph $G = (V,E)$ with respect to a subset $K$ of terminals is a graph $H = (K,E_H)$ on just the set of terminals, so that for any subset $A \subset K$, the minimum value of a cut in $G$ that separates $A$ from $K \setminus A$ is approximately equal to the value of the cut $(A,K \setminus A)$ in $H$. Note that this notion is substantially different from the notion of terminal congestion-preserving embedding, which we study in the current paper.

\subsection{Subsequent Work}

In a companion paper \cite{EFN15}, we study prioritized metric structures and embeddings. In that setting, along with the input metric $(X,d)$, a priority ranking of the points of $X$ is given, and the goal is to obtain a data structure (distance oracle, routing scheme) or an embedding with stretch/distortion that depends on the ranking of the points. This has some implications to the terminal setting, since the $k$ terminals can be given as the first $k$ points in the priority ranking. More concretely, implicit in \cite{EFN15} is an embedding into a single (non-subgraph) tree with strong terminal distortion $O(k)$, a probabilistic embedding into ultrametrics with expected strong terminal distortion $O(\log k)$, and embedding into $\ell_p$ space with strong terminal distortion $\tilde{O}(\log k)$. In the current paper we provide stronger and more general results: our single tree embedding has tight $2k-1$ stretch, the tree is a subgraph, and it can have low weight as well (at the expense of slightly increased stretch); we obtain probabilistic embedding into {\em spanning} trees, and in congestion-preserving trees; and our terminal embedding to $\ell_p$ space has a tight strong terminal distortion $(O(\log k),O(\log n))$ and low dimension. Furthermore, the results of this paper apply to numerous other settings (e.g., embeddings tailored for graphs excluding a fixed minor, negative-type metrics, spanners, etc.).

Following our work, \cite{BFN16} discovered a connection between terminal distortion and coarse partial distortion. First recall the notion of coarse partial distortion, introduced in \cite{KSW04,ABCD05}. Let $(X,d_X)$ be a metric space of size $|X|=n$. For $x\in X$ and $\epsilon\in\left(0,1\right)$, let $R(x,\epsilon)=\min\left\{ r:|B(x,r)|\ge\epsilon n\right\}$. A point $y$ is called $\epsilon$-{\em far} from $x$ if $d_X(x,y)\ge R(x,\epsilon)$. We say that an embedding $f:X\to Y$ has \emph{coarse  $(1-\epsilon)$-partial distortion} $\lambda$, if every pair $x,y\in X$ such that both $x,y$ are $\epsilon/2$-far from each other, has distortion at most $\lambda$. The connection between this notion and terminal distortion is roughly as follows.
\begin{itemize}
\item If a metric admits an embedding (into some target space) with terminal distortion $\gamma$ for a certain set of $k$ terminals, then this very embedding has coarse $(1-\frac{8}{k})$-partial distortion $5\cdot \gamma$.
\item If every metric admits an embedding with coarse $(1-\frac{1}{k})$-partial distortion $\gamma$, then every metric has embedding with terminal distortion $\gamma$ (for any set of $k$ terminals).
\end{itemize}

The terminal embedding results presented here are strictly stronger than those obtainable by using the state-of-the-art coarse partial embeddings with the above transformation.
E.g., by \cite{ABCD05} it follows that if every $n$-point metric embeds into $\ell_p$ with distortion $\alpha(n)$ and dimension $\beta(n)$ (where the embedding needs to fulfill a certain condition), then it embeds into $\ell_p$ with terminal distortion $\alpha(O(k))$ and dimension $\beta(O(k))\cdot O(\log n)$.
Our generic terminal embedding of \theoremref{thm:trans} provides dimension independent of $n$, does not restrict the original embedding, and has improved constants in the distortion. Also our embedding into spanning trees with $(\tilde{O}(\log k),\tilde{O}(\log n))$-strong terminal distortion improves the embedding obtainable by going through the coarse partial embedding of \cite{ABN07}, which would give only $(\tilde{O}(\log^2k),\tilde{O}(\log^2 n))$-strong terminal distortion.

\subsection{Organization}
The general transformations are presented in \sectionref{sec:trans}.  The results on graph spanners appear in \sectionref{sec:spanner-oracle}.
The tradeoff between terminal distortion and lightness in a single tree embedding is shown in \sectionref{sub:Light-k-term-gen}. Corresponding lower bounds in several settings appear in \sectionref{sec:Lower-Bounds} and in \appendixref{sub:lwr bnd metric}.
The probabilistic embedding into ultrametrics with strong terminal distortion appears in \sectionref{sec:strong-ultra}, the congestion preserving variant is in \sectionref{sec:cong}, and the probabilistic embedding into spanning trees is shown in \sectionref{sec:Distrbution spanninig tree}.
Algorithmic applications are described in \sectionref{app:app}.

\section{Preliminaries}\label{sec:prel}

Here we provide formal definitions for the notions of terminal distortion. Let $(X,d_X)$ be a finite metric space, with $K\subseteq X$ a set of terminals.
Throughout the paper we assume $|K|\le |X|/2$.

\begin{definition}
Let $(X,d_X)$ be a metric space, and let $K\subseteq X$ be a subset of terminals.  For a target metric $(Y,d_Y)$, an embedding $f:X\to Y$ has {\em terminal distortion} $\alpha$ if there exists $c>0$, such that for all $v\in K$ and $u\in X$,\footnote{In most of our results the embedding has a one-sided guarantee (that is, non-contractive or non-expansive) for all pairs.}
\[
d_X(v,u)\le c\cdot d_Y(f(v),f(u))\le \alpha\cdot d_X(v,u)~.
\]

We say that the embedding has {\em strong terminal distortion} $(\alpha,\beta)$ if it has terminal distortion $\alpha$, and in addition there exists $c'>0$, such that for all $u,w\in X$,
\[
d_X(u,w)\le c'\cdot d_Y(f(u),f(w))\le \beta\cdot d_X(u,w)~.
\]
\end{definition}
For a graph $G=(V,E)$ with a terminal set $K\subseteq V$, an {\em $\alpha$-terminal (metric) spanner} is a graph $H$ on $V$ such that for all $v\in K$ and $u\in V$,
\begin{equation}\label{eq:spanner}
d_G(u,v)\le d_H(u,v)\le \alpha\cdot d_G(u,v)~.
\end{equation}
$H$ is a {\em graph} spanner if it is a subgraph of $G$.

Denote by $\diam(X)=\max_{y,z\in X}\{d_X(y,z)\}$. For any $x\in X$ and $r\ge 0$ let
$B_X(x,r)=\{y\in X\mid d_X(x,y)\le r\}$ (we often omit the subscript when the metric is clear from context). For a point $x\in X$ and a subset $A\subseteq X$, $d_X(x,A)=\min_{a\in A}\{d_X(x,a)\}$. For $K\subseteq X$ we denote by $(K,d_K)$ the metric space where $d_K$ is the induced metric.

For a weighted graph $G=(V,E,w)$ where $w:E\to\R_+$, given a subgraph $H$ of $G$, let $w(H)=\sum_{e\in E(H)}w\left(e\right)$, and define the {\em lightness} of $H$ to be $\Psi\left(H\right)=\frac{w\left(H\right)}{w\left(MST\left(G\right)\right)}$, where $w(MST)$ is the weight of a minimum spanning tree of $G$.

By $\tilde{O}(f(n))$ we mean $O(f(n)\cdot{\rm polylog}(f(n)))$.
\section{A General Transformation}\label{sec:trans}

In this section we present general transformation theorems that create terminal embeddings into normed spaces and graph families from standard ones. We say that a family of graphs ${\cal G}$ is {\em leaf-closed}, if it is closed under adding leaves. That is, for any $G\in{\cal G}$ and $v\in V(G)$, the graph $G'$ obtained by adding a new vertex $u$ and connecting $u$ to $v$ by an edge, belongs to ${\cal G}$. Note that many natural families of graphs are leaf-closed, e.g. trees, planar graphs, minor-free graphs, bounded tree-width graphs, bipartite graphs, general graphs, and many others.

\begin{theorem}\label{thm:trans}
Let ${\cal X}$ be a family of metric spaces. Fix some $(X,d_X)\in{\cal X}$, and let $K\subseteq X$ be a set of terminals of size $|K|=k$, such that $(K,d_K)\in{\cal X}$. Then the following assertions hold:
\begin{itemize}
\item If there are functions $\alpha,\gamma:\N\to\R$, such that every $(Z,d_Z)\in{\cal X}$ of size $|Z|=m$ embeds into $\ell_p^{\gamma(m)}$ with distortion $\alpha(m)$, then there is an embedding of $X$ into $\ell_p^{\gamma(k)+1}$ with {\em terminal} distortion $2^{(p-1)/p}\cdot((2\alpha(k))^p+1)^{1/p}$.\footnote{Note that for any $p,\alpha\ge 1$ we have that $2^{(p-1)/p}\cdot((2\alpha)^p+1)^{1/p}\le 4\alpha$.}

\item If ${\cal G}$ is a leaf-closed family of graphs, and any $(Z,d_Z)\in{\cal X}$ of size $|Z|=m$ embeds into ${\cal G}$ with distortion $\alpha(m)$ such that the target graph has at most $\gamma(m)$ edges, then there is an embedding of $X$ into ${\cal G}$ with {\em terminal} distortion $2\alpha(k)+1$ and the target graph has at most $\gamma(k)+n-k$ edges.
\end{itemize}

\end{theorem}
\noindent{\bf Remark:} The second assertion holds under probabilistic embeddings as well.

\begin{proof}
	We start by proving the first assertion.
	By the assumption there exists an embedding $f:K\to\R^{\gamma(k)}$ with distortion $\alpha(k)$ under the $\ell_p$ norm. We assume w.l.o.g that $f$ is non-contractive. For each $x\in X$, let $k_x\in K$ be the nearest point to $x$ in $K$ (that is, $d(x,K)=d(x,k_x)$). Define the embedding $\hat{f}:X\to\R^{\gamma(k)+1}$ by letting for $x\in X$, $\hat{f}(x) = (f(k_x),d(x,k_x))$. Observe that for $t\in K$, $\hat{f}(t)=(f(t),0)$.
	Fix any $t\in K$ and $x\in X$. Note that by definition of $k_x$, $d(x,k_x)\le d(x,t)$, and by the triangle inequality, $d(t,k_x)\le d(t,x)+d(x,k_x)\le 2d(t,x)$, so that,
	\begin{eqnarray*}
		\|\hat{f}(t)-\hat{f}(x)\|_p^p &=& \|f(t)-f(k_x)\|_p^p + d(x,k_x)^p \\
		&\le& (\alpha(k)\cdot d(t,k_x))^p + d(x,k_x)^p\\
		&\le&(2\alpha(k)\cdot d(t,x))^p + d(t,x)^p\\
		&=&d(t,x)^p\cdot ((2\alpha(k))^p+1)~.
	\end{eqnarray*}
	On the other hand, since $f$ does not contract distances,
	\begin{eqnarray*}
		\|\hat{f}(t)-\hat{f}(x)\|_p^p &=& \|f(t)-f(k_x)\|_p^p + d(x,k_x)^p \\
		&\ge& d(t,k_x)^p + d(x,k_x)^p\\
		&\ge& (d(t,k_x) + d(x,k_x))^p/2^{p-1}\\
		&\ge& d(x,t)^p/2^{p-1}~,
	\end{eqnarray*}
	where the second inequality is by the power mean inequality.
	We conclude that the terminal distortion is at most $2^{(p-1)/p}\cdot((2\alpha(k))^p+1)^{1/p}$.
	
	For the second assertion, there is a non-contractive embedding $f$ of $K$ into $G\in{\cal G}$ with distortion at most $\alpha(k)$. As above, for each $x\in X\setminus K$ define $k_x$ as the nearest point in $K$ to $x$. And for each $x\in X$, add to $G$ a new vertex $f(x)$ that is connected by an edge of length $d_G(x,k_x)$ to $f(k_x)$. The resulting graph $G'\in {\cal G}$, because it is a leaf-closed family. Fix any $x\in X$ and $t\in K$, then as above $d(t,k_x)\le 2d(t,x)$, and so
	\begin{eqnarray*}
		d_{G'}(f(t),f(x)) &=& d_{G}(f(t),f(k_x))+d_{G'}(f(x),f(k_x))\\
		&\le& \alpha(k)\cdot d(t,k_x) + d(x,k_x)\\
		&\le& d(t,x)\cdot(2\alpha(k)+1)~.
	\end{eqnarray*}
	Also note that
	\[
	d_{G'}(f(t),f(x)) = d_G(f(t),f(k_x))+d(x,k_x) \ge d(t,k_x)+d(x,k_x)\ge d(t,x)~,
	\]
	so the terminal distortion is indeed $2\alpha(k)+1$. Since $f$ embeds into a graph with $\gamma(k)$ edges, and we  added $n-k$ new edges, the total number of edges is bounded accordingly, which concludes the proof.
\end{proof}

Next, we study strong terminal embeddings into normed spaces. Fix any metric $(X,d)$, a set of terminals $K\subseteq X$ and $1\le p\le \infty$. Let $f:K\to\ell_p$ be a non-expansive embedding. We say that $f$ is {\em Lipschitz extendable}, if there exists a non-expansive $\hat{f}:X\to\ell_p$ which is an extension of $f$ (that is, the restriction of $\hat{f}$ to $K$ is exactly $f$). It is not hard to verify that any Fr\'{e}chet embedding\footnote{In our context, it will be convenient to call an embedding $f:K\to\ell_p^t$ {\em Fr\'{e}chet}, if there are sets $A_1,\dots,A_t\subseteq X$ such that for all $i\in[t]$, and for every $x \in K$, we have $f_i(x)=\frac{d(x,A_i)}{t^{1/p}}$.} is Lipschitz extendable. For example, the embeddings of \cite{B85,KLMN04,ALN08} are Fr\'{e}chet.

\begin{theorem}\label{thm:lp-strong}
Let ${\cal X}$ be a family of metric spaces. Fix some $(X,d_X)\in{\cal X}$, and let $K\subseteq X$ be a set of terminals of size $|K|=k$, such that $(K,d_K)\in{\cal X}$. If any $(Z,d_Z)\in{\cal X}$ of size $|Z|=m$ embeds into $\ell_p^{\gamma(m)}$ with distortion $\alpha(m)$ by a {\em Lipschitz extendable map}, then there is a (non-expansive) embedding of $X$ into $\ell_p^{\gamma(n)+\gamma(k)+1}$ with {\em strong} terminal distortion $O(\alpha(k),\alpha(n))$.
\end{theorem}

\begin{proof}
	
	Let $(X,d)\in{\cal X}$ be a metric on $n$ points, $K\subseteq X$ of size $|K|=k$. There is a non-expansive embedding $g:X\to\ell_p^{\gamma(n)}$ with distortion at most $\alpha(n)$, and there exists a Lipschitz extendable embedding $f:K\to\ell_p^{\gamma(k)}$, which is non-expansive and has distortion $\alpha(k)$. Let $\hat{f}$ be the extension of $f$ to $X$, note that by definition of Lipschitz extendability, $\hat{f}$ is also non-expansive. Finally, let $h:X\to\R$ be defined by $h(x)=d(x,K)$. The embedding $F:X\to\ell_p^{\gamma(n)+\gamma(k)+1}$ is defined by the concatenation of these maps $F=g\oplus \hat{f}\oplus h$.
	
	Since all the three maps $g,\hat{f},h$ are non-expansive, it follows that for any $x,y\in X$,
	\[
	\|F(x)-F(y)\|_p^p\le \|g(x)-g(y)\|_p^p+\|\hat{f}(x)-\hat{f}(y)\|_p^p+|h(x)-h(y)|^p\le 3d(x,y)^p~,
	\]
	so $F$ has expansion at most $3^{1/p}$ for all pairs (which can easily be made 1 without affecting the distortion by more than a constant factor). Also note that
	\[
	\|F(x)-F(y)\|_p\ge \|g(x)-g(y)\|_p\ge\frac{d(x,y)}{\alpha(n)}~,
	\]
	which implies the distortion bound for all pairs is satisfied. It remains to bound the contraction for all pairs containing a terminal. Let $t\in K$ and $x\in X$, and let $k_x\in K$ be such that $d(x,K)=d(x,k_x)$ (it could be that $k_x=x$). If it is the case that $d(x,t)\le 3\alpha(k)\cdot d(x,k_x)$ then by the single coordinate of $h$ we get sufficient contribution for this pair:
	\[
	\|F(t)-F(x)\|_p\ge |h(t)-h(x)|=h(x)=d(x,k_x)\ge \frac{d(x,t)}{3\alpha(k)}~.
	\]
	The other case is that $d(x,t)> 3\alpha(k)\cdot d(x,k_x)$, here we will get the contribution from $\hat{f}$. First, observe that by the triangle inequality,
	\begin{equation}\label{eq:gkp}
	d(t,k_x)\ge d(t,x)-d(x,k_x) >  d(t,x)(1-1/(3\alpha(k)))\ge 2d(t,x)/3~.
	\end{equation}
	By another application of the triangle inequality, using $\hat{f}$ is non-expansive, and that $f$ has distortion $\alpha(k)$ on the terminals, we get the required bound on the contraction:
	\begin{eqnarray*}
		\|F(t)-F(x)\|_p&\ge&\|\hat{f}(t)-\hat{f}(x)\|_p\\
		&\ge&\|\hat{f}(t)-\hat{f}(k_x)\|_p-\|\hat{f}(k_x)-\hat{f}(x)\|_p\\
		&\ge&\|f(t)-f(k_x)\|_p-d(x,k_x)\\
		&\ge&\frac{d(t,k_x)}{\alpha(k)}-\frac{d(t,x)}{3\alpha(k)}\\
		&\stackrel{\eqref{eq:gkp}}{\ge}&\frac{2d(t,x)}{3\alpha(k)}-\frac{d(t,x)}{3\alpha(k)}\\
		&=&\frac{d(t,x)}{3\alpha(k)}~.
	\end{eqnarray*}

\end{proof}

\noindent{\bf Remark:} The results of \theoremref{thm:trans} and \theoremref{thm:lp-strong} hold also if ${\cal X}$ is a family of graphs, rather than of metrics, provided that the embedding for this family has the promised guarantees even for graphs with Steiner nodes. (E.g., if $\hat{Z}\in {\cal X}$ is a graph and $Z$ is a set of vertices of size $m$, then there exists a (Lipschitz extendable) embedding of $(Z,d_Z)$ to $\ell_p^{\gamma(m)}$ with distortion $\alpha(m)$, where $d_Z$ is the shortest path metric on $\hat{Z}$ induced on $Z$.) We note that many embeddings of graph families satisfy this condition, e.g. the embedding of \cite{KLMN04} to planar and minor-free graphs.\footnote{We remark that this requirement is needed for those graph families for which the following question is open: given a graph $Z$ in the family with terminals $K$, is there another graph in the family over the vertex set $K$, that preserves the shortest-path distances with respect to $Z$ (up to some constant). This question is open, e.g., for planar graphs.}

\begin{corollary}\label{cor:implications}
Let $(X,d)$ be a metric space on $n$ points, and $K\subseteq X$ a set of terminals of size $|K|=k$. Then for any $1\le p\le\infty$,

\begin{enumerate}

\item
\label{item:Bourgain}
$(X,d)$ can be embedded to $\ell_p^{O(\log k)}$ with terminal distortion $O(\log k)$.

\item
\label{item:JL}
If $(X,d)$ is an $\ell_2$ metric, it can be embedded to $\ell_2^{O(\log k)}$ with terminal distortion $O(1)$.

\item
\label{item:sp}
For any $t\ge 1$ there exists a $(4t-1)$-terminal (metric) spanner of $X$ with at most $O(k^{1+1/t})+n$ edges.

\item
\label{item:eucl_sp}
If $(X,d)$ is an $\ell_2$ metric, for any $t\ge 1$ there exists a $O(t)$-terminal spanner of $X$ with at most $O(k^{1+1/t^2})+n$ edges.

\item
\label{item:strong_Bourgain}
$(X,d)$ can be embedded to $\ell_p^{O(\log n+\log^2k)}$ with strong terminal distortion $(O(\log k),O(\log n))$.

\item
\label{item:planar}
If $(X,d)$ is a decomposable metric, \footnote{See \cite{KLMN04} for a definition of decomposability. We remark that doubling metrics, planar metric, and more generally metric arising from graphs excluding a fixed minor, are all decomposable.} it can be embedded to $\ell_p$ with strong terminal distortion $(O((\log k)^{\min\{\frac{1}{2},\frac{1}{p}\}}),O((\log n)^{\min\{\frac{1}{2},\frac{1}{p}\}}))$.

\item
\label{item:neg}
If $(X,d)$ is a negative type metric it can be embedded to $\ell_2$ with strong terminal distortion $(\tilde{O}(\sqrt{\log k}),\tilde{O}(\sqrt{\log n}))$. (We note that any $\ell_1$ metric is of negative type.)

\item
\label{item:infty}
For any $t\ge 1$, $(X,d)$ can be embedded to $\ell_\infty^{O(t\cdot k^{1/t}\cdot\log k)}$ with terminal distortion $O(t)$.

\end{enumerate}
\end{corollary}

The first two items and the last one use the first assertion of \theoremref{thm:trans}, the next two use its second assertion, and the next three apply \theoremref{thm:lp-strong}. The corollary follows from known embedding results: (\ref{item:Bourgain}) and (\ref{item:strong_Bourgain}) are from \cite{B85}, with improved dimension due to \cite{ABN06}, (\ref{item:JL}) is from \cite{JL84}, (\ref{item:sp}) is from \cite{ADDJS93} and (\ref{item:eucl_sp}) from \cite{HIS13}, (\ref{item:planar}) from \cite{KLMN04}, (\ref{item:neg}) from \cite{ALN07,ALN08}, and  (\ref{item:infty}) from \cite{M96,ABN06}.

\section{Graph Terminal Spanners}\label{sec:spanner-oracle}

While \theoremref{thm:trans} provides a general approach to obtain terminal spanners, it cannot provide spanners which are subgraphs of the input graph. We devise a construction of such terminal spanners in this section, while somewhat increasing the number of edges. Specifically, we show the following.
\begin{theorem}\label{thm:spanner2}
For any parameter $t\ge 1$, a graph $G=(V,E)$ on $n$ vertices, and a set of terminals $K\subseteq V$ of size $k$, there exists a $(4t-1)$-terminal {\em graph} spanner with at most $O(n+\sqrt{n}\cdot k^{1+1/t})$ edges.
\end{theorem}
\noindent {\bf Remark:} Note that the number of edges is linear in $n$ whenever $k\le n^{1/(2(1+1/t))}$.

We shall use the following result:
\begin{theorem}[\cite{CE05}]\label{thm:preserver}
Given a weighted graph $G=(V,E)$ on $n$ vertices and a set $P\subseteq{V\choose 2}$ of size $p$, then there exists a subgraph $G'$ with $O(n+\sqrt{n}\cdot p)$ edges, such that for all $\{u,v\}\in P$, $d_G(u,v)=d_{G'}(u,v)$.
\end{theorem}
\begin{proof}[Proof of \theoremref{thm:spanner2}]
The construction of the subgraph spanner with terminal distortion will be as follows. Consider the metric induced on the terminals $K$ by the shortest path metric on $G$. Create a $(2t-1)$ (metric) spanner $H'$ of this metric, using \cite{ADDJS93}, and let $P\subseteq{K\choose 2}$ be the set of edges of $H'$. Note that $p=|P|\le O(k^{1+1/t})$. Now, apply \theoremref{thm:preserver} on the graph $G$ with the set of pairs $P$, and obtain a graph $G'$ that for every $\{u,v\}\in P$, has $d_{G'}(u,v)=d_G(u,v)$. This implies that $G'$ is a $(2t-1)$-spanner for each pair of vertices $u,w\in K$. Moreover, $G'$ has at most $O(n+\sqrt{n}\cdot p)$ edges. Finally, create $H$ out of $G'$ by adding a shortest path tree in $G$ with the set $K$ as its root. This will guarantee that the spanner $H$ will have for each non-terminal, a shortest path to its closest terminal in $G$.  This concludes the construction of $H$, and now we turn to bounding the distortion. Since $H$ is a subgraph clearly it is non-contracting.
Fix any $v\in K$ and $u\in V$, let $k_u$ be the closest terminal to $u$, then $d_G(k_u,v)\le d_G(k_u,u)+d_G(u,v)\le 2d_G(u,v)$, and thus
\[
d_H(u,v)\le d_H(u,k_u)+d_{G'}(k_u,v)\le d_G(u,v)+(2t-1)d_G(k_u,v)\le (4t-1)d_G(u,v)~.
\]
Finally observe that the total number of edges in $H$ is at most $O(n+\sqrt{n}\cdot p)=O(n+\sqrt{n}\cdot k^{1+1/t})$.
\end{proof}

\section{Light Terminal Trees for General Graphs}\label{sub:Light-k-term-gen}

In this section we find a single spanning tree of a given graph, that has both light weight, and approximately preserves distances from a set of specified terminals. \theoremref{thm:trans} can provide a tree with terminal distortion $2k-1$ (using that any graph has a tree with distortion $n-1$), but that tree may not  be a subgraph and may have large weight.
The result of this section is summarized as follows.
\begin{theorem}
\label{thm:k-trm-general-grph}
For any parameter $\alpha\ge 1$, given a weighted graph $G=\left(V,E,w\right)$,
and a subset of terminals $K\subseteq V$
of size $k$, there exists a spanning tree $T$ of $G$ with terminal
distortion at most $k\cdot\alpha+\left(k-1\right)\alpha^{2}$ and lightness at most
$2\alpha+1+\frac{2}{\alpha-1}$.
\end{theorem}

When substituting $\alpha=1$ in \theoremref{thm:k-trm-general-grph}
we obtain a single tree with terminal distortion exactly $2k-1$, which is optimal (see \theoremref{thm:LowerBouTermstretchGrph}),
but no guarantee  for lightness.
 More generally, for small $\epsilon>0$, we get terminal distortion $2k-1+\epsilon$ and
lightness $3+\frac{6k}{\epsilon}$.
Also, note that the bound $2\alpha+1+\frac{2}{\alpha-1}$ is minimized by setting
$\alpha=2$, so there is no point in using the theorem with $\alpha>2$.

Next we describe the algorithm for constructing a spanning tree that satisfies the assertion of \theoremref{thm:k-trm-general-grph}.
We shall assume w.l.o.g that all
edge weights are different, and every two different paths have different
lengths. If it is not the case, then one can break ties in an arbitrary
(but consistent) way.

 A spanning tree $T$ is an $(\alpha,\beta)$-SLT with respect to a root $v\in V$, if for all $u\in V$, $d_T(v,u)\le\alpha\cdot d_G(v,u)$, and $T$ has lightness $\beta$.
A small modification of an SLT-constructing algorithm produces for any subset $K\subset V$,
a forest $F$, such that every component of $F$ contains exactly one vertex of $K$.\footnote{To obtain such a forest $F$, one should add a new vertex $v_{K}$
to the graph and connect it to each of the vertices of $K$ with edges
of weight zero. Then we compute an $\left(\alpha,\beta\right)$-SLT
with respect to $v_{K}$ in the modified graph. Finally, we remove
$v_{K}$ from the SLT. The resulting forest is $F$.
}
The forest $F$ has distortion $\alpha$ with respect to $K$, and {\em lightness}
$1+\frac{2}{\alpha-1}$. (Such a forest $F$ is said to have distortion
$\alpha$ with respect to $K$, if for every vertex $u\in V$, $d_{F}\left(K,u\right)\le\alpha\cdot d_{G}\left(K,u\right)$.)

The algorithm starts by building the aforementioned SLT-forest $F$
from the terminal set $K$.
No two terminals belong to the same connected component of $F$.
Denote $K=\{v_1,\dots,v_k\}$, let $V_{i}$ be the unique connected component of $F$ containing $v_{i}$, and let $T_{i}\subseteq F$
be the edges of the forest $F$ induced by $V_{i}$. It follows that for every
$u\in V_{i},\ d_{F}\left(K,u\right)=d_{T_{i}}\left(v_{i},u\right)\le\alpha\cdot d_{G}\left(K,u\right)$.
Let $G'=\left(K,E',w'\right)$ be the super-graph in which two terminals
share an edge between them if and only if there is an edge between
the components $V_{i}$ to $V_{j}$ in $G$. Formally, $\mbox{\ensuremath{E'=\left\{ \left\{v_{i},v_{j}\right\}:\,\exists u_{i}\in V_{i},u_{j}\in V_{j}\mbox{ such that }\left\{ u_{i},u_{j}\right\} \in E\right\} }}$.
The weight $w'\left(v_{i},v_{j}\right)$ is defined to be the length
of the shortest path between $v_{i}$ and $v_{j}$ which uses $exactly\ one\ edge$
that does not belong to $F$. (In other words, among all the paths
between $v_{i}$ and $v_{j}$ in $G$ which use exactly one edge that
does not belong to $F$, let $P$ be the shortest one. Then $w'\left(v_{i},v_{j}\right)=w\left(P\right)$.)
Note also that $w'\left(v_{i},v_{j}\right)$ is given by $\mbox{\ensuremath{w'\left(v_{i},v_{j}\right)=\min{}_{e\in E}\left\{ d_{F\cup\left\{ e\right\} }\left(v_{i},v_{j}\right)\right\} }}$.
We call the edge $e_{i,j}=\left\{ u_{i},u_{j}\right\} $ that implements
this minimum ($w'\left(v_{i},v_{j}\right)=d_{F\cup\left\{ e_{i,j}\right\} }\left(v_{i},v_{j}\right)$)
the $representative\ edge$ of $\left\{ v_{i},v_{j}\right\} $. (Recall that w.l.o.g the shortest paths, and thus the representative edges, are unique.) Observe that $\left\{ v_{i},v_{j}\right\} \in E'$
implies that $w'\left(v_{i},v_{j}\right)<\infty$. Let $T'$ be the
$MST$ of $G'$. Define $R=\left\{ e_{i,j}\ |e_{i,j}\mbox{ is the representative edge of }e'_{i,j}=\left(v_{i},v_{j}\right)\in T'\right\} $.
Finally, set $T=F\cup R=\bigcup_{i=1}^{k}T_{i}\cup R$. Obviously,
$T$ is a spanning tree of $G$. This concludes the construction, next we turn to the analysis.

As an embedding of a graph into its spanning
tree is non-contractive, the tree $T$ will have terminal distortion
$\alpha$ if for all $v\in K$, $u\in V$, $d_{T}\left(v,u\right)\le\alpha\cdot d_{G}\left(v,u\right)$.

The next lemma shows that for every pair of terminals $v_{i},v_{j}$,
there is a path between them in $G'$ in which all edges have weight
(with respect to $w'$) at most $\alpha\cdot d_{G}\left(v_{i},v_{j}\right)$.
\begin{lemma}
	\label{lem:Bottle neck in G'}[The bottleneck lemma:] For every
	$v_{i},v_{j}\in K$, there exists a path $P:v_{i}=z_{0},z_{1},...,z_{r}=v_{j}$
	in $G'$ such that for every $s=0,1,\dots,r-1$, it holds that $\left\{z_{s,}z_{s+1}\right\}\in E'$
	and $w'\left(z_{s},z_{s+1}\right)\leq\alpha\cdot d_{G}\left(v_{i},v_{j}\right)$.\end{lemma}
\begin{proof}
	Let $P_{i,j}:v_{i}=u_{0},u_{1},...,u_{s}=v_{j}$ be the shortest path
	from $v_{i}$ to $v_{j}$ in $G$, i.e., $w(P_{i,j})=d_{G}\left(v_{i},v_{j}\right)$.
	For each $0\le a\le s$, denote by $V^{(a)}$ the connected component of $F$ that contains $u_a$, and let $v^{(a)}$ be the unique terminal in that component. Consider the path $P=v^{\left(0\right)},v^{\left(1\right)},...,v^{\left(s\right)}$.
	(This path is not necessarily simple. In particular, it might contain
	self-loops. See \figureref{fig:bottleneck} for an illustration.) For every index $a<s$,
	\begin{eqnarray*}
		w'\left(v^{\left(a\right)},v^{\left(a+1\right)}\right) & \underset{\left(1\right)}{\le} & d_{F\cup\left\{ \{u_{a},u_{a+1}\}\right\} }\left(v^{\left(a\right)},v^{\left(a+1\right)}\right)\\
		 &=&d_{F}\left(v^{\left(a\right)},u_{a}\right)+d_{G}\left(u_{a},u_{a+1}\right)+d_{F}\left(u_{a+1},v^{\left(a+1\right)}\right)\\
		& \underset{\left(2\right)}{\le} & \alpha\cdot d_{G}\left(v_{i},u_{a}\right)+d_{G}\left(u_{a},u_{a+1}\right)+\alpha\cdot d_{G}\left(v_{j},u_{a+1}\right)\\
		& < & \alpha\cdot\left(d_{G}\left(v_{i},u_{a}\right)+d_{G}\left(u_{a},u_{a+1}\right)+d_{G}\left(v_{j},u_{a+1}\right)\right)\\
		&\underset{\left(3\right)}{=}&\alpha\cdot d_{G}\left(v_{i},v_{j}\right).
	\end{eqnarray*}
	Note that if for some index $a$ it holds that $v^{\left(a\right)}=v^{\left(a+1\right)}$
	then $w'\left(v^{\left(a\right)},v^{\left(a+1\right)}\right)=0$,
	and the inequality above holds trivially. Otherwise, if $v^{\left(a\right)}\ne v^{\left(a+1\right)}$, then
	inequality $\left(1\right)$ follows from the assumptions that $\left\{ u_{a},u_{a+1}\right\} \in E$,
	$u_{a}\in V^{\left(a\right)}$, $u_{a+1}\in V^{\left(a+1\right)}$.
	Inequality $\left(2\right)$ follows from the properties of the $SLT$
	tree $T$ (as $d_{F}\left(v^{\left(a\right)},u_{a}\right)=d_{F}\left(K,u_{a}\right)\le\alpha\cdot d_{G}\left(K,u_{a}\right)\le\alpha\cdot d_{G}\left(v_{i},u_{a}\right)$).
	Equality $\left(3\right)$ follows because the edge $\left\{ u_{a},u_{a+1}\right\} $
	is on the shortest path from $v_{i}$ to $v_{j}$ in $G$.
	
	In particular, one can remove cycles from $P$ and obtain a simple
	path with the desired properties. We get a simple path $P'$
	such that for every edge $v,v'$ on this path, we have $w'\left(v,v'\right)\le\alpha\cdot d_{G}\left(v_{i},v_{j}\right)$,
	as required.
\end{proof}

The following is a simple corollary.

\begin{corollary}
	\label{cor:imlemented edge alpha approx}For $\left\{ v_{i},v_{j}\right\} \in T'$,
	we have $w'\left(v_{i},v_{j}\right)=d_{T}\left(v_{i},v_{j}\right)\leq\alpha\cdot d_{G}\left(v_{i},v_{j}\right)$.
\end{corollary}
\begin{proof}
	By \lemmaref{lem:Bottle neck in G'}, $w'\left(v_{i},v_{j}\right)\leq\alpha\cdot d_{G}\left(v_{i},v_{j}\right)$.
	(Indeed, otherwise the edge $\left\{ v_{i},v_{j}\right\} $ is strictly
	the heaviest edge in a cycle in $G'$, contradiction to the assumption
	that it belongs to the MST of $G'$.) Since $\left\{ v_{i},v_{j}\right\} \in E'$
	and the representative edge of $\left\{ v_{i},v_{j}\right\} $ was
	taken into $T$, it follows that $w'\left(v_{i},v_{j}\right)=d_{T}\left(v_{i},v_{j}\right)$.
\end{proof}

We conclude the following lemma, which bounds the distortion of terminal pairs.

\begin{lemma}\label{Lemma: terminal bottle neck}
	For $v_{i},v_{j}\in K$, we have
	$d_{T}\left(v_{i},v_{j}\right)\le d_{T'}\left(v_{i},v_{j}\right)\le\alpha\cdot\left(k-1\right)\cdot d_{G}\left(v_{i},v_{j}\right)$.
\end{lemma}

\begin{figure}
	\begin{centering}
		\includegraphics[scale=0.8]{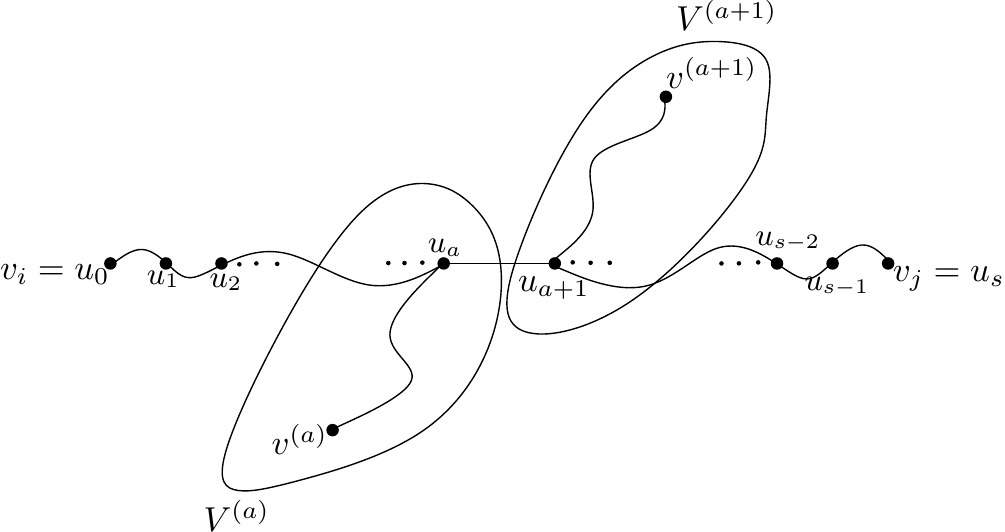}
		\par\end{centering}
	
	\caption{\label{fig:bottleneck}An illustration for the bottleneck lemma:{\small{}
			$v_{i}$ and $v_{j}$ are terminals. The edge $\left\{ u_{a},u_{a+1}\right\} $
			belongs to the shortest path from $v_{i}$ to $v_{j}$ in $G$. We
			conclude that for terminals $v^{\left(a\right)},v^{\left(a+1\right)}$
			such that $u_{a}\in V^{\left(a\right)}$ and $u_{a+1}\in V^{\left(a+1\right)}$,
			it holds that }$w'\left(v^{\left(a\right)},v^{\left(a+1\right)}\right)\leq\alpha\cdot d_{G}\left(v_{i},v_{j}\right)$.}
\end{figure}
\begin{figure}
	\begin{centering}
		\includegraphics[scale=0.8]{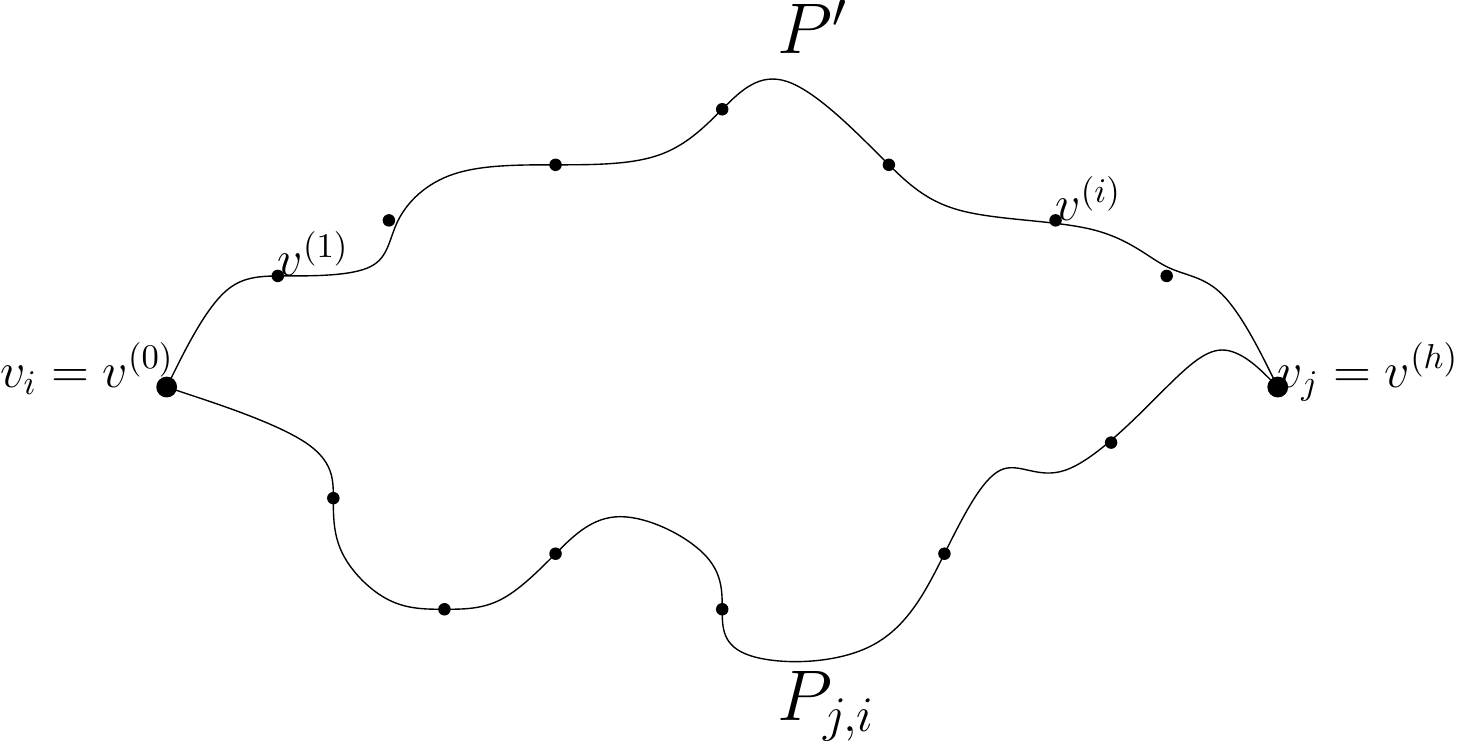}
		\par\end{centering}
	
	\caption{\label{fig:terminal bottle neck}The two paths $P'$ and $P_{j,i}$
		considered in the proof of \lemmaref{Lemma: terminal bottle neck}.
		The path $P'$ is contained in $T'$, while all edges of $P_{j,i}$
		have weight at most $\alpha\cdot d_{G}\left(v_{i},v_{j}\right)$. }
\end{figure}

\begin{proof}
	Let $P':v_{i}=v^{\left(0\right)},v^{\left(1\right)},\dots,v^{\left(h\right)}=v_{j}$
	be the (unique) path in $T'$ between $v_{i}$ and $v_{j}$. Since
	$T'$ is a spanning tree of the $k$-vertex graph $G'$, it follows
	that $h\le k-1$. Observe also that for every index $a\in\left[h-1\right]$,
	by \corollaryref{cor:imlemented edge alpha approx}, the edge $w'\left(v^{\left(a\right)},v^{\left(a+1\right)}\right)=d_{T}\left(v^{\left(a\right)},v^{\left(a+1\right)}\right)$.
	Also, we next argue that $w'\left(v^{\left(a\right)},v^{\left(a+1\right)}\right)\le\alpha\cdot d_{G}\left(v_{i},v_{j}\right)$.
	Indeed, suppose for contradiction that $w'\left(v^{\left(a\right)},v^{\left(a+1\right)}\right)>\alpha\cdot d_{G}\left(v_{i},v_{j}\right)$.
	Let $P_{j,i}$ be a path between $v_{j}$ and $v_{i}$ in $G'$ such
	that all its edges have weight at most $\alpha\cdot d_{G}\left(v_{i},v_{j}\right)$.
	The existence of this path is guaranteed by \lemmaref{lem:Bottle neck in G'}.
	In particular, since $w'\left(v^{\left(a\right)},v^{\left(a+1\right)}\right)>\alpha\cdot d_{G}\left(v_{i},v_{j}\right)$,
	it follows that $\left\{ v^{\left(a\right)},v^{\left(a+1\right)}\right\} \notin P_{j,i}$.
	Consider the cycle $P'\circ P_{j,i}$ in $G'$. It is not necessarily
	a simple cycle, but since $\left\{ v^{\left(a\right)},v^{\left(a+1\right)}\right\} \notin P_{j,i}$,
	the edge $\left\{ v^{\left(a\right)},v^{\left(a+1\right)}\right\} $
	belongs to a simple cycle $C$ contained in $P'\circ P_{j,i}$. The
	heaviest edge of $C$ clearly does not belong to $P_{j,i}$, because
	the edge $\left\{ v^{\left(a\right)},v^{\left(a+1\right)}\right\} $
	is heavier than each of them. Hence the heaviest edge belongs to $P'$,
	but $P'\subseteq T'$. This is a contradiction to the assumption that
	$T'$ is an MST of $G'$. (See \figureref{fig:terminal bottle neck}
	for an illustration). Hence $d_{T}\left(v^{\left(a\right)},v^{\left(a+1\right)}\right)=w'\left(v^{\left(a\right)},v^{\left(a+1\right)}\right)\le\alpha\cdot d_{G}\left(v_{i},v_{j}\right)$.
	Finally,
	\begin{eqnarray*}
		d_{T}\left(v_{i},v_{j}\right) & \le & \sum_{a=0}^{h-1}d_{T}\left(v^{\left(a\right)},v^{\left(a+1\right)}\right)=\sum_{a=0}^{h-1}w'\left(v^{\left(a\right)},v^{\left(a+1\right)}\right)\\
		& \le & \sum_{a=0}^{h-1}\alpha\cdot d_{G}\left(v_{i},v_{j}\right)\le h\cdot\alpha\cdot d_{G}\left(v_{i},v_{j}\right)\le\alpha\cdot\left(k-1\right)\cdot d_{G}\left(v_{i},v_{j}\right).
	\end{eqnarray*}
	
\end{proof}

Next, we analyze the terminal distortion of $T$.
\begin{lemma}
	\label{lem:Gen Grp distortion Bound} The terminal distortion of $T$ is at most $k\cdot\alpha+\left(k-1\right)\alpha^{2}$.
\end{lemma}
\begin{proof}
	For each terminal $v_{i}\in K$ and any vertex $u\in V_{j}$, for some $j \in \{1,2,\ldots,k\}$, it holds that
	\begin{eqnarray*}
		d_{T}\left(v_{i},u\right) & \leq & d_{T}\left(v_{i},v_{j}\right)+d_{T}\left(v_{j},u\right)\le\alpha\cdot\left(k-1\right)\cdot d_{G}\left(v_{i},v_{j}\right)+\alpha\cdot d_{G}\left(v_{i},u\right)\\
		& \le & \alpha\cdot\left(k-1\right)\cdot\left(d_{G}\left(v_{i},u\right)+d_{G}\left(u,v_{j}\right)\right)+\alpha\cdot d_{G}\left(v_{i},u\right)\\
		& \le & \alpha\cdot\left(k-1\right)\cdot\left(d_{G}\left(v_{i},u\right)+\alpha\cdot d_{G}\left(v_{i},u\right)\right)+\alpha\cdot d_{G}\left(v_{i},u\right)\\
		&=&\left(k\cdot\alpha+\left(k-1\right)\alpha^{2}\right)\cdot d_{G}\left(v_{i},u\right).
	\end{eqnarray*}
	The last inequality is because $d_{G}\left(v_{j},u\right)\le d_{F}\left(v_{j},u\right)=d_{F}\left(K,u\right)\le\alpha\cdot d_{G}\left(K,u\right)\le\alpha\cdot d_{G}\left(v_{i},u\right)$.

\end{proof}
Next, we analyze the lightness of $T$. A tree $T=\left(K,E',w'\right)$ is called a $Steiner\mbox{ }tree$
for a graph $G=\left(V,E,w\right)$ if $\left(1\right)$ $V\subseteq K$,
$\left(2\right)$ for any edge $e\in E\cap E'$, the edge has the
same weight in both $G$ and $T$, i.e. $w\left(e\right)=w'\left(e\right)$,
and $\left(3\right)$ for any pair of vertices $u,v\in V$ it holds
that $d_{T}\left(u,v\right)\ge d_{G}\left(u,v\right)$. The $minimum\ Steiner\ tree$
$T$ of $G$, denoted $SMT\left(G\right)$, is a Steiner tree of $G$
with minimum weight. It is well-known that for any graph $G$, $w\left(SMT\left(G\right)\right)\ge\frac{1}{2}MST\left(G\right)$.
(See, e.g., \cite{GilbertP68}, Section 10.) The next
lemma bounds the lightness of the tree $T$.
\begin{lemma}
	\label{lem:Gen Grp lighness Bound}The $lightness$ of $T$ is bounded
	by $\Psi\left(T\right)\le2\alpha+1+\frac{2}{\alpha-1}$.\end{lemma}
\begin{proof}
	The main challenge is to bound $w\left(R\right)$. (Recall that $R$
	is the set of the $representative\ edges$ of $T'$.) Consider an
	edge $\left\{ v_{i},v_{j}\right\} \in T'$, and let $\left\{ u_{i},u_{j}\right\} $
	be its representative edge. Then $d_{G}\left(u_{i},u_{j}\right)\le w'\left(v_{i},v_{j}\right)$.
	Also, since $\left\{ v_{i},v_{j}\right\} \in T'\subseteq E'$, it
	follows that $w'\left(v_{i},v_{j}\right)=d_{G'}\left(v_{i},v_{j}\right)$.
	Hence $d_{G}\left(u_{i},u_{j}\right)\le d_{G'}\left(v_{i},v_{j}\right)$.
	Therefore, $w\left(R\right)\le w'\left(T'\right)$. Next we provide
	an upper bound for $w'\left(T'\right)$. Define the graph $\widetilde{G}$ as the complete graph on the vertex set $K$, with weights $\tilde{w}$ induced by $d_G$ (the shortest path distances in $G$). Let $\widetilde{T}$ be the $MST$ of $\widetilde{G}$. We
	build a new tree $\hat{T}$ by the following process:
	
	\begin{enumerate}
		\item Let $\hat{T}\leftarrow\widetilde{T}$;
		\item For each $\left\{ v_{i},v_{j}\right\} =\tilde{e}\in\tilde{T}$ :
		\begin{enumerate}
			\item Let $P_{\tilde{e}}$ be a path from $v_{i}$ to $v_{j}$ which consists
			of edges in $E'$, such that for each edge $e$ in $P_{\tilde{e}}$,
			$w'\left(e\right)\le\alpha\cdot d_{G}\left(v_{i},v_{j}\right)=\alpha\cdot\tilde{w}\left(\tilde{e}\right)$;
			(By \lemmaref{lem:Bottle neck in G'}, such a path exists);
			\item Let $e'\in P_{\tilde{e}}$ be an edge such that $(\hat{T}\setminus\left\{ \tilde{e}\right\} )\cup\left\{ e'\right\}$ is connected;
			\item Set $\hat{T}\leftarrow(\hat{T}\setminus\left\{ \tilde{e}\right\})\cup\left\{ e'\right\} $;
		\end{enumerate}
	\end{enumerate}

	In each step in the loop we replace an edge $\tilde{e}=\left\{ v_{i},v_{j}\right\} $
	from $\widetilde{T}$ by an edge $e'$ from $G'$ of weight $w'\left(e\right)\le\alpha\cdot\tilde{w}\left(\tilde{e}\right)$.
	Hence, the resulting tree $\hat{T}$ is a spanning tree of $G'$, and
	$w'\left(\hat{T}\right)\le\alpha\cdot\tilde{w}\left(\widetilde{T}\right)$.
	Since $T'$ is the $MST$ of $G'$, it follows that $w'\left(T'\right)\le w'\left(\hat{T}\right)$.
	The $MST$ of $G$ is a Steiner tree for $\tilde{G}$, so that $\tilde{w}\left(SMT\left(\tilde{G}\right)\right)\le w\left(MST\left(G\right)\right)$.
	Also, $\tilde{w}\left(MST\left(\tilde{G}\right)\right)=\tilde{w}\left(\tilde{T}\right)\le2\cdot\tilde{w}\left(SMT\left(\tilde{G}\right)\right)\le2\cdot w\left(MST\left(G\right)\right)$.
	We obtain that
	\begin{eqnarray*}
		w\left(R\right) & \le & w'\left(T'\right)\le w'\left(\hat{T}\right)\le\alpha\cdot\tilde{w}\left(\widetilde{T}\right)\le2\cdot\alpha\cdot w\left(MST\left(G\right)\right).
	\end{eqnarray*}

	Since $w\left(F\right)\le\left(1+\frac{2}{\alpha-1}\right)\cdot w\left(MST\left(G\right)\right)$,
	we conclude that
	\begin{eqnarray*}
		w\left(T\right) & = & w\left(R\cup F\right)=w\left(R\right)+w\left(F\right)\le\left(2\alpha+1+\frac{2}{\alpha-1}\right)\cdot w\left(MST\left(G\right)\right).
	\end{eqnarray*}
	
\end{proof}

\section{Lower Bounds on Distortion-Lightness Tradeoffs for a Single Tree}\label{sec:Lower-Bounds}

In this section we provide lower bounds on the terminal distortion and
lightness of $k$-terminal trees. For each lower bound we start with
showing a lower bound for graphs and proceed to showing a lower bound
for metric spaces. (Observe that the latter is more general.) The
lower bounds exhibit similar tradeoffs in both cases, while the analysis is significantly
simpler for graphs than for metric spaces.

\subsection{A lower bound for terminal distortion}

\begin{theorem}
	\label{thm:LowerBouTermstretchGrph}For any $k$ there is a weighted
	graph $G$ with $n=2k$ vertices and $k$ terminals such that any
	spanning tree has terminal distortion at least $2k-1$.\end{theorem}
\begin{proof}
	Consider the cycle graph $C_{2k}$ such that the terminals and the
	non-terminal vertices alternate (as usual $V'$ stands for the set
	of terminals). Any spanning tree $T$ of $C_{2k}$ is obtained by
	removing a single edge $\left\{ v_{i},v_{i+1}\right\} $. Observe
	that either $v_{i}$ or $v_{i+1}$ is a terminal. Hence the terminal distortion is at least
	\[
	\frac{d_{C_{2k}\setminus\left\{ v_{i},v_{i+1}\right\} }\left(v_{i},v_{i+1}\right)}{d_{C_{2k}}\left(v_{i},v_{i+1}\right)}=2k-1.\]\end{proof}
\begin{remark}
	The same result for any $n>2k$ follows by adding additional $n-2k$
	non-terminal vertices to the graph, and connecting them all to an
	arbitrary vertex of $C_{2k}$. In the resulting graph every spanning
	tree has terminal distortion at least $2k-1$.
\end{remark}
Next we extend \theoremref{thm:LowerBouTermstretchGrph} to metric
spaces.
\begin{theorem}
	\label{thmLwrBndTrmStchMet}For any $k$ there is a metric space
	$\left(M,d\right)$ with $n=2k$ vertices and $k$ terminals such
	that any spanning tree for $\left(M,d\right)$ has terminal distortion
	at least $2k-1$.\end{theorem}
\begin{proof}
	Let $M$ be the metric generated by the cycle graph $C_{2k}$, where
	there are $k$ terminals and the terminals and the non-terminal vertices
	alternate. In \cite{DBLP:conf/soda/Gupta01}, Lemma 7.1, Gupta showed
	that the distortion of every spanning tree of $M$ is at least $n-1=2k-1$.
	Moreover, the maximal distortion is achieved by an original edge $e$
	of $C_{2k}$. One of the two endpoints of $e$ is a terminal. Hence
	the terminal distortion of any spanning tree for $M$ is at least $2k-1$.\end{proof}
\begin{remark}
	One can extend \theoremref{thmLwrBndTrmStchMet} to $n>2k$ by adding
	additional $n-2k$ non-terminal vertices to $M$ at a large distance
	from all the terminals. For a spanning tree $T$ of $M$, if some
	shortest path between a terminal $v$ and a vertex $u$ which is belongs
	to $C_{2k}$ uses at least one of the new vertices then the terminal
	distortion is too high. Hence if the terminal distortion is small, the tree
	restricted to $C_{2k}$ is connected. Hence by \theoremref{thmLwrBndTrmStchMet}
	it follows that the terminal distortion is at least $2k-1$.
\end{remark}

\subsection{A lower bound on the lightness in graphs}\label{sub:lwr bnd grph}
\label{sec:lightness}

In this section we prove a lower bound on the tradeoff between terminal
distortion and lightness of $k$-terminal trees.
\begin{theorem}
	\label{thm:Lwr Bnd Wght Grph}For any positive integer parameters
	$k$, $n$ and $\epsilon>0$ such that $k\le\frac{\epsilon}{2}n$,
	there exists a graph $G$ with $N=\left(n+1\right)k$ vertices, such
	that any spanning tree $T$ for $G$ with terminal distortion at most
	$\left(2k-1\right)\left(1+\frac{\epsilon}{k^{2}}\right)$ has lightness
	at least $\Omega\left(\frac{1}{\epsilon}\right)$.\end{theorem}
\begin{proof}
	Consider the following graph $G$ on $N=k\cdot\left(n+1\right)$ vertices.
	There are $k$ terminals $V'=\left\{ v_{0},v_{1},...,v_{k-1}\right\} $
	in $G$. For every index $i\in\left[0,k-1\right]$, there is a $n$-vertex
	path $P_{i}$. All edges in these paths have unit weight. Also, for
	each index $i\in\left[0,k-1\right]$, both $v_{i}$ and $v_{i+1\left(\mbox{mod }k\right)}$
	are connected to every vertex in $P_{i}$ by edges of weight $w$,
	where $w>1$ is a parameter that will be determined later. (To simplify
	the notation we will henceforward write $v_{i+1}$ instead $v_{i+1\left(\mbox{mod }k\right)}$.
	Generally, all the arithmetic operations on indexes of vertices $v_{0},...,v_{k-1}$
	and paths $P_{0},...,P_{k-1}$ are performed modulo $k$.) See
	\figureref{fig:LowBndGrph} for an illustration.
	
	\begin{figure}[h]
		\begin{centering}
			\includegraphics[width=343pt]{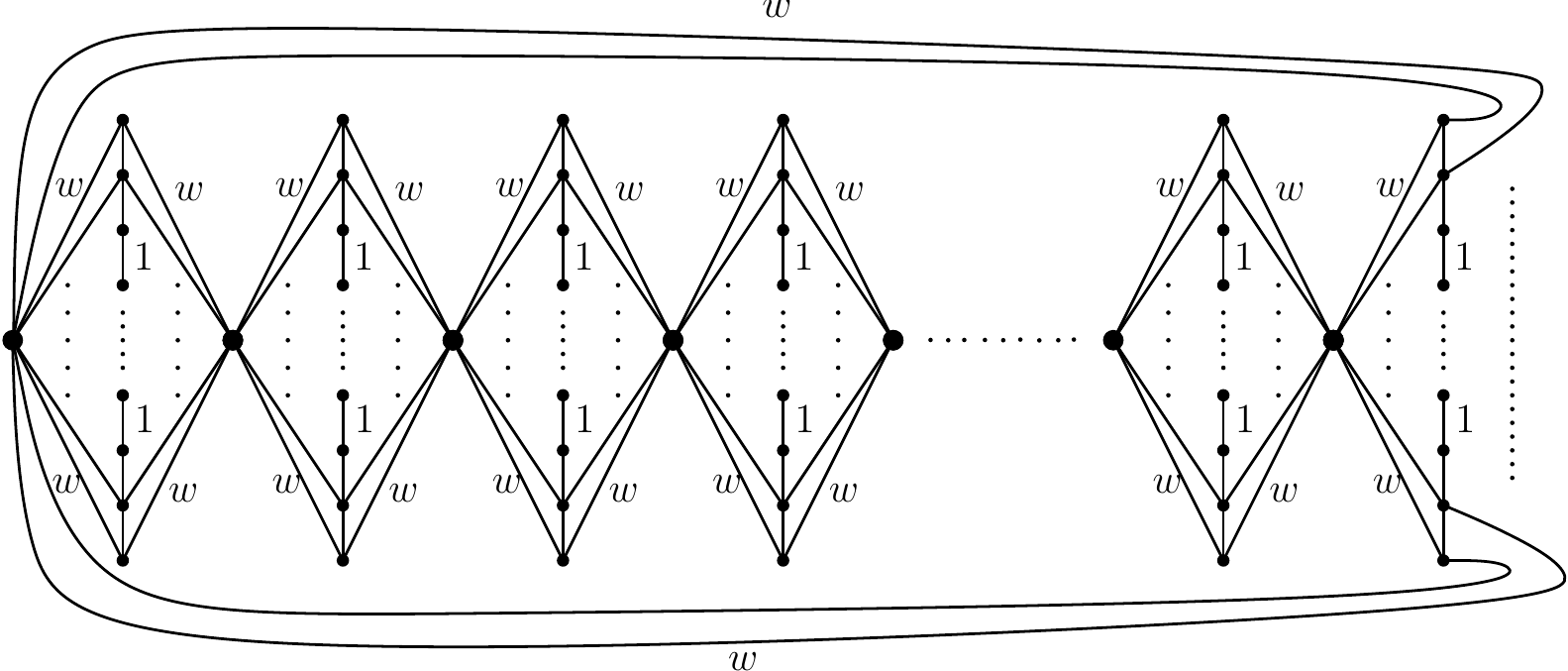}
			\par\end{centering}
		
		\caption{\label{fig:LowBndGrph}An illustration of the graph used in the proof
			of \theoremref{thm:Lwr Bnd Wght Grph}. The $k$ terminals are depicted
			by the big dots. The vertices $v_{i}$ and $v_{i+1}$ are connected
			to each vertex of an $n$-vertex path $P_{i}$ by edges of weight
			$w$.}
	\end{figure}

	Each spanning tree of $G$ contains $N-1=k+kn-1$ edges. There are
	$k\cdot\left(n-1\right)$ edges of unit weight, and all the other
	edges have weight $w$. Hence the weight of the MST is at least $k\left(n-1\right)\cdot1+\left(2k-1\right)\cdot w$.
	It is easy to verify that there actually exists a spanning tree of
	that weight.
	We will show that for any $\beta<\frac{1}{w(2k-1)}$, every tree with terminal distortion at most $\left(2k-1\right)\left(1+\beta\right)$ has weight at least $\Omega(nw)$.

	Let $T$ be a spanning tree for $G$ with terminal distortion at most
	$\left(2k-1\right)\left(1+\beta\right)$, for some $\beta>0$. There
	exists an index $i$ such that the path between $v_{i}$ and $v_{i+1}$
	in $T$ does not use vertices from the path $P_{i}$. (Otherwise there
	is a cycle in $T$ that passes through $v_{0},\dots,v_{k-1}$.) Without
	loss of generality assume that $i=0$. Therefore, $d_{T}\left(v_{0},v_{1}\right)\ge\left(k-1\right)\cdot2w$.
	\begin{claim}
		\label{Clm: Grp Lwe Bwnd must edge}For every vertex $u$ in $P_{0}$,
		if $\beta<\frac{1}{\left(2k-1\right)w}$ then either $\left(v_{0},u\right)$
		or $\left(v_{1},u\right)$ is an edge of $T$.\end{claim}
	\begin{proof}
		Without loss of generality the shortest path from $v_{0}$ to $u$
		goes through  $v_{1}$. Assume for contradiction that 	the edge $\left(v_{1},u\right)$
		does not belong to $T$. Then the terminal distortion is at least
		\begin{eqnarray*}
			\frac{d_{T}\left(v_{0},u\right)}{d_{G}\left(v_{0},u\right)} & = & \frac{d_{T}\left(v_{0},v_{1}\right)+d_{T}\left(v_{1},u\right)}{d_{G}\left(v_{0},u\right)}\ge\frac{\left(k-1\right)2w+w+1}{w}\\
			& = & 2k-1+\frac{1}{w}=\left(2k-1\right)\left(1+\frac{1}{\left(2k-1\right)w}\right)\\
			&>&\left(2k-1\right)\left(1+\beta\right),
		\end{eqnarray*}
		contradiction.
	\end{proof}
	By \claimref{Clm: Grp Lwe Bwnd must edge}, $\beta<\frac{1}{\left(2k-1\right)w}$
	implies that $w\left(T\right)\ge n\cdot w$. Hence for every $k$-terminal
	tree $T$ with terminal distortion at most $\left(2k-1\right)\left(1+\beta\right)$,
	with $\beta<\frac{1}{\left(2k-1\right)w}$, it holds that $\Psi\left(T\right)=\frac{w\left(T\right)}{w\left(MST\right)}\ge\frac{nw}{k(n-1)+\left(2k-1\right)w}$.
	We set $w=\frac{k}{\epsilon}$. Then the condition $\beta<\frac{1}{\left(2k-1\right)w}$
	translates to $\beta<\frac{\epsilon}{\left(2k-1\right)k}$. This condition
	implies that
	\[
	\Psi\left(T\right)=\frac{n\frac{k}{\epsilon}}{k(n-1)+\left(2k-1\right)\frac{k}{\epsilon}}=\frac{nk}{\epsilon k(n-1)+\left(2k-1\right)k}\ge\frac{nk}{\epsilon nk+2k^{2}}=\frac{n}{\epsilon n+2k}.
	\]
	As $k\le\frac{\epsilon}{2}n$ we obtain $\Psi\left(T\right)\ge\frac{1}{2\epsilon}$.
\end{proof}
Our algorithm from \theoremref{thm:k-trm-general-grph} guarantees
terminal distortion $\left(2k-1\right)\left(1+O\left(\epsilon\right)\right)$
and lightness $O\left(\frac{1}{\epsilon}\right)$. In the graph $G$
from the above proof lightness smaller than $\frac{1}{\epsilon}$
implies terminal distortion at least $\left(2k-1\right)\left(1+\frac{\epsilon}{\left(2k-1\right)k}\right)$.
Hence our bounds are tight for $k=O\left(1\right)$, but generally
there is a gap of $O\left(k^{2}\right)$ between the upper and lower
bounds.

In \appendixref{sub:lwr bnd metric} we extend this lower bound to metric spaces.

\def\APPENDMETRICLB
{
\section{A lower bound on the tradeoff between lightness and stretch for terminal trees in metric
	spaces}
\label{sub:lwr bnd metric}

In this section we extend our lower bound from \sectionref{sec:lightness} to the metric case.

The metric space case is similar to the graph case. We will use the
metric closure of the graph from the previous section. The main difficulty
is however to show that the non-graph edges do not help at all.

For a positive integer parameter $k$ and a $k-$sequence $n_{0},n_{1},\dots,n_{k-1}$
of positive integer numbers, we define a graph $G_{k,n_{0},n_{1},\dots,n_{k-1}}$.
The graph has $n=k+\sum_{i}n_{i}$ vertices and $k$ terminals. The
$k$ terminals are $V'=\left\{ v_{0},\dots,v_{k-1}\right\} $. For
every index $i\in\left[0,k-1\right]$ there is an $n_{i}$-vertex
path $P_{i}$. (We will use $P_{i}$ to denote both the path, and
the set of vertices in the path.) All edges in these paths have unit
weight. Also, for each $i\in\left[0,k-1\right]$, both $v_{i}$ and
$v_{i+1}$ (the index arithmetic is modulo $k$) are connected to
every vertex in $P_{i}$ by edges of weight $w$, for a parameter
$w>1$. Observe that the graph $G$ from \sectionref{sub:lwr bnd grph}
satisfies $G=G_{k,n_{0},n_{1},\dots,n_{k-1}}$ with $n_{0}=n_{1}=\cdots=n_{k-1}=n$.
Let $\overline{G}_{k,n_{0},\dots,n_{k-1}}$ denote the metric closure
of $G_{k,n_{0},\dots,n_{k-1}}$. We also write $G=G_{k,n_{0},\dots,n_{k-1}}$
and $\overline{G}=\overline{G}_{k,n_{0},\dots,n_{k-1}}$.
\begin{lemma}
	\label{lem: Gknnn}For any spanning tree $T$ of $\overline{G}_{k,n_{0},\dots,n_{k-1}}$
	there exists an index $i$ such that for any vertex $u\in P_{i}$
	either $\frac{d_{T}\left(v_{i},u\right)}{d_{\overline{G}_{k,n_{0},\dots,n_{k-1}}}\left(v_{i},u\right)}\ge2k-1$
	or $\frac{d_{T}\left(v_{i+1},u\right)}{d_{\overline{G}_{k,n_{0},\dots,n_{k-1}}}\left(v_{i},u\right)}\ge2k-1$. \end{lemma}
\begin{remark}
	Observe that for a graph spanning tree $T$ this lemma follows directly
	from the observation that there exists an index $i$ such that the
	path in $T$ from $v_{i}$ to $v_{i+1}$ does not contain vertices
	of $P_{i}$. Indeed, for this index $i$ and a vertex $u\in P_{i}$,
	either $\left(v_{i},u\right)\notin T$ or $\left(v_{i+1},u\right)\notin T$.
	In the former case $\frac{d_{T}\left(v_{i},u\right)}{d_{G_{k,n_{0},\dots,n_{k-1}}}\left(v_{i},u\right)}\ge2k-1$
	and in the latter $\frac{d_{T}\left(v_{i+1},u\right)}{d_{G_{k,n_{0},\dots,n_{k-1}}}\left(v_{i+1},u\right)}\ge2k-1$.
	The lemma proves this statement in a much greater generality, specifically,
	for $T$ being a spanning tree of the metric closure $\overline{G}_{k,n_{0},\dots,n_{k-1}}$
	of $G_{k,n_{0},\dots,n_{k-1}}$.\end{remark}
\begin{proof}
	\sloppy For a spanning tree $T$ of $\overline{G}_{k,n_{0},\dots,n_{k-1}}$
	and an index $i\in\left[0,k-1\right]$, denote by $t_{T,i}$ the number
	of vertices $u$ in $P_{i}$ such that $\frac{d_{T}\left(v_{i},u\right)}{d_{\overline{G}_{k,n_{0},\dots,n_{k-1}}}\left(v_{i},u\right)}\ge2k-1$
	or $\frac{d_{T}\left(v_{i+1},u\right)}{d_{\overline{G}_{k,n_{0},\dots,n_{k-1}}}\left(v_{i+1},u\right)}\ge2k-1$.
	We will show that for any spanning tree $T$ there exists an index
	$i$ such that $t_{T,i}=n_{i}$. For a tree $T$, define $\gamma\left(T\right)=\mbox{min}_{i}\left\{ n_{i}-t_{T,i}\right\} $.
	Observe that $\gamma\left(T\right)\ge0$. It suffices to prove that
	for every tree $T$, $\gamma\left(T\right)=0$. Also, let $\mu=\mbox{max}_{T}\left\{ \gamma\left(T\right)\right\} $,
	where the maximum is taken over all spanning trees of $\overline{G}_{k,n_{0},\dots,n_{k-1}}$.
	It is enough to show that $\mu=0$, i.e., that for every spanning
	tree $T$, $\gamma\left(T\right)=0$.
	
	For each vertex, we define the right and the left \emph{hemisphere}
	with respect to this vertex. Consider a supergraph, where we replace
	each path $P_{i}$ by a supernode $p_{i}$. We obtain the $2k-$cycle
	$C_{2k}$. The \emph{right hemisphere} of $v_{i}$ consists of all
	the vertices between $v_{i}$ and its antipodal vertex (in the path
	$p_{i},v_{i+1},p_{i+1},\dots$), while the \emph{left hemisphere}
	of $v_{i}$ consists of all the vertices in the other shortest path
	from $v_{i}$ to its antipodal vertex ($p_{i-1},v_{i-1},p_{i-2},\dots$).
	Formally, for a terminal $v_{i}$, if $k$ is even, let $R(v_{i})=\bigcup\left\{ P_{i},\left\{ v_{i+1}\right\} ,P_{i+1},\ldots,\left\{ v_{i+\frac{k}{2}}\right\} \right\} $
	and $L\left(v_{i}\right)=\bigcup\left\{ \left\{ v_{i+\frac{k}{2}}\right\} ,P_{i+\frac{k}{2}},\left\{ v_{i+\frac{k}{2}+1}\right\} ,\dots,P_{i-1}\right\} $
	denote the right and the left hemispheres with respect to $v_{i}$,
	respectively. Similarly, if $k$ is odd then $R(v_{i})=\bigcup\left\{ P_{i},\left\{ v_{i+1}\right\} ,\ldots,P_{i+\frac{k-1}{2}}\right\} $
	and $L\left(v_{i}\right)=\bigcup\left\{ P_{i+\frac{k-1}{2}},\left\{ v_{i+\frac{k+1}{2}}\right\} ,\dots,P_{i-1}\right\} $.
	In addition we define hemispheres for non-terminal vertices. For an
	index $i$, all the vertices in $P_{i}$ have the same hemispheres.
	For a vertex $u\in P_{i}$, if $k$ is even, then $R(u)=\bigcup\left\{ \left\{ v_{i+1}\right\} ,P_{i+1},\left\{ v_{i+2}\right\} ,\ldots,P_{i+\frac{k}{2}}\right\} $
	and $L\left(u\right)=\bigcup\left\{ P_{i+\frac{k}{2}},\left\{ v_{i+\frac{k}{2}+1}\right\} ,P_{i+\frac{k}{2}+1},\dots,\left\{ v_{i}\right\} \right\} $.
	If $k$ is odd then $R(u)=\bigcup\left\{ \left\{ v_{i+1}\right\} ,P_{i+1},\left\{ v_{i+2}\right\} ,\ldots,\left\{ v_{i+\frac{k+1}{2}}\right\} \right\} $
	and $L\left(u\right)=\bigcup\left\{ \left\{ v_{i+\frac{k+1}{2}}\right\} ,P_{i+\frac{k+1}{2}},\left\{ v_{i+\frac{k+1}{2}+1}\right\} ,\dots,\left\{ v_{i}\right\} \right\} $.
	(See  \figureref{fig: hemisphere ilus} for an illustration.)
	
	\begin{figure}
		\begin{centering}
			\includegraphics[width=343pt]{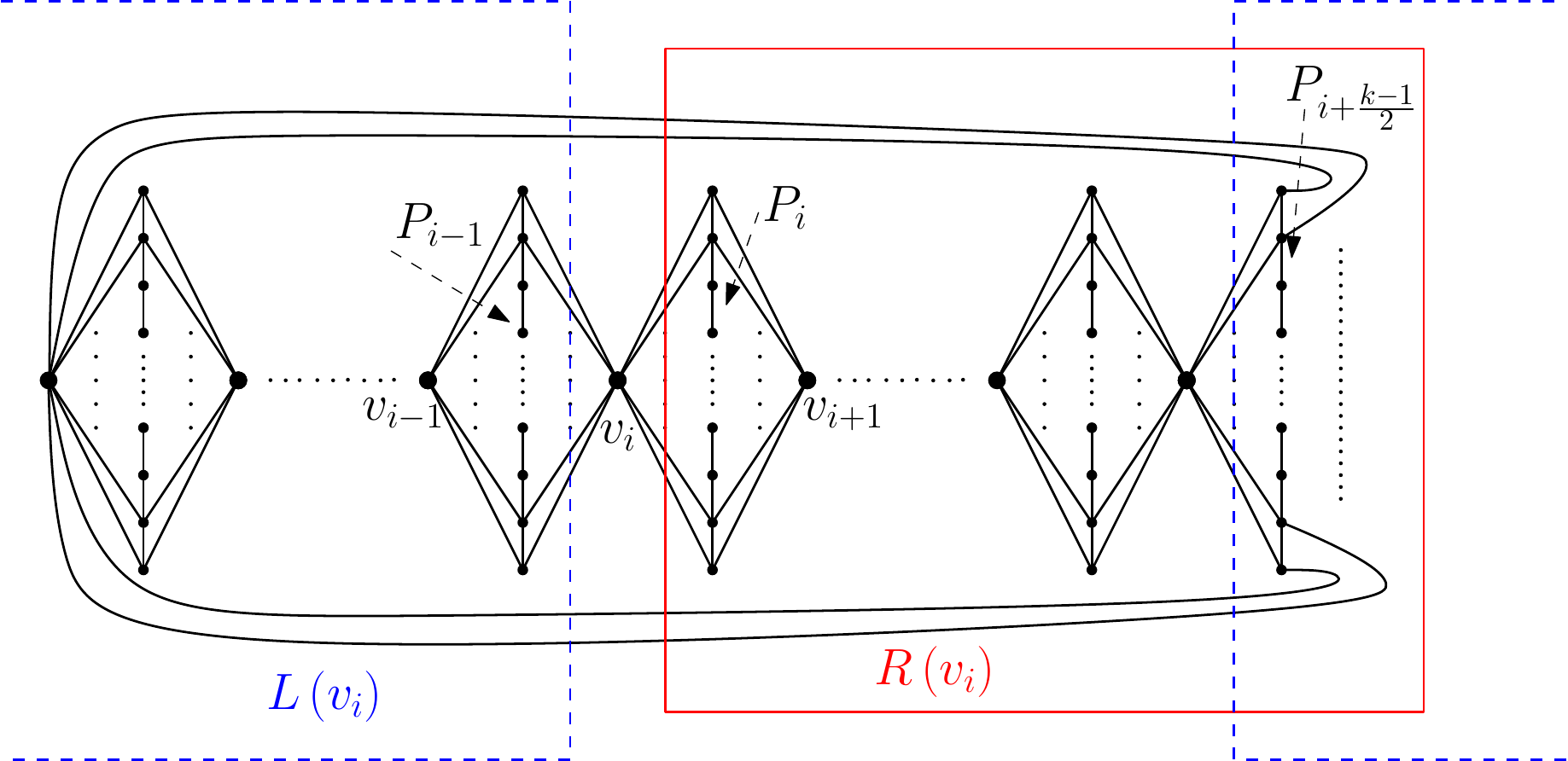}
			\par\end{centering}
		
		\caption{\label{fig: hemisphere ilus}An illustration of the partition of the
			graph to the right and the left hemispheres with respect to a terminal
			$v_{i}$. This example is for odd $k$. Note that the path $P_{i+\frac{k-1}{2}}$
			is both in the right and the left hemispheres.}
	\end{figure}

	For a vertex $u\in V$ and an index $i\in\left[0,k-1\right]$, we
	say that all the edges from $u$ to $P_{i}$, i.e., $\left\{ \left\{ u,z\right\} \ |\ z\in P_{i}\right\} ,$
	are of the \emph{same type}. In addition, the edge from $u$ to $v_{i}$
	has a unique type. (Note that each non-terminal vertex might have
	edges of $2k$ different types, while a terminal vertex can be incident
	to edges of $2k-1$ different types.) An edge $e=\left\{ u,z\right\} $
	such that $u,z$ are in the same path $P_{i}$ called an \textit{path-internal
		edge}. A vertex $u$ is called a \textit{same-path neighbor} of $a,b$
	in the tree $T$, if there are edges $\left\{ u,a\right\} ,\left\{ u,b\right\} $
	in $T$ and $a,b$ belong to the same path $P_{i}$. For a simple
	path $\pi$ in $\overline{G}_{k,n_{0},\dots,n_{k-1}}$ we say that
	$\pi$ is a \emph{one-sided} path if for any internal vertex $x$
	in $\pi$, the two edges $\left(x,y_{1}\right),\left(x,y_{2}\right)$
	which are incident on $x$ in $\pi$, connect $x$ to different hemispheres,
	i.e., e.g., $y_{1}$ is in the left hemisphere with respect to $x$,
	and $y_{2}$ is in the right one.
	
	The proof (that $\mu=0$) is by induction on $\sum_{i=1}^{k}n_{i}=A$.
	The base case where $\sum_{i=1}^{k}n_{i}=k$ (i.e., for all $i\in\left[0,k-1\right]$,
	$n_{i}=1$) follows by  \theoremref{thmLwrBndTrmStchMet}.
	
	The induction step: assume that the claim is true for $A$ and we
	will prove it for $A+1$. Let $T$ be some spanning tree of $\overline{G}_{k,n_{0},\dots,n_{k-1}}$
	with minimal weight among all the trees with $\gamma\left(T\right)=\mu$.
	By our assumption $\sum_{i=1}^{k}n_{i}=A+1$.
	\begin{claim}
		For any vertex $u$ (either terminal or a non-terminal one), if there
		exist two edges $\left\{ u,a\right\} $,$\left\{ u,b\right\} $ in
		$T$ that connect $u$ to two vertices $a,b\in R\left(u\right)$,
		then the two vertices $a$ and $b$ belong to the same path $P_{i}$.
		The same is true for the left hemisphere $L\left(u\right)$ of $u$
		as well. \end{claim}
	\begin{proof}
		Suppose for contradiction that there exist two edges $\left\{ u,a\right\} $,$\left\{ u,b\right\} $
		in $T$ as above (i.e., with $a,b\in R\left(u\right)$) and such that
		these two edges have different type. (In other words, $d_{\overline{G}}\left(u,a\right)\ne d_{\overline{G}}\left(u,b\right)$.)
		Without loss of generality $d_{\overline{G}}\left(u,a\right)<d_{\overline{G}}\left(u,b\right)$.
		We construct a new tree $T'$ by replacing $\left\{ u,b\right\} $
		by $\left\{ a,b\right\} .$ This change decreases the weight of the
		tree. Note that any path $\pi$ in $T$ that uses the edge $\left\{ u,b\right\} $
		can be replaced by a similar path that uses the edges $\left\{ u,a\right\} $
		and $\left\{ a,b\right\} $ instead of $\left\{ u,b\right\} $. Clearly,
		the length of the path does not change. Hence for every index $i$,
		the value $t_{T,i}$ does not increase. Hence $\gamma\left(T\right)=\min_{i}\left\{ n-t_{T,i}\right\} $
		does not decrease, and since $T$ is a tree with $\gamma\left(T\right)=\mu=\max_{T''}\left\{ \gamma\left(T''\right)\right\} $,
		it follows that $\gamma\left(T'\right)=\gamma\left(T\right)=\mu$.
		This is a contradiction to the minimality of the weight of $T$ among
		trees with $\gamma\left(\right)$ value equal to $\mu$. Obviously,
		the same argument also applies for $a,b\in L\left(u\right)$.
	\end{proof}
	We now continue proving \lemmaref{lem: Gknnn}. The rest of the proof
	splits into a number of cases which are characterized by existence
	of certain edges in the tree $T$.
	
	\begin{figure}
		\begin{centering}
			\includegraphics[scale=0.7]{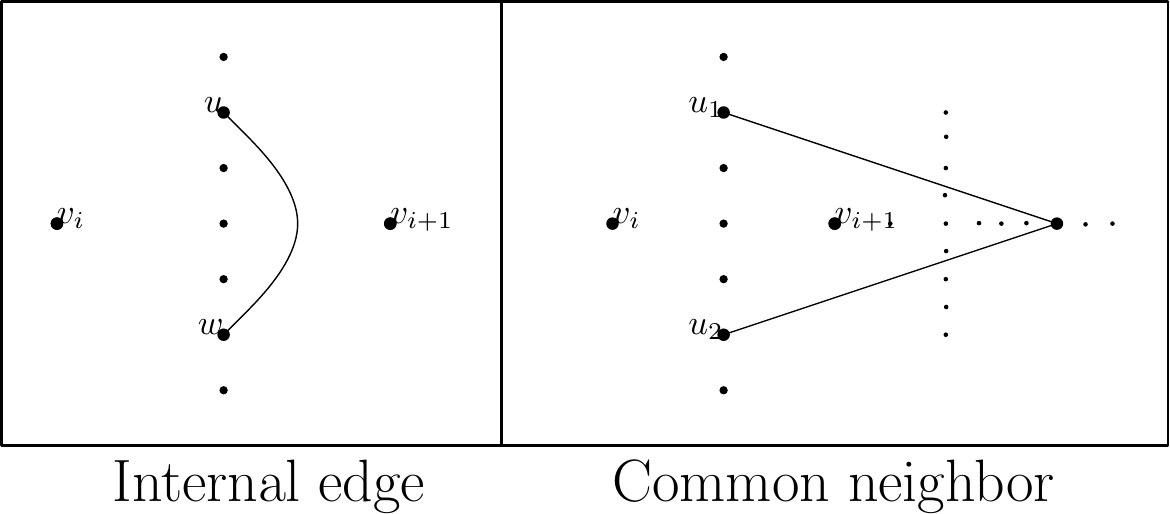}
			\par\end{centering}
		
		\caption{An illustration of the special edges that we use in the proof of \lemmaref{lem: Gknnn}.}
	\end{figure}
	\begin{figure}
		\begin{centering}
			\includegraphics[scale=0.7]{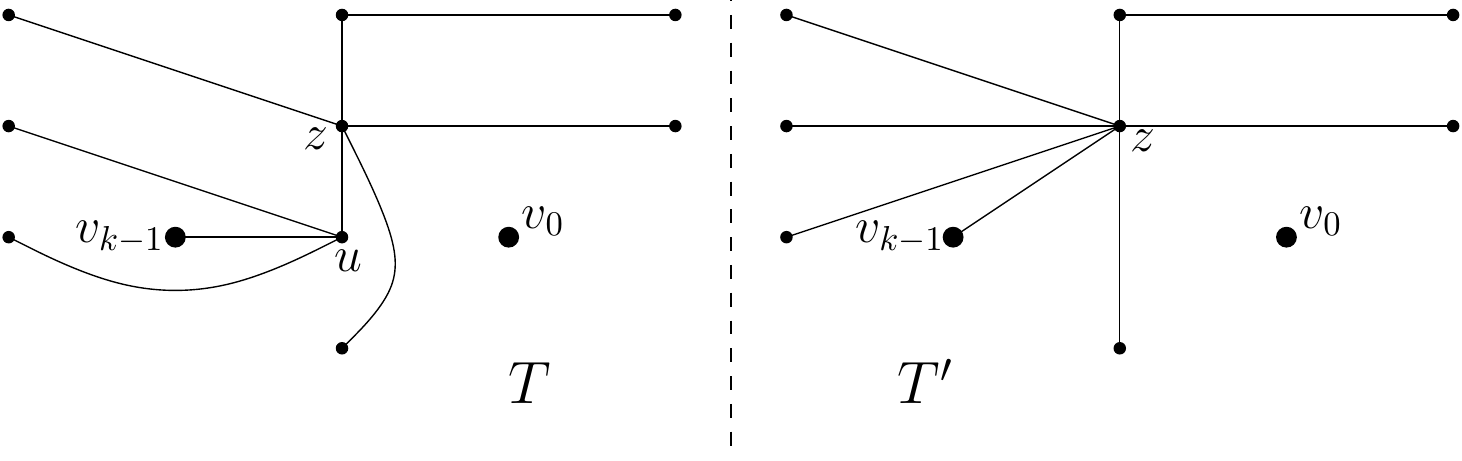}
			\par\end{centering}
		
		\caption{\label{fig:deliting edge}An illustration of deleting a path-internal
			edge $\left\{ u,z\right\} $ in the tree $T$.}
	\end{figure}
	
	\begin{itemize}
		\item Case 1: There is a path-internal edge in $T$. Assume there is a path-internal
		edge between vertices in a path $P_{j}$, for some index $j\in\left[0,k-1\right]$.
		Note that if all the vertices which have path-internal edges in $P_{j}$
		have more than one path-internal edge incident on them in $T$, then
		the graph induced by those vertices contains a cycle. However, this
		graph is a subgraph of $T$, contradiction. Hence there is a vertex
		$u\in P_{j}$ such that $u$ has exactly one path-internal edge $\left\{ u,z\right\} $
		incident on $u$ in $T$. Assume without loss of generality that $j=0$.
		Delete the vertex $u$ and replace each edge $\left\{ u,a\right\} $
		incident on $u$ in $T$ by an edge $\left\{ z,a\right\} $. (Note
		that $u$ has only one path-internal edge in $T$. Thus is the edge
		$\left\{ u,z\right\} $. Therefore vertex $a$ as above do not belong
		to $P_{j}=P_{0}$.) Denote the resulting tree by $T'$. Note that
		$T'$ is a spanning tree for $\overline{G}_{k,n_{0}-1,n_{1},\dots,n_{k-1}}$.
		(See \figureref{fig:deliting edge} for an illustration.) ($T'$
		is obviously connected and we reduce the number of edges by exactly
		1, hence it is a spanning tree of $\overline{G}_{k,n_{0}-1,n_{1},\dots,n_{k-1}}$.)
		Note that for every two vertices $a,b$ in $\overline{G}_{k,n_{0}-1,n_{1},\dots,n_{k-1}}$,
		we have $d_{T'}\left(a,b\right)\le d_{T}\left(a,b\right)$, to see
		this consider the shortest path $\pi$ from $a$ to $b$ in $T$.
		If the vertex $u$ is not taking part in $\pi$ then the claim is
		trivial, hence we assume that $u$ is taking part. If the path $\pi$
		include the edge $\left\{ u,z\right\} $, then obviously $\pi\backslash\left\{ \left\{ u,z\right\} \right\} $
		is a path from $a$ to $b$ in $T'$ and therefore $d_{T'}\left(a,b\right)\le d_{T}\left(a,b\right)$.
		Otherwise, the path $\pi$ does not include the edge $\left\{ u,z\right\} $,
		set $\pi=aw_{1}\dots w_{r}uw_{r+1}\dots w_{m}b$ where $w_{r},w_{r+1}\notin P_{0}$
		because $u$ has only one path-internal edge to $z$, therefore $\pi'=aw_{1}\dots w_{r}zw_{r+1}\dots w_{m}b$
		is path from $a$ to $b$ in $T'$ of length smaller or equal then
		$\pi$ in $T$ ($w_{r},w_{r+1}\notin P_{0}$ implies $d_{T}\left(w_{r},u\right)=d_{T'}\left(w_{r},z\right)$
		and $d_{T}\left(u,w_{r+1}\right)=d_{T'}\left(z,w_{r=1}\right)$, and
		all the other edges remain in $T$ with smaller or equal weight.)
		Moreover, for each $q\in P_{i}$, $d_{\overline{G}_{k,n_{0},\dots,n_{k-1}}}\left(v_{i},q\right)=d_{\overline{G}_{k,n_{0},\dots,n_{k-1}}}\left(v_{i+1},q\right)=d_{\overline{G}_{k,n_{0}-1,n_{1},\dots,n_{k-1}}}\left(v_{i},q\right)=d_{\overline{G}_{k,n_{0}-1,n_{1},\dots,n_{k-1}}}\left(v_{i+1},q\right)=w$,
		therefore $d_{T}\left(v_{i},q\right)\ge d_{T'}\left(v_{i},q\right)$
		and $d_{T}\left(v_{i+1},q\right)\ge d_{T'}\left(v_{i+1},q\right)$
		implies $\max\left\{ \frac{d_{T'}\left(v_{i},q\right)}{\overline{G}_{k,n_{0},\dots,n_{k-1}}\left(v_{i},q\right)},\frac{d_{T'}\left(v_{i+1},q\right)}{\overline{G}_{k,n_{0},\dots,n_{k-1}}\left(v_{i+1},q\right)}\right\} \ge\max\left\{ \frac{d_{T'}\left(v_{i},q\right)}{\overline{G}_{k,n_{0}-1,n_{1},\dots,n_{k-1}}\left(v_{i},q\right)},\frac{d_{T'}\left(v_{i+1},q\right)}{\overline{G}_{k,n_{0}-1,n_{1},\dots,n_{k-1}}\left(v_{i+1},q\right)}\right\} $.
		Denote $n'_{0}=n_{0}-1$, $n'_{i}=n_{i}$, for every $i\in\left[1,k-1\right].$
		By the induction hypothesis, there is an index $i$ such that $t_{T',i}=n'_{i}$.
		If $i\ne0$ then $t_{T',i}=n'_{i}$, and in particular by the previews
		argument $t_{T,i}\ge t_{T',i}$. Since, by definition, $t_{T,i}\le n_{i}$,
		it follows that $t_{T,i}=n_{i}$, as required. If $i=0$ , then $t_{T',0}=n'_{0}=n_{0}-1$.
		A vertex $a\in P_{0}\backslash\left\{ u,z\right\} $, contributes
		1 to $t_{T',0}$, and hence contributes 1 also to $t_{T,0}$. Moreover,
		$d_{T}\left(v_{0},u\right),d_{T}\left(v_{0},z\right)\ge d_{T'}\left(v_{0},z\right)$,
		and analogously $d_{T}\left(v_{1},u\right),d_{T}\left(v_{1},z\right)\ge d_{T'}\left(v_{1},z\right)$.
		Since $z$ contributes 1 to $t_{T',0}$ it follows that $u$ and $z$
		each contribute 1 to $t_{T,0}$ and so $t_{T,0}=n_{0}$, completing
		the proof of case 1.
		
		\item Case 2: There are no path-internal edges in $T$, and there is a path
		$P_{i}$, with two vertices $u_{1},u_{2}\in P_{i}$ which have a same-path
		neighbor $z\notin P_{i}$. Without loss of generality $i=0$. Similarly
		to the previous case, delete $u_{2}$ and replace any edge of the
		form $\left\{ u_{2},a\right\} $ by edge of the form $\left\{ u_{1},a\right\} $.
		Denote the resulting tree by $T'$. Note that $T'$ is a spanning
		tree for $\overline{G}_{k,n_{0}-1,n_{1},\dots,n_{k-1}}$. ($T'$ is
		obviously connected and we reduce the number of edges by exactly 1.)
		For each two vertices $a,b$ in $\overline{G}_{k,n_{0}-1,n_{1},\dots,n_{k-1}}$,
		$d_{T'}\left(a,b\right)\le d_{T}\left(a,b\right)$, to see this, consider
		the shortest path $\pi$ from $a$ to $b$ in $T$. If the vertex
		$u_{2}$ does not taking part in $\pi$ then the claim is trivial.
		Otherwise, set $\pi=aw_{1}\dots w_{r}u_{2}w_{r+1}\dots w_{m}b$, therefore
		$\pi'=aw_{1}\dots w_{r}u_{2}w_{r+1}\dots w_{m}b$ is path from $a$
		to $b$ in $T'$ of equal length to the path $\pi$ in $T$ (There
		are no path-internal edges hence $d_{T}\left(w_{r},u_{1}\right)=d_{T'}\left(w_{r},u_{2}\right)$
		and $d_{T}\left(u_{1},w_{r+1}\right)=d_{T'}\left(u_{2},w_{r=1}\right)$,
		and all the other edges remain with equal weight.). As in previews
		case, this argument implies for every $q\in P_{i}$, $\max\left\{ \frac{d_{T'}\left(v_{i},q\right)}{\overline{G}_{k,n_{0},\dots,n_{k-1}}\left(v_{i},q\right)},\frac{d_{T'}\left(v_{i+1},q\right)}{\overline{G}_{k,n_{0},\dots,n_{k-1}}\left(v_{i+1},q\right)}\right\} \ge\max\left\{ \frac{d_{T'}\left(v_{i},q\right)}{\overline{G}_{k,n_{0}-1,n_{1},\dots,n_{k-1}}\left(v_{i},q\right)},\frac{d_{T'}\left(v_{i+1},q\right)}{\overline{G}_{k,n_{0}-1,n_{1},\dots,n_{k-1}}\left(v_{i+1},q\right)}\right\} $.
		Denote $n'_{0}=n_{0}-1$, $n'_{i}=n_{i}$, for every $i\in\left[1,k-1\right].$
		By the induction hypothesis, there is a path $P_{i}^{'}$ such that
		$t_{T',i}=n'_{i}$. If $i\ne0$ then obviously $t_{T,i}=n'_{i}$.
		Otherwise, $i=0$, vertex $a\in P_{0}\backslash\left\{ u_{2}\right\} $,
		contributes 1 to $t_{T',0}$, and hence contributes 1 also to $t_{T,0}$.
		In particular $\max\left\{ d_{T}\left(u_{2},v_{0}\right),d_{T}\left(u_{2},v_{k-1}\right)\right\} \ge\max\left\{ d_{T'}\left(u_{1},v_{0}\right),d_{T'}\left(u_{1},v_{k-1}\right)\right\} $,
		hence $u_{2}$ also contributes 1 to $t_{T',0}$, we get $t_{T,0}=t_{T',0}+1=n'_{0}+1=n_{0}$
		as required.
		\item Case 3: There are no path-internal edges, and no same-path neighbors
		in $T$. In particular, any vertex has at most two edges: one edge
		to his right and left hemispheres (There is at most one type in each
		hemisphere, because two edges of the same type implies same-path neighbor.)
		Hence any path $\pi$ in $T$ is one-sided. There have to be two terminals
		$v_{i},v_{i+1}$ such that there is no path between them of length
		$2w$ (otherwise, there is a cycle). Hence for all $u\in P_{i}$ there
		are no edges to both $v_{i},v_{i+1}$. If without loss of generality
		$\left\{ u,v_{i}\right\} \notin T$, then the shortest one-sided path
		between $u$ and $v_{i}$ is of length at least $2k-1$, implying that
		the terminal distortion is at least $2k-1$, hence $t_{T,i}=0$ as required!
	\end{itemize}
\end{proof}
\begin{theorem}
	\label{thm: metric weight bound}For any positive integer parameters
	$k$, $n$ and $\epsilon>0$ such that $k\le\frac{\epsilon}{2}n$,
	any spanning tree $T$ of $\overline{G}_{k,n,n,\dots,n}$ where $w=\frac{k}{\epsilon}$,
	with terminal distortion at most $\left(2k-1\right)\left(1+\frac{\epsilon}{k^{2}}\right)$
	has lightness at least $\Omega\left(\frac{1}{\epsilon}\right)$. \end{theorem}
\begin{proof}
	By \lemmaref{lem: Gknnn} there is a path $P_{i}$ such that any
	$u\in P_{i}$ has distortion at least $2k-1$ from $v_{i}$ or $v_{i+1}$.
	We will show that each $u\in P_{i}$ is included in non path-internal
	edge of weight at least $w$. Some observations before we prove it:
	\begin{itemize}
		\item Each $u\in P_{i}$ is at distance exactly $\left(2k-1\right)\cdot w$
		from $v_{i}$ or $v_{i+1}$. Otherwise the distance is at least $\left(2k-1\right)w+1$
		which implies
		\begin{eqnarray*}
			\frac{d_{T}\left(u,v_{i}\right)}{d_{\overline{G}_{k,n,\dots,n}}\left(u,v_{i}\right)} & = & \frac{\left(2k-1\right)w+1}{w}=\left(2k-1\right)\left(1+\frac{1}{\left(2k-1\right)w}\right)\\
			& = & \left(2k-1\right)\left(1+\frac{\epsilon}{\left(2k-1\right)k}\right)>\left(2k-1\right)\left(1+\frac{\epsilon}{k^{2}}\right)
		\end{eqnarray*}
		in contradiction to the $\left(2k-1\right)\left(1+\frac{\epsilon}{k^{2}}\right)$
		terminal distortion.
		\item For two neighbors $u_{1},u_{2}$ in a tree, and any other vertex $v$.
		The shortest path from $u_{1}$ to $v$ goes trough $u_{2}$ or the
		shortest path from $u_{2}$ to $v$ goes trough $u_{1}$. Otherwise
		there is a cycle.
	\end{itemize}
	\begin{figure}[t]
		\begin{centering}
			\includegraphics[scale=0.55]{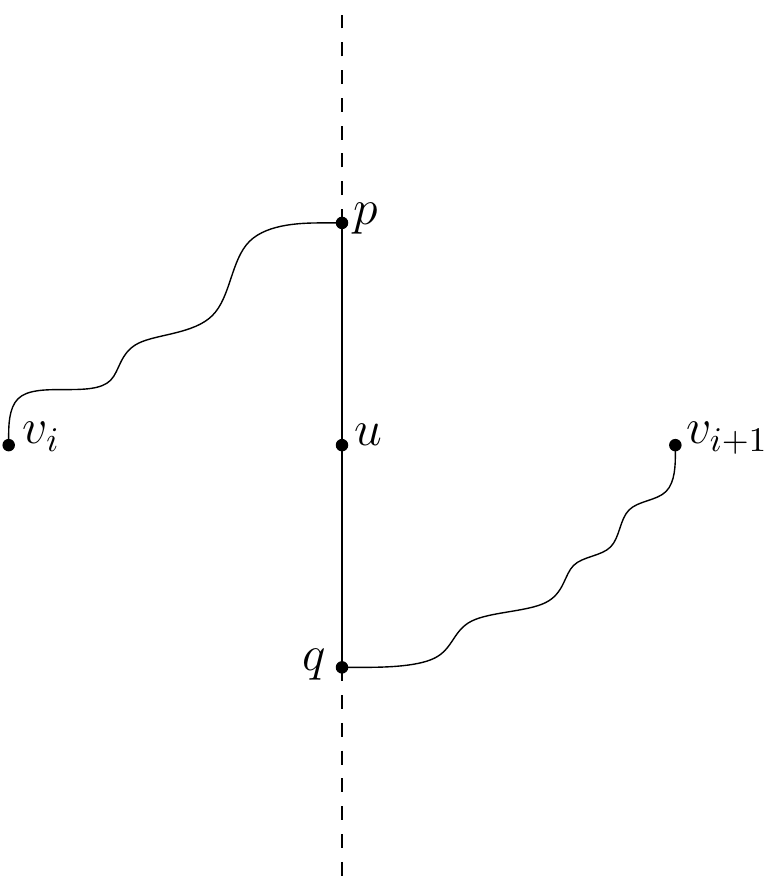}
			\par\end{centering}
		
		\caption{A visualization of the contradiction during the proof of \theoremref{thm: metric weight bound}.
			We assume that $u$ does not have outgoing edge of weight $w$. From
			the assumption that the distance from $u$ to $v_{i}$ is $\left(2k-1\right)w$
			we conclude that the distance from $q$ to $v_{i}$ is larger then
			$\left(2k-1\right)w$, a contradiction.}
	\end{figure}

	Assume by contradiction, that there is $u\in P_{i}$ with no outgoing
	edge of weight at least $w$. Let $p$ be the first vertex on the
	shortest path from $u$ to $v_{i}$, and $q$ be the first vertex
	on the shortest path from $u$ to $v_{i+1}$. W.l.o.g $d_{T}\left(v_{i},u\right)=\left(2k-1\right)w$.
	Hence $d_{T}\left(v_{i},p\right)<\left(2k-1\right)w$ which implies
	$d_{T}\left(v_{i+1},p\right)=\left(2k-1\right)w$. The shortest path
	from $p$ to $v_{i+1}$ goes through $u$, as otherwise $d_{T}\left(v_{i+1},u\right)>\left(2k-1\right)w$,
	contradiction. Hence $p\ne q$. Necessarily the shortest path from
	$q$ to $v_{i}$ goes through $u$, which implies $d_{T}\left(v_{i},q\right)>d_{T}\left(v_{i},u\right)=\left(2k-1\right)w$,
	contradiction.
	
	We conclude that $w\left(T\right)>n\cdot w$. Hence,
	\[
	\Psi\left(T\right)=\frac{w\left(T\right)}{w\left(MST\right)}\ge\frac{n\cdot w}{\left(2k-1\right)w+\left(n-1\right)k}=\frac{n\cdot\frac{k}{\epsilon}}{\left(2k-1\right)\frac{k}{\epsilon}+\left(n-1\right)k}>\frac{n}{\epsilon n+2k}.
	\]
	As $k\le\frac{\epsilon}{2}n$ we obtain $\Psi\left(T\right)\ge\frac{1}{2\epsilon}$
	as required.
\end{proof}
}

\section{Probabilistic Embedding into Ultrametrics with Strong Terminal Distortion}\label{sec:strong-ultra}

In this section we show our terminal variant of the result of \cite{FRT03}. We note that in a companion paper \cite{EFN15} a more general result is shown, and we give the full details here for completeness.

Recall that an ultrametric $\left(U,d\right)$ is a metric space satisfying a
strong form of the triangle inequality, that is, for all $x,y,z\in U$,
$d(x,z)\le\max\left\{ d(x,y),d(y,z)\right\} $. The following definition
is known to be an equivalent one (see \cite{DBLP:journals/dcg/BartalLMN05}).
\begin{definition}\label{def:ultra}
An ultrametric $U$ is a metric space $\left(U,d\right)$ whose elements
are the leaves of a rooted labeled tree $T$. Each $z\in T$ is associated
with a label $\Phi\left(z\right)\ge0$ such that if $q\in T$ is a
descendant of $z$ then $\Phi\left(q\right)\le\Phi\left(z\right)$
and $\Phi\left(q\right)=0$ iff $q$ is a leaf. The distance between
leaves $z,q\in U$ is defined as $d_{T}(z,q)=\Phi\left(\mbox{lca}\left(z,q\right)\right)$
where $\mbox{lca}\left(z,q\right)$ is the least common ancestor of
$z$ and $q$ in $T$.
\end{definition}

We now define probabilistic embeddings with terminal distortion. For a class of metrics ${\cal Y}$, a distribution ${\cal D}$ over embeddings $f_Y:X\to Y$ with $Y\in{\cal Y}$ has {\em expected terminal distortion} $\alpha$ if each $f_Y$ is non-contractive (that is, for all $u,w\in X$ and $Y\in\supp({\cal D})$, it holds that $d_X(u,w)\le d_Y(f_Y(u),f_Y(w))$), and for all $v\in K$ and $u\in X$,
\[
\E_{Y\sim{\cal D}}[d_Y(f_Y(v),f_Y(u))]\le \alpha\cdot d_X(v,u)~.
\]
The notion of strong terminal distortion translates to this setting in the obvious manner.

Let $(X,d)$ be a metric on $n$ points, and $K\subseteq X$ the set of terminals of size $|K|=k$.
\theoremref{thm:trans} combined with \cite{FRT03} gives an embedding to a distribution over ultrametrics with expected terminal distortion $O(\log k)$.
However, this does not give any guarantee on the distortion of two non-terminal points.\footnote{In fact, it is not hard to verify that this approach may lead to unbounded distortion.}
In order to obtain {\em strong} terminal distortion, we modify the FRT algorithm, so that it gives "preference" to the terminals. In particular, in the heart of the FRT algorithm is a construction of a random partitioning scheme, based on \cite{CKR01}, whose purpose is to decompose the metric to bounded diameter pieces, such that the probability of separating pairs is "sufficiently small". The partition is created by choosing a random permutation of $X$, a random radius in some specified interval, and then creating clusters as balls of the chosen radius, in the order given by the permutation. To obtain improved expected distortion guarantees for the terminals, we enforce them to be the first $k$ points of the permutation. We show that this restriction improves dramatically the expected distortion of {\em any pair} that contains a terminal, while all the other pairs suffer only a factor of 2 in the expected distortion bound.

\begin{theorem}\label{thm:strong-ultra}
Given a metric space $(X,d)$ of size $|X|=n$ and a subset of terminals $K\subseteq X$ of size $|K|=k$, there exists a distribution over embeddings of $X$ into ultrametrics with strong terminal distortion $\left(O(\log k),O(\log n)\right)$.
\end{theorem}
\begin{proof}

Let $\Delta$ be the diameter of $X$. We assume w.l.o.g that the minimal
distance in $X$ is $1$, and let $\delta$ be the minimal integer so that $\Delta\le 2^{\delta}$. We shall create a hierarchical laminar partition, where for each $i\in\{0,1,\dots,\delta\}$, the clusters of level $i$ have diameter at most $2^i$, and each of them is contained in some level $i+1$ cluster. Recall \defref{def:ultra}. The ultrametric is built in the natural manner, the root corresponds to the level $\delta$ cluster which is $X$, and each cluster in level $i$ corresponds to an inner node of the ultrametric with label $2^i$, whose children correspond to the level $i-1$ clusters contained in it. The leaves correspond to singletons, that is, to the elements of $X$. Clearly, the ultrametric will dominate $(X,d)$.\footnote{A metric $d'$ on $X$ {\em dominates} $d$ if for all $x,y\in X$, $d(x,y)\le d'(x,y)$.}

In order to define the partition, we sample a permutation uniformly at random from the set of permutations that have the terminals as the preface (there are $k!\cdot(n-k)!$ such permutations),
and sample a number $\beta\in\left[1,2\right]$ uniformly at random.
In each step, each vertex $x$ chooses the vertex $u$ with minimal value according
to $\pi$ among the vertices of distance at most $\beta_{i}=\beta\cdot 2^{i-2}$
from $x$, and joins the cluster of $u$. Note that a vertex may not
belong to the cluster associated with it, and some clusters may be empty (which we can discard). The description of the hierarchical partition appears in \algref{alg:mod-frt}.

\begin{algorithm}[h]
\caption{\texttt{Modified FRT}$(X,K)$}\label{alg:mod-frt}
\begin{algorithmic}[1]
\STATE Choose a random permutation $\pi$ of $X$ such that $\pi^{-1}\left(K\right)=\left[k\right]$.
\STATE Choose $\beta\in\left[1,2\right]$ randomly by the distribution with the following probability density function $p\left(x\right)=\frac{1}{x\ln2}$.
\STATE Let $D_{\delta}=X$; $i\leftarrow\delta-1$.
\WHILE {$D_{i+1}$ has non-singleton clusters}
\STATE Set $\beta_{i}\leftarrow\beta\cdot 2^{i-2}$.
\FOR {$l=1,\dots,n$}
\FOR {every cluster $S$ in $D_{i+1}$}
\STATE Create a new cluster in $D_i$, consisting of all unassigned vertices in $S$ closer than $\beta_{i}$ to $\pi\left(l\right)$.
\ENDFOR
\ENDFOR
\STATE $i\leftarrow i-1$.
\ENDWHILE
\end{algorithmic}
\end{algorithm}

Let $T$ denote the (random) ultrametric created by the hierarchical partition of \algref{alg:mod-frt},
and $d_{T}\left(u,v\right)$ the distance between $u$ and $v$ in
$T$. Consider the
clustering step at some level $i$, where clusters in $D_{i+1}$ are picked for partitioning. In each iteration $l$, all unassigned
vertices $z$ such that $d\left(z,\pi(l)\right)\le\beta_{i}$,
assign themselves to the cluster of $\pi(l)$.
Fix an arbitrary pair $\left\{ v,u\right\} $. We say that center $w$ {\em settles} the pair
$\left\{ v,u\right\} $ at level $i$, if it is the first center so that at least one of $u$ and $v$ gets assigned to its cluster. Note that exactly
one center $w$ settles any pair $\left\{ v,u\right\}$ at any particular
level. Further, we say that a center $w$ {\em cuts} the pair $\left\{ v,u\right\} $
at level $i$, if it settles them at this level,
and exactly one of $u$ and $v$ is assigned to the cluster of $w$ at level $i$.
Whenever $w$ cuts a pair $\left\{ v,u\right\} $ at level $i$, $d_{T}\left(v,u\right)$
is set to be $2^{i+1}\le 8\beta_i$. We charge this length to the
vertex $w$ and define $d_T^w\left(v,u\right)$ to be $\sum_{i}\mathbf{1}\left(w\mbox{ cuts \ensuremath{\left\{ v,u\right\} } at level \ensuremath{i}}\right)\cdot 8\beta_i$
(where $\mathbf{1}\left(\cdot\right)$ denotes an indicator function). Clearly,
$d_{T}\left(v,u\right)\le\sum_{w\in X}d_T^w\left(v,u\right)$.

We now arrange the vertices of $K$ in non-decreasing order of their distance
from the pair $\{v,u\}$ (breaking ties arbitrarily).
Consider the $s$th terminal $w_{s}$
in this sequence. We now upper bound the expected value of $d_T^{w_s}\left(v,u\right)$
for an arbitrary $w_{s}$. W.l.o.g assume that $d_{X}\left(w_{s},v\right)\le d_{X}\left(w_{s},u\right)$.
For a center $w_{s}$ to cut $\left\{ v,u\right\} $, it must be the
case that:
\begin{enumerate}
\item $d_{X}\left(w_{s},v\right)\le\beta_{i}<d_{X}\left(w_{s},u\right)$
for some $i$.
\item $w_{s}$ settles $\{v,u\}$ at level $i$.
\end{enumerate}
Note that for each $x\in[d_{X}\left(w_{s},v\right),d_{X}\left(w_{s},u\right))$, the probability that $\beta_i\in[x,x+dx)$ is at most $\frac{dx}{x\cdot\ln 2}$. Conditioning on $\beta_i$ taking such a value $x$, any one of $w_1,\dots w_s$ can settle $\{v,u\}$, and the probability that $w_s$ is first in the permutation among $w_1,\dots w_s$ is $\frac{1}{s}$. Thus we obtain,
\begin{equation}\label{eq:w_s}
\E[d_T^{w_s}(v,u)]\le \int_{d_{X}(w_{s},v)}^{d_{X}(w_{s},u)}8x\cdot\frac{dx}{x\ln 2}\cdot\frac{1}{s}
=\frac{8(d_{X}(w_{s},u)-d_{X}(w_{s},v))}{s\cdot\ln 2}
\le\frac{16}{s}\cdot d_X(v,u)~.
\end{equation}

The crucial observation is that if $y\in X\setminus K$, and at least one of $v,u$ is a terminal, w.l.o.g $v\in K$, then $y$ cannot settle $\{v,u\}$. The reason is that $v$ always appears before $y$ in $\pi$, so $v$ will surely be assigned to a cluster when it is the turn of $y$ to create a cluster. This leads to the conclusion that for all $v\in K$ and $u\in X$
\[
\E[d_T(v,u)]\le\sum_{s=1}^k\E[d_T^{w_s}(v,u)]\stackrel{\eqref{eq:w_s}}{\le}16d_X(v,u)\sum_{s=1}^k\frac{1}{s}=O(\log k)\cdot d_X(v,u)~.
\]

It remains to bound the expected distortion of non-terminal pairs, so fix some $v,u\in X\setminus K$. Arrange the vertices $y_1,\dots,y_{n-k}$ of $X\setminus K$ according to their distance from $\{v,u\}$, in non-decreasing order. A similar reasoning as above gives the bound of \eqref{eq:w_s} for the $s$th vertex $y_s$ in this ordering to cut $\{v,u\}$ (recall that the permutation over the non-terminals is uniformly distributed). In fact, the bound is even slightly loose, because we disregard the possibility that a terminal will settle $\{v,u\}$. We conclude that
\begin{eqnarray*}
\E[d_T(v,u)]&\le&\sum_{s=1}^k\E[d_T^{w_s}(v,u)] + \sum_{s=1}^{n-k}\E[d_T^{y_s}(v,u)]\\ &\le&16d_X(v,u)\left(\sum_{s=1}^k\frac{1}{s}+\sum_{s=1}^{n-k}\frac{1}{s}\right)\\
&=&O(\log k+\log n)\cdot d_X(v,u)\\
&=&O(\log n)\cdot d_X(v,u)~.
\end{eqnarray*}

\end{proof}

\section{Probabilistic Embedding into Trees with Terminal Congestion}\label{sec:cong}

In this section we focus on embeddings into trees that approximate  capacities of cuts, rather than distances between vertices. This framework was introduced by R{\"a}cke \cite{R02} (for a single tree), and in \cite{R08} he showed how to obtain capacity preserving probabilistic embedding from a distance preserving one, such as the ones given by \cite{FRT03}. Later, \cite{AF09} showed a complete equivalence between these notions in random tree embeddings. Here we show our terminal variant of these results. Informally, we construct a distribution over {\em capacity-dominating} trees (each cut in each tree is at least as large as the corresponding cut in the original graph), and for each edge, its expected congestion is bounded accordingly, with an improved bound for edges {\em containing a terminal}.

We next elaborate on the notions of capacity and congestion, and their relation to distance and distortion, following the notation of \cite{AF09}.
Given a graph $G=(V,E)$, let ${\cal P}$ be a collection of multisets of edges in $G$.
A map $M:E\rightarrow {\cal P}$, where $M(e)$ is a path between the endpoints of $e$, is called a {\em path mapping} (the path is not necessarily simple).
Denote by $M_{e'}(e)$ the number of appearances of $e'$ in $M(e)$.

The path mapping relevant to the rest of this section is constructed as follows: given a tree $T=\left(V,E_{T}\right)$ (not necessarily a subgraph),
for each edge $e'\in E_{T}$ let $P_{G}(e')$ be a shortest path
between the endpoints of $e'$ in $G$ (breaking ties arbitrarily), and similarly for $e\in E$,
let $P_{T}(e)$ be the unique path between the endpoints of $e$
in $T$. Then for an edge $e\in E$, where $P_{T}(e)=e'_{1}e'_{2}\dots e'_{r}$,
the path $M(e)$ is defined as $M(e)=P_{G}(e'_{1})\circ P_{G}(e'_{2})\circ\cdots\circ P_{G}(e'_{r})$ (where $\circ$ denotes concatenation).
In what follows fix a tree $T$, and let $M$ be the path mapping of $T$.

Fix a weight function $w:E\rightarrow\mathbb{R}_{+}$, and a capacity function
$c:E\rightarrow\mathbb{R}_{+}$. For an edge $e\in E$, $\mbox{dist}_T\left(e\right)=\sum_{e'\in E}M_{e'}(e)\cdot w\left(e'\right)$
is the weight of the path $M(e)$, and $\mbox{load}_{T}\left(e\right)=\sum_{e'\in E}M_{e}(e')\cdot c\left(e'\right)$
is the sum (with multiplicities) of the capacities of all the edges whose path is using $e$.
Define $\str_T(e)=\frac{\mbox{dist}_{T}\left(e\right)}{w(e)}$ to be the distortion of $e$ in $T$, and $\con_T(e)=\frac{\mbox{load}_{T}\left(e\right)}{c(e)}$ is the congestion of $e$.
Note that if $T$ is a subgraph of $G$, then
$\mbox{dist}_{T}\left(e\right)$ is the length of the unique path
between the endpoints of $e$, while $\mbox{load}_{T}\left(e\right)$
is the total capacity of all the edges of $E$ that are in the cut obtained
by deleting $e$ from $T$ (for $e\notin T$, $\mbox{load}_{T}(e)=0$).

\begin{definition}
Let $K\subseteq V$ be a set of terminals of size $k$, and let $E_K\subseteq E$ be the set of edges that contain a terminal. We say that a distribution ${\cal D}$ over trees has strong terminal congestion $(\alpha,\beta)$ if for every $e\in E_K$.
\[
\con_{\cal D}(e):=\E_{T\sim {\cal D}}[\con_T(e)]\le \alpha~,
\]
and for any $e\in E$, $\con_{\cal D}(e)\le \beta$.
\end{definition}

A tight connection between distance preserving and capacity preserving mappings was shown in \cite{AF09}. We generalize their theorem to the terminal setting in the following manner.
\begin{theorem}\label{thm:cong}
The following statements are equivalent for a graph $G$:
\begin{itemize}
\item For every possible {\em weight} assignment $G$ admits a probabilistic embedding into trees with strong terminal {\em distortion} $(\alpha,\beta)$.

\item For every possible {\em capacity} assignment $G$ admits a probabilistic embedding into trees with strong terminal {\em congestion} $(\alpha,\beta)$.
\end{itemize}
\end{theorem}
A corollary of \theoremref{thm:cong}, achieved by applying \theoremref{thm:strong-ultra} using the algorithmic framework described in \cite{AF09}, is:\footnote{Even though the embedding of \theoremref{thm:strong-ultra} is into ultrametrics, which contain Steiner vertices, these can be removed while increasing the distortion of each pair by at most a factor of 8 \cite{G01}.}
\begin{corollary}\label{cor:cong}
For any graph $G=(V,E)$ on $n$ vertices, a set $K\subseteq V$ of $k$ terminals, and any capacity function, there exists a distribution over trees with strong terminal congestion $(O(\log k),O(\log n))$. Moreover, there is an efficient procedure to sample from the distribution.
\end{corollary}

\begin{proof}[Proof of \theoremref{thm:cong}]
Assuming the first item holds we prove the second. Let $\kappa(e)=\left\{\begin{array}{ccc} 1/\alpha & e\in E_K\\
1/\beta & \text{otherwise} \end{array} \right.$.
Given any capacity function $c:E\to\R_+$, we would like to show that there exists a distribution ${\cal D}'$ such that for any $e\in E$, $\E_{T\sim {\cal D}'}[\kappa(e)\cdot\con_T(e)]\le 1$.
By applying the minimax principle (as in \cite{AKPW95}), it suffices to show that for any coefficients $\{\lambda_{e}\}_{e\in E}$ with $\lambda_e\ge 0$ and $\sum_{e\in E}\lambda_e=1$, there exists a single tree $T$ such that
\begin{equation}\label{eq:ssho}
\sum_{e\in E}\lambda_e\cdot\kappa(e)\cdot\con_T(e)\le 1~.
\end{equation}
To this end, define the weights $w(e)=\kappa(e)\cdot\frac{\lambda_e}{c(e)}$, and by the first assertion there exists a distribution ${\cal D}$ over trees such that for any $e\in E$,
\[
\E_{T\sim {\cal D}}[\kappa(e)\cdot\str_T(e)]\le 1~.
\]
By applying the minimax again, there exists a single tree $T$ such that
\[
\sum_{e\in E}\lambda_e\cdot\kappa(e)\cdot\str_T(e)\le 1~.
\]

Now,

\begin{eqnarray*}
1&\ge&\sum_{e\in E}\lambda_e\cdot\kappa(e)\cdot\str_T(e)\\
&=&\sum_{e\in E}\lambda_e\cdot\kappa(e)\cdot \frac{\sum_{e'\in E}M_{e'}(e)\cdot w\left(e'\right)}{w(e)}\\
&=&\sum_{e\in E}\lambda_e\cdot\kappa(e)\cdot \frac{\sum_{e'\in E}M_{e'}(e)\cdot \kappa(e')\cdot\lambda_{e'}/c(e')}{\kappa(e)\cdot\lambda_e/c(e)}\\
&=&\sum_{e'\in E}\lambda_{e'}\cdot\kappa(e')\cdot \frac{\sum_{e\in E}M_{e'}(e)\cdot c(e)}{c(e')}\\
&=&\sum_{e'\in E}\lambda_{e'}\cdot\kappa(e')\cdot \con_T(e')~,
\end{eqnarray*}
which concludes the proof of \eqref{eq:ssho}. The second direction is symmetric.
\end{proof}

{\bf Capacity Domination Property.} As \cite{R08,AF09} showed, under the natural capacity assignment, any tree $T$ supported by the distribution of \theoremref{thm:cong} has the following property: Any multi-commodity flow in $G$ can be routed in $T$ with no larger congestion. We would like to show this explicitly, using the language of cuts, as this will be useful for the algorithmic applications.

Fix some tree $T=(V,E_T)$, and for any edge $e'\in E_T$ let $S_{T,e'}\subseteq V$ be the cut obtained by deleting $e'$ from $T$. Define the capacities $C_T:E_T\to\R_+$ by
\[
C_T(e')=\sum_{e\in E(S_{T,e'},\bar{S}_{T,e'})}c(e)~,
\]
where $E(S,\bar{S})$ denotes the set of edges in the graph crossing the cut $S$.
(Observe that for spanning trees, $C_T(e)={\rm load}_T(e)$.)
\begin{lemma}\label{lem:capacities inequality}
For any graph $G=(V,E)$ and tree $T=(V,E_T)$ with capacities as defined above, for any set $S\subseteq V$ it holds that
\begin{equation}
\sum_{e\in E\left(S,\bar{S}\right)}c(e)\le\sum_{e'\in E_{T}\left(S,\bar{S}\right)}C_T(e')\le\sum_{e\in E\left(S,\bar{S}\right)}{\rm load}_{T}(e)~.\label{eq:capacities inequality}
\end{equation}
\end{lemma}
\begin{proof}
We begin with the left inequality. For any graph edge $e\in E(S,\bar{S})$, there exists a tree edge $e'\in E_T(S,\bar{S})$ such that $e'\in P_T(e)$, because the path $P_T(e)$ must cross the cut. Since removing $e'$ from $T$ separates the endpoints of $e$, $C_T(e')$ will contain the term $c(e)$. We conclude that
\[
\sum_{e\in E\left(S,\bar{S}\right)}c(e)\le\sum_{e'\in E_{T}\left(S,\bar{S}\right)}C_T(e')~.
\]

For the right inequality, consider an edge $e\in E$, and note that for any tree edge $e'\in E_T$ such that $e\in P_G(e')$, every edge $e''\in E(S_{T,e'},\bar{S}_{T,e'})$ will have $e\in M(e'')$ and thus contribute to ${\rm load}_T(e)$ (perhaps multiple times, due to different $e'$). This implies that
\begin{equation}\label{eq:loadd}
{\rm load}_T(e)=\sum_{\{e'\in E_T~:~e\in P_G(e')\}}C_T(e')~.
\end{equation}
Next, observe that any tree edge $e'\in E_T(S,\bar{S})$ must have at least one graph edge $e\in E(S,\bar{S})$ such that $e\in P_G(e')$. This suggests that
\begin{eqnarray*}
\sum_{e\in E(S,\bar{S})}{\rm load}_T(e)&\stackrel{\eqref{eq:loadd}}{=}&\sum_{e\in E(S,\bar{S})}\sum_{\{e'\in E_T:e\in P_G(e')\}}C_T(e')\\
&=&\sum_{e'\in E_T}|E(S,\bar{S})\cap P_G(e')|\cdot C_T(e')\\
&\ge&\sum_{e'\in E_T(S,\bar{S})}C_T(e')~.
\end{eqnarray*}

\end{proof}

\section{Strong Terminal Embedding of Graphs into a Distribution of Spanning Trees}\label{sec:Distrbution spanninig tree}

In this section we consider strong terminal embedding of graphs into a distribution over their spanning trees. We will show $(\tilde{O}(\log k),\tilde{O}(\log n))$-strong terminal distortion as promised in the introduction.

\begin{theorem}\label{thm:strong-spanning}
Given a weighted graph $G=\left(V,E,w\right)$ on $n$ vertices, and a subset of $k$ terminals $K\subseteq V$, there exists a distribution over spanning trees of $G$ with strong terminal distortion $(\tilde{O}(\log k),\tilde{O}(\log n))$.
\end{theorem}

We will use the framework of \cite{AN12}, in particular their petal decomposition structure in order to obtain the distribution over spanning trees.
Roughly speaking, it is an iterative method to build a spanning tree. In each level, the current graph is partitioned into smaller diameter pieces, called petals, and a single central piece, which are then connected by edges in a tree structure. Each of the petals is a ball in a certain metric. The main advantage of this framework, is that it produces a spanning tree whose diameter is proportional to the diameter of the graph, but allows very large freedom for the choice of radii of the petals. Specifically, if the graph diameter is $\Delta$, each radius can be chosen in an interval of length $\approx\Delta$. Intuitively, if a pair is separated by the petal decomposition, then its distance in the tree will be $O(\Delta)$. So we would like a method to choose the radii, that will give the appropriate bounds both on pairs containing a terminal, and for all other pairs, simultaneously.
We note that in the star-partition framework of \cite{EEST05} (used also in \cite{ABN08}), the tree radius increases introduces a factor to the distortion that inherently depends on $n$, which does not seem to allow terminal distortion that depends on $k$ alone.

In order to "take care" of pairs containing a terminal, we will need to somehow give the terminals a "preference" in the decomposition. Since the \cite{CKR01,FRT03} style partitioning scheme cannot work in a graph setting (because it does not produce strong diameter clusters, and the cluster's center may not be contained in it), we turn to the partitions based on truncated exponential distribution as in \cite{B96,B04,ABN08}. To implement the "terminal preference" idea, we first build petals with the terminals as centers, and only then build petals for the remaining points.
There are few technical subtleties needed to assure that pairs containing a terminal suffer small expected distortion.
By a careful choice for the petal's center and the interval from which to choose the radius, we guarantee that in a level where no terminals are separated, none of the relevant balls around terminals are in danger of being cut. Using that the radius of clusters decreases geometrically, we conclude that each ball is "at risk" in at most $O(\log k)$ levels.

\subsection{Petal decomposition description}

In this section we present the petal decomposition algorithm, and quote some of its properties.
We do not provide full proofs of these properties, these can be found in \cite{AN12}. The petals are built using an alternative metric on the graph.
\begin{definition}[Cone metric] Given a graph $G=(V,E,w)$, a subset $X\subseteq V$
and points $x,y\in X$, define the $\emph{cone-metric}$ $\rho=\rho(X,x,y):X^{2}\to\mathbb{R}^{+}$
as $\rho(u,v)=\left|\left(d_{X}(x,u)-d_{X}(y,u)\right)-\left(d_{X}(x,v)-d_{X}(y,v)\right)\right|$.
\end{definition}
This is actually a pseudo metric. For a cone metric $\rho=\rho(X,x,y)$,
\[
\rho(y,u)=\left|\left(d_{X}(x,u)-d_{X}(y,u)\right)-\left(d_{X}(x,y)-d_{X}(y,y)\right)\right|=d_{X}(x,y)+d_{X}(y,u)-d_{X}(x,u)
\]
is the difference between the shortest path from $x$ to $u$ in $G[X]$
and the shortest path from $x$ to $u$ in $G[X]$ that goes through $y$.\footnote{For $X\subseteq V$, $G[X]$ is the induced graph on the vertices of $X$.}
Therefore, the ball $B_{(X,\rho)}(y,r)$ in the cone metric
$\rho=\rho(X,x,y)$ centered in $y$ with radius $r$, contains all such vertices $u$. Denote by $P_{xy}(G)$ the shortest\footnote{For simplicity, we assume that for each pair of vertices there is a unique shortest path in $G$.} path between $x$ and $y$ in $G$. If $Y$ is a subset of the vertices of $G$, we denote by $P_{xy}(Y)$ the shortest path from $x$ to $y$ in $G[Y]$.

A petal is a union of balls in the cone metric. In the $\texttt{create-petal}$
algorithm, while working in a subgraph $G[Y]$ with two specified vertices: a center $x_{0}$
and a target $t$, we define $W_{r}\left(Y,x_{0},t\right)=\bigcup_{p\in P_{x_{0}t}:\ d_{Y}(p,t)\le r}B_{(Y,\rho(Y,x_{0},p))}(p,\frac{r-d_{Y}(p,t)}{2})$
which is union of balls in the cone metric, where any vertex $p$
in the shortest path from $x_{0}$ to $t$ of distance at most $r$
from $t$ is a center of a ball with radius $\frac{r-d_{Y}(p,t)}{2}$. We will often omit the parameters and write just $W_{r}$ if $Y,x_{0}$
and $t$ are clear from the context. The next claim, which is implicit in \cite{AN12}, states that $W_{r}$
is monotone (in $r$) and that $W_{r+4l}$ contains the ball with
radius $l$ around $W_{r}$ (where the ball is taken in the shortest path metric in $G[Y]$).
\begin{claim}
\label{ClmW_rProp}For $W_{r}\left(Y,x_{0},t\right)=\bigcup_{p\in P_{x_{0}t}:\ d_{Y}(p,t)\le r}B_{(Y,\rho(Y,x_{0},p))}(p,(r-d_{Y}(p,t))/2)$
\begin{enumerate}
\item $W_r$ is monotone in $r$, i.e., for $r\le r'$ it holds that $W_{r}\subseteq W_{r'}$.
\item For every $y\in W_{r}$ and $l$, the ball $B_Y(y,l)$ contained in $W_{r+4l}$, i.e., $\forall z,~d_Y\left(y,z\right)\le l\Rightarrow z\in W_{r+4l}$
\end{enumerate}
\end{claim}

Next, we quote some properties of the petal decomposition algorithm proved in \cite{AN12}.
Their proofs remain valid even after our modification, because the algorithm allows for arbitrary choice of targets after the first petal is built, and allows any radius in the appropriate interval. For a subset $X\subseteq V$, let $\Delta_{X,x_0}$ denote the radius of $G[X]$ with respect to $x_0$ (the maximal distance from $x_0$ in $G[X]$). We often omit the subscript $x_0$ if it is clear from the context.
\begin{fact}
\label{FactPetal3/4Radius} For a graph $G$, a vertex set $X$,
and some vertices $x_{0},t\in X$,  \textup{$\texttt{petal-decomposition }(G\left[X\right],x_{0},t,K)$}
returns as output the sequence $(X_{0},\dots,X_{s},\left\{ y_{1},x_{1}\right\} ,\dots,\left\{ y_{s},x_{s}\right\} ,t_{0},\dots,t_{s})$
such that for each $j\le s$, $\Delta_{X_j,x_j}\le 7\Delta_{X,x_0}/8$.

\end{fact}
\begin{fact}
\label{Fact:PetalDecoReturnsTree} The $\texttt{hierarchical-petal-decomposition}$
algorithm returns a tree.

\end{fact}
\begin{fact}
\label{Fact PetalDec reduce wght once} Every edge $e\in E$ can
have its weight multiplied by $\frac{1}{2}$ at most once throughout
the execution of the algorithm. \footnote{This multiplication may happen only for edges in the special paths $P_{x_j,t_j}$ for a petal $X_j$ which is not the first special petal. Note that after the creation of $X_j$, $t_j$ will be the target for the first special petal in the decomposition of $X_j$ which will ensure that the weights of the edges in $P_{x_j,t_j}$ won't decrease again.}

\end{fact}
\begin{fact}
\label{FactPetalTreeRadiusBound4} For a graph $G$ and tree a $T$
created by the $\texttt{hierarchical-petal-decomposition}$,
\[
\Delta_{T}\le4\Delta_{G}.
\]
\end{fact}
\begin{fact}
\label{FactPetals dont cut shortest path to x0} Let $1\le j\le s$
be an integer and $z\in Y_{j}$,\footnote{$Y_{j}$ is defined in the \texttt{petal-decomposition} procedure.} then $P_{x_{0}z}\left(X\right)\subseteq G\left[Y_{j}\right]$.
\end{fact}
From \factref{FactPetalTreeRadiusBound4} we can deduce that if the
radius of $G$ was $\Delta$ (with respect to some vertex $x_{0}$)
then the diameter of the tree created by the hierarchical petal decomposition
is bounded by $8\Delta$. In lines 15 and 24 of the $\texttt{petal-decomposition}$
procedure we divide the weight of some edges by 2. By \factref{Fact PetalDec reduce wght once}
it can happen to any edge at most once, so in the analysis
of the distortion we will ignore this factor (for simplicity).

For completeness, we present the full petal decomposition algorithm. We illustrate below the main changes compared to the version presented in \cite{AN12}.
\begin{enumerate}
\item The order of choosing targets in $\texttt{petal-decomposition}$.
In the original algorithm the targets were chosen arbitrarily, while
in our version we first choose terminals, and only when there are
no terminals left sufficiently far from $x_0$, we choose the other vertices. The purpose of this
ordering, is to ensure that only the terminal petals may
cut certain balls around the other terminals.
\item The segments sent to create petal procedure. In line 22 in $\texttt{petal-decomposition}$ 
    we send an interval of a different size from the original algorithm. The purpose of this change is to ensure that petals which are created in line 22 are far away from any non-covered terminal.
\item We stop constructing petals when all points outside $B_X(x_0,7\Delta/8)$ have been clustered (rather than $B_{X}(x_{0},3\Delta/4)$). The purpose is to ensure that small balls around terminals in $X_0$ are not cut.
\item The choice of the radius in the $\texttt{create-petal}$ procedure
is done by sampling according to a truncated exponential distribution with parameter $\lambda$. The parameter is computed using the density of terminals or other vertices around the target depending on whether the petal will contain a terminal.
\item
The interval in which we choose the radius is determined using the $\texttt{Reform-Interval}$
procedure. The purpose of this is to ensure that certain balls around terminals may be cut only at levels when some terminals are separated (note that there can be at most $k$ such levels). The new interval returned is guaranteed to be contained in the original one.
\end{enumerate}

\begin{algorithm}[h]
\caption{$T=\texttt{hierarchical-petal-decomposition}(G[X],x_{0},t,K)$}
\begin{algorithmic}[1]
\IF {$|X|=1$}
\RETURN $G[X]$.
\ENDIF
\STATE Let $(X_{0},\dots,X_{s},(y_{1},x_{1}),\dots,(y_{s},x_{s}),t_{0},\dots,t_{s})=\texttt{petal-decomposition}(G[X],x_{0},t,K)$;
\FOR {each $j\in[0,\dots,s]$}
\STATE $T_{j}=\texttt{hierarchical-petal-decomposition}(G[X_{j}],x_{j},t_{j},K\cap X_{j})$;
\ENDFOR
\STATE Let $T$ be the tree formed by connecting $T_{0},\dots,T_{s}$ using the edges $\{y_{1},x_{1}\},\dots,\{y_{s},x_{s}\}$;
\end{algorithmic}

\end{algorithm}

\begin{algorithm}[h]
\caption{$(X_{0},\dots,X_{s},\{y_{1},x_{1}\},\dots,\{y_{s},x_{s}\},t_{0},
\dots,t_{s})=\texttt{petal-decomposition}(G[X],x_{0},t,K)$}
\begin{algorithmic}[1]
\STATE Let $\Delta=\Delta_{X,x_0}$; Let $Y_0=X$; Set $j=1$;
\IF {$d_{X}(x_{0},t)\ge5\Delta/8$}
\STATE Let $(X_{1},x_{1})=\texttt{create-petal}(X,t,x_{0},[d_{X}(x_{0},t)-5\Delta/8,d_{X}(x_{0},t)-\Delta/2],K)$;
\STATE Let $Y_{1}=Y_{0}\setminus X_{1}$;
\STATE Let $y_{1}$ be the neighbor of $x_{1}$ on $P_{x_{0}t}$ (the one closer to $x_{0}$);
\STATE Set $t_{0}=y_{1}$, $t_{1}=t$; $j=2$;
\ELSE
\STATE set $t_{0}=t$.
\ENDIF
\STATE {\em Creating the terminal petals:}
\WHILE {$K\cap Y_{j-1}\setminus B_{X}(x_{0},3\Delta/4)\neq\emptyset$}
\STATE Let $t_{j}\in Y_{j-1}\cap K$ be an arbitrary terminal satisfying $d_{X}(x_{0},t_{j})>3\Delta/4$;
\STATE Let $(X_{j},x_{j})=\texttt{create-petal}(Y_{j-1},t_{j},x_{0},[0,\Delta/8],K)$;
\STATE $Y_{j}=Y_{j-1}\setminus X_{j}$;
\STATE For each edge $e\in P_{x_{j}t_{j}}$, set its weight to be $w(e)/2$;
\STATE Let $y_{j}$ be the neighbor of $x_{j}$ on $P_{x_{0}t_{j}}$ (the one closer to $x_{0}$);
\STATE Let $j=j+1;$
\ENDWHILE

\STATE {\em Creating the non-terminal petals:}
\WHILE {$Y_{j-1}\setminus B_{X}(x_{0},7\Delta/8)\neq\emptyset$}
\STATE Let $t_{j}\in Y_{j-1}$ be an arbitrary vertex satisfying $d_{X}(x_{0},t_{j})>7\Delta/8$;
\STATE Let $(X_{j},x_{j})=\texttt{create-petal}(Y_{j-1},t_{j},x_{0},[0,\Delta/32],\emptyset)$;
\STATE $Y_{j}=Y_{j-1}\setminus X_{j}$;
\STATE For each edge $e\in P_{x_{j}t_{j}}$, set its weight to be $w(e)/2$;
\STATE Let $y_{j}$ be the neighbor of $x_{j}$ on $P_{x_{0}t_{j}}$ (the one closer to $x_{0}$);
\STATE Let $j=j+1;$
\ENDWHILE
\STATE Let $s=j-1$;
\STATE {\em Creating the stigma $X_{0}$:}
\STATE Let $X_{0}=Y_{s}$;
\end{algorithmic}
\end{algorithm}

\begin{algorithm}[h]
\caption{$(W,x)=\texttt{create-petal}(X,Y,t,x_{0},[lo,hi],K)$}
\begin{algorithmic}[1]
\STATE $W_{r}=\bigcup_{p\in P_{x_{0}t}:\ d_{Y}(p,t)\le r}B_{(Y,\rho(Y,x_{0},p))}(p,(r-d_{Y}(p,t))/2)$;
\IF {$K\neq\emptyset$}
    \STATE $[lo,hi]=\texttt{Reform-Interval}(Y,t,x_{0},[lo,hi],K)$;
    \STATE $L_k=\lceil\log\log k\rceil$; $R=lo-hi$;
    \STATE Let $1\le q\le L_k$ be the minimal integer satisfying $\left|W_{lo+qR/L_k}\right|_{k}\le\frac{2\left|X\right|_{k}}{2^{\log^{1-q/L_{k}}k}}$; Set $a=lo+(q-1)R/L_k$, $b'=a+R/2L_k$, $\chi=\frac{|X|_k+1}{\left|W_{a}\right|_{k}}$, $\hat\chi=\max\{\chi,e\}$;
    \STATE Set $\lambda=(2\ln\hat{\chi})/(b'-a)=(4L_{k}\ln\hat{\chi})/R=(160L_{k}\ln\hat{\chi})/\Delta$;
\ELSE
    \STATE $L_n=\lceil\log\log n\rceil$; $R=lo-hi$;
    \STATE Let $1\le q\le L_n$ be the minimal integer satisfying $\left|W_{lo+qR/L}\right|\le\frac{2\left|X\right|_{n}}{2^{\log^{1-q/L_{n}}n}}$; Set $a=lo+(q-1)R/L_n$, $b'=a+R/2L_n$, $\chi=\frac{|X|+1}{\left|W_{a}\right|}$, $\hat\chi=\max\{\chi,e\}$;
    \STATE Set $\lambda=(2\ln\hat{\chi})/(b'-a)=(4L_{n}\ln\hat{\chi})/R=(2^{7}L_{n}\ln\hat{\chi})/\Delta$;
\ENDIF

\STATE Sample $r\in\left[a,b'\right]$ according to the truncated exponential distribution with parameter $\lambda$, which has the following density function:
$f\left(r\right)=\frac{\lambda\cdot e^{-\lambda r}}{e^{-\lambda\cdot a}-e^{-\lambda\cdot b'}}$;
\STATE Let $r'\le r$ be the maximum value such that there exists a point $p_{r'}$
of distance $r'$ from $t$ on $P_{x_{0}t}$;
\RETURN $(W_{r},p_{r'})$;
\end{algorithmic}
\end{algorithm}

\begin{algorithm}[h]
\caption{$[lo',hi']=\texttt{Reform-Interval}(Y,t,x_{0},[lo,hi],K)$}
\begin{algorithmic}[1]
\STATE Let $W_{r}=\bigcup_{p\in P_{x_{0}t}:\ d_{Y}(p,t)\le r}B_{(Y,\rho(Y,x_{0},p))}(p,(r-d_{Y}(p,t))/2)$;
\STATE Set $\alpha\leftarrow hi-lo$;
\IF {$K\subseteq W_{lo+3\alpha/5}$}
\RETURN $[lo+4\alpha/5,hi]$;
\ENDIF
\IF {$K\cap W_{lo+2\alpha/5}=\emptyset$}
\RETURN $[lo,lo+\alpha/5]$;
\ENDIF
\RETURN $[lo+2\alpha/5,lo+3\alpha/5]$;
\end{algorithmic}

\end{algorithm}

\subsection{Analysis}

In what follows, fix a cluster $X$ with center $x_0$, a target $t$, and the set of terminals $K$ contained in $X$. Let $(X_{0},\dots,X_{s},\{y_{1},x_{1}\},\dots,\{y_{s},x_{s}\},t_{0},\dots,t_{s})$ be the result of applying \texttt{petal-decomposition}$(G[X],x_0,t,K)$.
The following claim will be useful for bounding the number of levels in which a relevant ball
around a terminal is in danger of being cut.

\begin{claim}
\label{Clm reform effect}
The following assertions hold true:
\begin{itemize}
\item If $X_1$ is a petal created at line 3 of $\texttt{petal-decomposition}$ and $K\cap X_1=\emptyset$, then for all $v\in K$, $B_X(v,\Delta/160)\cap X_1=\emptyset$.
\item If $X_j$ is a petal created at lines 3 or 13 of $\texttt{petal-decomposition}$ and $K\subseteq X_j$, then for all $v\in K$, $B_X(v,\Delta/160)\subseteq X_j$.
\item If $X_j$ is a petal created at line 22 of $\texttt{petal-decomposition}$, then for all $v\in K\cap Y_{j-1}$, $B_X(v,\Delta/16)\cap X_j=\emptyset$.
\end{itemize}
\end{claim}
\begin{proof}
Let $W_r$ be as defined in line 1 of \texttt{create-petal} when forming the petal $X_j$.
For the first assertion, note that the condition of line 6 in \texttt{reform-interval} must be satisfied, no terminals are in $W_{lo+2\alpha/5}$, and thus $hi'$ is set to $lo+\alpha/5$ (we use $\alpha$ as in \texttt{reform-interval}, the size of the interval sent by \texttt{create-petal}.) Now, for every $u\in X_1\subseteq W_{hi'}$ and $v\in K$, $d_X(u,v)\ge \alpha/20 = \Delta/160$, as otherwise, \claimref{ClmW_rProp} implies that $v\in W_{hi'+4\cdot \alpha/20} = W_{lo+2\alpha/5}$, contradiction.

For the second assertion, observe that $j\in\{1,2\}$, because $X_2$ certainly contains a terminal since we are at the "creating terminal petals" stage. As $K\subseteq X_j$, it must be that the condition of line 3 in \texttt{reform-interval} is satisfied, and in line 4 we set $lo'$ to be $lo+4\alpha/5$.
We conclude that if $u\in B_X(v,\Delta/160)$ for some $v\in K$, then $d_{Y_{j-1}}(u,v)\le\Delta/160$ (because by the first assertion, this ball around $v$ was not cut by the first petal). By \claimref{ClmW_rProp}, $u\in W_{lo+3\alpha/5+\Delta/40}= W_{lo+4\alpha/5}\subseteq X_j$.

Finally we prove the third assertion. By the termination condition in line 11 of \texttt{petal-decomposition}, it must be that $v\in B_X(x_0,3\Delta/4)$. By \factref{FactPetals dont cut shortest path to x0}, all distances in $Y_{j-1}$ to $x_0$ remain the same as in $X$. Note that the parameter $r$ of the petal $X_j=W_r$ is at most $\Delta/32$, and that $d_X(x_0,t_j)\ge 7\Delta/8$. By the definitions of a petal and of the cone-metric, for any $p\in P_{x_0t}$ with $d_X(p,t)\le\Delta/32$ we have that if $u\in B_{(Y_{j-1},\rho(Y_{j-1},x_0,p))}(p,(r-d_{Y_{j-1}}(p,t))/2)$ then
\begin{eqnarray*}
  d_X(x_0,u) &\ge& d_X(x_0,p)+d_{Y_{j-1}}(p,u)-(r-d_{Y_{j-1}}(p,t))/2> d_X(x_0,p)-r \\
  &\ge& d_X(x_0,t)-d_X(p,t)-\Delta/32 \ge 7\Delta/8-2\Delta/32\ge 13\Delta/16~.
\end{eqnarray*}
By the triangle inequality we obtain that $d_X(v,u)\ge d_X(x_0,u)-d_X(x_0,v)>\Delta/16$, which concludes the proof.
\end{proof}

By \claimref{Clm reform effect}, given a terminal $v$ in cluster $X$ of radius $\Delta$, and a ball $B=B_X(v,l)$ of radius at most $l\le\Delta/160$, the ball $B$ might be cut in the petal decomposition of $X$ only if the terminal $v$ is separated from some other terminal. \texttt{petal-decomposition} picks terminal targets arbitrary but deterministically. Therefore we will be more specific, and decide that \texttt{petal-decomposition} pick terminal targets by increasing order of the remaining terminals. If the first petal is built in line 3, then the first target is given, and by the \texttt{reform-interval} procedure, it is deterministically determined whether the first petal will contain all, part, or none of the terminals. If there are terminals remaining, and the first special petal (if exists) contains none of the terminals, then the second target is once again picked deterministically, and again \texttt{reform-interval} deterministically determines whether the petal will contain all or only part of the terminals. In any case, the issue of whether all the terminals will be in single petal or not, is resolved deterministically.
Hence, given a cluster $X$, root $x_0$ and target $t$, it is known if there is some danger that $B$ may be cut. Moreover, there are at most $k$ levels in which some terminals are separated,
and hence involving a danger of cutting balls around terminals. (If $l>\Delta/160$ then we do not care whether $B$ is cut or not, because the distortion induced on $v$ and points outside of $B$ will be a constant.)

\subsection{Expected Distortion}

When dealing with weighted graphs, it could be that a tiny ball participates in many recursive levels, and thus is "threatened" many times. In order to avoid such situations, we shall change the algorithm slightly, and when performing a petal decomposition on a cluster $X$ of radius $\Delta$, we shall contract for each vertex $u\in X$, the ball of radius $\Delta/n^3$ around $u$. In addition, if terminals are separated (which is determined deterministically) then for each terminal $v\in K$, we contract the ball of radius $\Delta/k^3$ around $v$. These contractions are done sequentially; when a vertex contracts a ball of radius $l$ it becomes a "supernode" of all the vertices within distance $l$, and for any edge leaving a vertex in the ball to some vertex $z$, we will have a corresponding edge with the same weight from the supernode to $z$.
These contractions yield a cluster $X'$, on which we run the \texttt{petal-decomposition} algorithm (note that $X'$ may contains supernodes - vertices that correspond to several original vertices, and that $X'$ induces a multi-graph). After the partition of $X'$ to $X_0,\dots,X_s$ is determined, we expand back the contracted balls, so that each vertex belongs to the cluster of its supernode.

This guarantees that a ball of radius $l$ around each terminal can participate only at partitions of radii in the range $[l,k^3\cdot l]$, and by \factref{FactPetal3/4Radius}, the radii decrease by a factor of at least $7/8$ every level, so there are at most $16\log k$ levels in which each such ball participates. Similarly, a ball of radius $l'$ around some vertex $u\in X$, can participate only at partitions of radii in the range $[l',n^3\cdot l']$, so there are at most $16\log n$ levels in which each such ball participates.
This contraction, while saving small balls from being cut, may have an effect on the radius of the tree when we expand back the vertices. We claim that this will be a minor increase. To see this, note that in a particular level, expanding back the balls around terminals can increase distances by at most $2\Delta/k^2$ (because every contracted ball has diameter at most $2\Delta/k^3$, and there can be at most $k$ contracted balls). Similarly, expanding back all the other contracted balls may increase distances by at most an additional term of $2\Delta/n^2$.
Now, since there are at most $k$ iterations in which terminals are separated (only in such levels the balls around them are contracted), even if the radius is increased by a factor of $1+2/k^2$ in every one of them, the total increase is at most $(1+2/k^2)^k<e^{2/k}$. Similarly, there are at most $n$ iterations all in all (since $\emptyset\subsetneq X_0\subsetneq X$), and even if the radius is increased by a factor of $1+2/n^2$ in every one of them, the total increase is at most $(1+2/n^2)^n<e^{2/n}$. Henceforth we shall ignore this minor increase, as it affects the distortion of every pair by at most a factor of $e^{2/k+2/n}$.

\sloppy We will show that the distribution generated by the \texttt{hierarchical-petal-decomposition}$(G,x_0,x_0,K)$ algorithm has strong terminal distortion $(\tilde{O}(\log k),\tilde{O}(\log n))$ (where $x_0$ is an arbitrary vertex).
The hierarchical partition naturally corresponds to a laminar family of clusters, that are arranged in a hierarchical tree structure denoted by ${\cal T}$ (with $V$ as the root, and each cluster $X$ which is partitioned by \texttt{petal-decomposition} to $X_0,\dots,X_s$, has them as its children in ${\cal T}$). The {\em level} of a cluster is its distance in ${\cal T}$ to the root. Note that there might be several trees that correspond to a single hierarchical partition (this depends on the precise edges connecting the different petals).
For a pair of vertices ${x,y}$ and hierarchical tree $\mathcal{T}$, let $d_\mathcal{T}(x,y)$ be the maximal distance between $x$ and $y$ in a tree $T$ that corresponds to the hierarchical tree $\mathcal{T}$. Note that if $x$ is separated from $y$ in the hierarchical tree $\mathcal{T}$ in cluster $X$ of radius $\Delta_X$, then by \factref{FactPetalTreeRadiusBound4}, $d_\mathcal{T}(x,y)\le 8\Delta_X$.

Fix any two points $x,y\in V$ and denote by $B_x=B_G(x,d(x,y))$. For $j\in[s]$, 
we say that $B_x$ is {\em cut} by $X_j$ if $B_x\cap X_j\notin\{\emptyset,B_x\}$. Further, $B_x$ is cut in $X$ if there exists some petal $X_j$ that cuts $B_x$. 
For $i\ge0$, let $S_{X,i}$ be the event that $B_x\subseteq X$ and $X$ is a cluster in level $i$. When calling \texttt{petal-decomposition} on the cluster $X$, a sequence of clusters $X_1,\dots,X_s,X_0$ is generated, with $Y_j$ as defined in the algorithm. For $Y\subseteq X$ and integer $j\ge 1$, let $S_{X,i,Y,j}$ denote the event that $S_{X,i}$ holds, $B_x\subseteq Y$, and $Y=Y_{j-1}$.
Define the following events
\begin{alignat*}{2}
{\cal C}_{X,i,j}&=\{B_x\cap X_j\notin\{\emptyset,B_x\}\} & \qquad &
\text{``$B_x$ cut in level $i$ at iteration $j$"}, \\
{\cal F}_{X,i,j}&=\{B_x\cap X_j=\emptyset\} & \qquad &
\text{``$B_x$ is not cut in level $i$ iteration $j$''}.
\end{alignat*}
Denote by ${\cal F}_{X,i,(< j)}$ the event $\bigwedge_{0<h<j}{\cal F}_{X,i,h}$, by $\bar{\cal F}_{X,i,j}=\{B_x\cap X_j\neq\emptyset\}$ the event that $B_x$ is either cut or contained in $X_j$, and let
\[
{\cal E}_{X,i,j} = \{{\cal C}_{X,i,j} \land {\cal F}_{X,i,(<j)}\} \qquad
\text{``$B_x$ is cut in level $i$ iteration $j$, but not before''}.
\]
Finally, define ${\cal E}_{X,i}=\left\{{\cal S}_{X,i}\wedge\bigcup_j{\cal E}_{X,i,j}\right\}$, as the event that $B_x$ is cut for the first time in level $i$.
With our new notations, we can bound the expected distance, as follows:

\begin{eqnarray}\label{eq:main}
\E[d_T(x,y)]
&\le&\sum_{{\cal T}}\Pr[{\cal T}]\cdot d_{\cal T}(x,y)\nonumber\\\nonumber
&=&\sum_{{\cal T}}\sum_{i\in\N}\sum_{X\subseteq V}\Pr[{\cal E}_{X,i}]\cdot\Pr[{\cal T}\mid{\cal E}_{X,i}]\cdot
d_{\cal T}(x,y)\\\nonumber
&=&\sum_{X,i}\Pr[{\cal E}_{X,i}]\sum_{\cal T}\Pr[{\cal T}\mid{\cal E}_{X,i}]\cdot
d_{\cal T}(x,y)\\\nonumber
&\le&8\sum_{X,i}\Pr[{\cal E}_{X,i}]\cdot\Delta_X\cdot\sum_{\cal T}\Pr[{\cal T}\mid{\cal E}_{X,i}]\\
&=&8\sum_{X,i}\Pr[{\cal E}_{X,i}]\cdot\Delta_X~.
\end{eqnarray}

Recall the constants defined in \texttt{create-petal}: $L_{k}=\left\lceil \log\log k\right\rceil$ and $L_{n}=\left\lceil \log\log n\right\rceil$. For $Y\subseteq X\subseteq V$ define $\varphi_k(X,Y,x)=\max\left\{\frac{|X|_k}{\left|B_{Y}\left(x,
\Delta_{X}/(2^{8}L_{k})\right)\right|_{k}},e\right\}$, and similarly $\varphi_n(X,Y,x)=\max\left\{\frac{|X|}{\left|B_{Y}\left(x,
\Delta_{X}/(2^{8}L_{n}\right))\right|},e\right\}$.  The following lemma bounds the probability that a ball around a terminal or an arbitrary vertex, is cut. We defer its proof to \sectionref{sec:lemma-4}.

\begin{lemma}
\label{lem:Cut Prob} For any integer $i$, and a cluster $X$ with radius $\Delta_X$ (from some root vertex), it holds that if $x$ is a terminal ($x\in K$) then (where $x,y$ are as defined earlier)
\begin{equation}\label{Cut Prob terminal}
\Pr[{\cal E}_{X,i}]\Delta_{X}\le2^{14}d(x,y)L_{k}\sum_{j\le k+1}\sum_{Y\subseteq X}\Pr\left[\mathcal{S}_{X,i,Y,j}\wedge\overline{\mathcal{F}}_{X,i,j}\right]\ln\left(\varphi_k\left(X,Y,x\right)\right)~.
\end{equation}
Moreover, without dependence wether $x$ is a terminal, it holds that
\begin{eqnarray}\label{Cut Prob non terminal}
\Pr[{\cal E}_{X,i}]\Delta_{X}
&\le&2^{14}d(x,y)L_{k}\sum_{j\le k+1}\sum_{Y\subseteq X}\Pr\left[\mathcal{S}_{X,i,Y,j}\wedge\overline{\mathcal{F}}_{X,i,j}\right]\ln\left(\varphi_k\left(X,Y,x\right)\right)\nonumber\\
&&~+ 2^{14}d(x,y) L_{n}\sum_{j}\sum_{Y\subseteq X}\Pr\left[\mathcal{S}_{X,i,Y,j}\wedge\overline{\mathcal{F}}_{X,i,j}\right]\ln\left(\varphi_n\left(X,Y,x\right)\right).
\end{eqnarray}
\end{lemma}

The algorithm returns us a hierarchical tree ${\cal T}$. We refer to the base level (where there are only one cluster that contains all the vertices) as level $0$, to the next level as level $1$ and so on.
For a hierarchical tree ${\cal T}$, let $A_{\cal T}\subseteq \N$ be the set of levels $i$, in which $B_x$ is included in a cluster $X$ with radius $\Delta_X$, such that $d(x,y)$ is larger than $\Delta_X/k^{3}$. Note that if $x$ is a terminal, due to the contractions, if $d(x,y)\le\Delta_X/k^{3}$, then the ball $B_x$ could not be cut, hence ${\cal E}_{X,i}=\emptyset$.
By \factref{FactPetal3/4Radius}, $A_{\cal T}\subseteq \N$ is of size at most $\log_{8/7}(2\cdot k^{3})+1< 22\log k$. (The radius has to be in the range $[d(u,v)/2,k^{3}d(u,v)]$, once the radius is less than $d(u,v)/2$, $B_x$ will surely be cut.) For any level $i$, denote by $A_{{\cal T},i}=\{i'\in A_{\cal T}~:~i'\ge i\}$.
Similarly, let $Q_{{\cal T}}$ be the set of levels $i$, in which $B_x$ is included in a cluster $X$ with diameter at most $\Delta_X\le n^3\cdot d(x,y)$. For every vertex $x$ due to the contraction, if $d(x,y)\le\Delta_X/n^{3}$, ${\cal E}_{X,i}=\emptyset$. By \factref{FactPetal3/4Radius}, $Q_{\cal T}\subseteq \N$ is of size at most $22\log n$. For any level $i$, denote by $Q_{{\cal T},i}=\{i'\in Q_{\cal T}~:~i'\ge i\}$.

Define $l_k=\left\lceil \log_{8/7}(2^{9}L_{k})\right\rceil$, and $l_n=\left\lceil \log_{8/7}(2^{9}L_{n})\right\rceil$. Consider the partition of the indices of all the levels in the hierarchical tree to the sets $I_m=\{i\in\N~:~ i=m~(\text{mod } l_k)\}$ for $m\in\{0,1,\dots,l_k-1\}$. Similarly, for $g\in\{0,1,\dots,l_n-1\}$, $J_g=\{i\in\N~:~ i=g~(\text{mod } l_n)\}$. The next lemma, combined with \lemmaref{lem:Cut Prob}, is used to bound the cut probability.

\begin{claim}\label{Clm bound on the bound of cut}
For every positive integers $m,g$ and $t$, it holds that:
\begin{equation}\label{eq:indd}
\sum_{i\in I_{m}:~i\ge t}\sum_{X,j,Y}\Pr\left[\mathcal{S}_{X,i,Y,j}\wedge\overline{\mathcal{F}}_{X,i,j}\right]\ln\left(\varphi_{k}\left(X,Y,x\right)\right)\le\sum_{X,{\cal T}}\Pr[{\cal S}_{X,t}\wedge{\cal T}]\left(\ln|X|_{k}+|A_{{\cal T},t}|\right)
\end{equation}

\begin{equation}\label{eq:indd2}
\sum_{i\in J_{g}:~i\ge t}\sum_{X,j,Y}\Pr\left[\mathcal{S}_{X,i,Y,j}\wedge\overline{\mathcal{F}}_{X,i,j}\right]\ln\left(\varphi_{n}\left(X,Y,x\right)\right)\le\sum_{X,{\cal T}}\Pr[{\cal S}_{X,t}\wedge{\cal T}]\left(\ln|X|+|Q_{{\cal T},t}|\right)
\end{equation}
\end{claim}

\begin{proof}
Fix any such $m$, we will show \eqref{eq:indd} by (reverse) induction on $t\in I_m$.

For the base case, note that when $t$ is sufficiently large, $x$ and $y$ must have been separated. E.g. if $t=\log_{8/7}(2\diam(G)/d(x,y))$, at levels $i\ge t$ the radius of any cluster will be less then $d(x,y)/2$ (using that the radius drops by at least $7/8$ at every level). We get that $\Pr[{\cal E}_{X,i}]=0$ for any possible $X$ and $i\ge t$. Hence $\sum_{i\in I_{m}:~i\ge t}\sum_{X,j,Y}\Pr\left[\mathcal{S}_{X,i,Y,j}\wedge\overline{\mathcal{F}}_{X,i,j}\right]=0$ and  \eqref{eq:indd} holds.

Assume \eqref{eq:indd} holds for $t+l_k\in I_m$, and prove for $t$. Observe that if $Z$ is a cluster in level $t+l_k$ such that event ${\cal S}_{Z,t+l_k}$ holds, then there must be a cluster $X$ at level $t$ (the ancestor of $Z$) such that ${\cal S}_{X,t}$ holds, and there are also $Y$ and $j$ such that $Y=Y_{j-1}$, event $\bar{\cal F}_{X,i,j}$ holds, and $Z\subseteq Y\subseteq X$ ($X_j$ is the ancestor of $Z$ at level $t+1$). Also note that $\Delta_{X}\ge (8/7)^{l_k}\Delta_{Z}=2^9L_k \cdot\Delta_{Z}$, therefore $Z\subseteq B_{Z}(x,2\Delta_{Z})\subseteq B_{Y}(x,\Delta_{X}/(2^{8}L_{k}))$. It follows that
\begin{eqnarray}\label{eq:eww}
\lefteqn{\sum_{Z,{\cal T}}\Pr[{\cal S}_{Z,t+l_{k}}\wedge{\cal T}]\left(\ln|Z|_{k}+|A_{{\cal T},t+l_{k}}|\right)}\nonumber\\
&=&\sum_{{\cal T},X,Y,Z,j~:~Z\subseteq Y\subseteq X}\Pr[{\cal S}_{X,t,Y,j}\wedge\bar{{\cal F}}_{X,t,j}\wedge{\cal T}]\cdot\Pr[{\cal S}_{Z,t+l_{k}}\mid{\cal S}_{X,t,Y,j}\wedge\bar{{\cal F}}_{X,t,j}\wedge{\cal T}]\cdot\nonumber\\
&&~~~~~~~~~~~~~~~\cdot\left(\ln|Z|_{k}+|A_{{\cal T},t+l_{k}}|\right)\nonumber\\
&\le&\sum_{{\cal T},X,Y,Z,j~:~Z\subseteq Y\subseteq X}\Pr[{\cal S}_{X,t,Y,j}\wedge\bar{{\cal F}}_{X,t,j}\wedge{\cal T}]\cdot\nonumber\\
&& ~~~~~~~~~~~~~~~\cdot\Pr[{\cal S}_{Z,t+l_{k}}\mid{\cal S}_{X,t,Y,j}\wedge\bar{{\cal F}}_{X,t,j}\wedge{\cal T}]\left(\ln|B_{Y}(x,\frac{\Delta_{X}}{2^{8}L_{k}})|_{k}+|A_{{\cal T},t+l_{k}}|\right)\nonumber\\
&\le&\sum_{{\cal T},X,Y,j: Y\subseteq X}\Pr[{\cal S}_{X,t,Y,j}\wedge\bar{{\cal F}}_{X,t,j}\wedge{\cal T}]\left(\ln|B_{Y}(x,\frac{\Delta_{X}}{2^{8}L_{k}})|_{k}+|A_{{\cal T},t+l_k}|\right)~.
\end{eqnarray}

Where the last inequality is by changing the order of summation and the fact that $\sum_{Z\subseteq Y}\Pr[{\cal S}_{Z,t+l_k}\mid{\cal S}_{X,t,Y,j}\wedge\bar{{\cal F}}_{X,t,j}\wedge{\cal T}]\le 1$.
Now, if in the hierarchical tree ${\cal T}$, $t\in A_{{\cal T}}$, then $|A_{{\cal T},t}|\ge |A_{{\cal T},t+l_k}|+1$, so that
\begin{eqnarray}\label{eq:kkp}
\lefteqn{\sum_{i\in I_{m},~i\ge t}\sum_{X,j,Y}\Pr\left[\mathcal{S}_{X,i,Y,j}\wedge\overline{\mathcal{F}}_{X,i,j}\right]\ln\left(\varphi_k\left(X,Y,x\right)\right)}\nonumber\\
&=&\sum_{i\in I_{m},~i\ge t+l_{k}}\sum_{X,j,Y}\Pr\left[\mathcal{S}_{X,i,Y,j}\wedge\overline{\mathcal{F}}_{X,i,j}\right]\ln\left(\varphi_k\left(X,Y,x\right)\right)\nonumber\\
&&+~~\sum_{X,j,Y}\Pr\left[\mathcal{S}_{X,t,Y,j}\wedge\overline{\mathcal{F}}_{X,t,j}\right]\ln\left(\varphi_k\left(X,Y,x\right)\right)\nonumber\\
&\stackrel{\eqref{eq:indd}}{\le}& \sum_{X,{\cal T}}\Pr[{\cal S}_{X,t}\wedge{\cal T}]\left(\ln|X|_{k}+|A_{{\cal T},t+l_k}|\right) + \!\!\!\! \sum_{X,j,Y_{j-1}}\!\!\!\!\Pr\left[\mathcal{S}_{X,t,Y,j}\wedge\overline{\mathcal{F}}_{X,t,j}\right]\ln\left(\varphi_k\left(X,Y_{j-1},x\right)\right)\nonumber\\
&\stackrel{\eqref{eq:eww}}{\le}&\sum_{{\cal T},X,Y,j: Y\subseteq X}\Pr[{\cal S}_{X,t,Y,j}\wedge\bar{{\cal F}}_{X,t,j}\wedge{\cal T}]\left(\ln|B_{Y}(x,\frac{\Delta_{X}}{2^{8}L_{k}})|_{k}+|A_{{\cal T},t+l_k}|\right)\nonumber\\
&&+~~\sum_{{\cal T},X,Y,j}\Pr[{\cal S}_{X,t,Y,j}\wedge\bar{{\cal F}}_{X,t,j}\wedge{\cal T}]\ln{\varphi_{k}}(X,Y,x)\nonumber\\
&\le&\sum_{{\cal T},X,Y,j}\Pr[{\cal S}_{X,t,Y,j}\wedge\bar{{\cal F}}_{X,t,j}\wedge{\cal T}]\left(\ln|X|_{k}+|A_{{\cal T},t}|\right)~.
\end{eqnarray}

In the last inequality we used that if $\varphi_{k}(X,Y,x)=\frac{|X|_{k}}{|B_{Y}(x,\Delta_{X}/2^{8}L_{k})|_k}$, then
\[
\ln|B_{Y}(x,\Delta_{X}/2^{8}L_{k})|_{k}+|A_{{\cal T},t+l}|+\ln\varphi_{k}(X,Y,x)=\ln|X|_{k}+|A_{{\cal T},t+l}|\le\ln|X|_{k}+|A_{{\cal T},t}|~,
\]
and if $\varphi_k(X,Y,x)=e$, then
\begin{eqnarray*}
\ln|B_{Y}(x,\Delta_{X}/2^{8}L)|_{k}+|A_{{\cal T},t+l}|+\ln\varphi_{k}(X,Y,x)
&=&\ln|B_{Y}(x,\Delta_{X}/2^{8}L|_{k}+|A_{{\cal T},t+l}|+1\\
&\le&\ln|X|_k+|A_t|~.
\end{eqnarray*}

In addition note that
\begin{eqnarray}\label{eq:Change index in summution}
\lefteqn{\sum_{{\cal T},X,Y,j}\Pr[{\cal S}_{X,t,Y,j}\wedge\bar{{\cal F}}_{X,t,j}\wedge{\cal T}]\left(\ln|X|_{k}+|A_{{\cal T},t}|\right)}\nonumber\\
&=& \sum_{{\cal T},X,Y,j}\Pr[{\cal S}_{X,t}\wedge{\cal T}]\cdot\Pr[{\cal S}_{X,t,Y,j}\wedge\bar{{\cal F}}_{X,t,j}\mid{\cal S}_{X,t}\wedge{\cal T}]\left(\ln|X|_{k}+|A_{{\cal T},t}|\right) \nonumber\\
&=& \sum_{{\cal T},X}\Pr[{\cal S}_{X,t}\wedge{\cal T}]\left(\ln|X|_{k}+|A_{{\cal T},t}|\right)\sum_{Y,j}\Pr[{\cal S}_{X,t,Y,j}\wedge\bar{{\cal F}}_{X,t,j}\mid{\cal S}_{X,t}\wedge{\cal T}] \nonumber\\
&=& \sum_{{\cal T},X}\Pr[{\cal S}_{X,t}\wedge{\cal T}]\left(\ln|X|_{k}+|A_{{\cal T},t}|\right)
\end{eqnarray}
where the last equality follows by $\sum_{Y,j}\Pr[{\cal S}_{X,t,Y,j}\wedge\bar{{\cal F}}_{X,t,j}\mid{\cal S}_{X,t}\wedge{\cal T}]=1 $, which holds because for every different $Y,j$, the events ${\cal S}_{X,t,Y,j}\wedge\bar{{\cal F}}_{X,t,j}$ are disjoint.
This concludes the proof of \eqref{eq:indd}.

The proof of \eqref{eq:indd2} is fully symmetric. Fix some g, we will show \eqref{eq:indd2} by (reverse) induction on $t\in J_g$. The base case is trivial (for every $X$ and $t\ge\log_{8/7}(2\diam(G)/d(x,y))$, $\Pr[{\cal S}_{X,i}]=0$ and therefore $\sum_{i\in J_{g}:~i\ge t}\sum_{X,j,Y}\Pr\left[\mathcal{S}_{X,i,Y,j}\wedge\overline{\mathcal{F}}_{X,i,j}\right]=0$.)
Assume \eqref{eq:indd2} holds for $t+l_n\in J_g$, and prove for $t$. For cluster $Z$ such that event ${\cal S}_{Z,t+l_n}$ holds, there are  unique clusters $X,Y$ and an index $j$, such that the events ${\cal S}_{X,t},\bar{{\cal F}}_{X,i,j}$ holds, $Y=X_j$, and $Z\subseteq B_{Z}(x,2\Delta_{Z})\subseteq B_{Y}(x,\Delta_{X}/(2^{8}L_{n}))$. Similarly to equation \eqref{eq:eww} we get
\begin{eqnarray}\label{eq:eww2}
\lefteqn{\sum_{Z,{\cal T}}\Pr[{\cal S}_{Z,t+l_{n}}\wedge{\cal T}]\left(\ln|Z|+|Q_{{\cal T},t+l_{n}}|\right)}\nonumber\\
&=&\sum_{{\cal T},X,Y,Z,j~:~Z\subseteq Y\subseteq X}\Pr[{\cal S}_{X,t,Y,j}\wedge\bar{{\cal F}}_{X,t,j}\wedge{\cal T}]\cdot\Pr[{\cal S}_{Z,t+l_{n}}\mid{\cal S}_{X,t,Y,j}\wedge\bar{{\cal F}}_{X,t,j}\wedge{\cal T}]\cdot\nonumber\\
&&~~~~~~~~~~~~~~~~~~\cdot\left(\ln|Z|+|Q_{{\cal T},t+l_{n}}|\right)\nonumber\\
&\le&\sum_{{\cal T},X,Y,j: Y\subseteq X}\Pr[{\cal S}_{X,t,Y,j}\wedge\bar{{\cal F}}_{X,t,j}\wedge{\cal T}]\left(\ln|B_{Y}(x,\frac{\Delta_{X}}{2^{8}L_{n}})|+|Q_{{\cal T},t+l_{n}}|\right)~.
\end{eqnarray}
Now we make the induction step:
\begin{eqnarray*}
\lefteqn{\sum_{i\in J_{g},~i\ge t}\sum_{X}\sum_{j,Y}\Pr\left[\mathcal{S}_{X,i,Y,j}\wedge\overline{\mathcal{F}}_{X,i,j}\right]\ln\left(\varphi_n\left(X,Y,x\right)\right)
}\\
&=& \sum_{i\in J_{g},~i\ge t+l_n}\sum_{X}\sum_{j,Y}\Pr\left[\mathcal{S}_{X,i,Y,j}\wedge\overline{\mathcal{F}}_{X,i,j}\right]\ln\left(\varphi_n\left(X,Y,x\right)\right)\\
&&+~\sum_{X}\sum_{j,Y}\Pr\left[\mathcal{S}_{X,t,Y,j}\wedge\overline{\mathcal{F}}_{X,t,j}\right]\ln\left(\varphi_n\left(X,Y,x\right)\right)\\
&\stackrel{\eqref{eq:indd2}}{\le}&\sum_{X,{\cal T}}\Pr[{\cal S}_{X,t+l_{n}}\wedge{\cal T}]\left(\ln|X|+|Q_{{\cal T},t+l_{n}}|\right)\\
&&+~\sum_{X}\sum_{j,Y}\Pr\left[\mathcal{S}_{X,t,Y,j}\wedge\overline{\mathcal{F}}_{X,t,j}\right]\ln\left(\varphi_n\left(X,Y,x\right)\right)\\
&\stackrel{\eqref{eq:eww2}}{\le}& \sum_{{\cal T},X,Y,j}\Pr[{\cal S}_{X,t,Y,j}\wedge\bar{{\cal F}}_{X,t,j}\wedge{\cal T}]\left(\ln|B_{Y}(x,\frac{\Delta_{X}}{2^{8}L_{k}})|+|Q_{{\cal T},t+l_{n}}|\right)\\
&&+~\sum_{{\cal T},X,Y,j}\Pr\left[\mathcal{S}_{X,t,Y,j}\wedge\overline{\mathcal{F}}_{X,t,j}\wedge{\cal T}\right]\ln\left(\varphi_n\left(X,Y_{j-1},x\right)\right)\\
&\le&\sum_{{\cal T},X,Y,j}\Pr[{\cal S}_{X,t,Y,j}\wedge\bar{{\cal F}}_{X,t,j}\wedge{\cal T}]\left(\ln|X|+|Q_{{\cal T},t}|\right)~.
\end{eqnarray*}
By the same calculation as in equation \eqref{eq:Change index in summution} we have $\sum_{{\cal T},X,Y,j}\Pr[{\cal S}_{X,t,Y,j}\wedge\bar{{\cal F}}_{X,t,j}\wedge{\cal T}]\left(\ln|X|+|Q_{{\cal T},t}|\right)=\sum_{X,{\cal T}}\Pr[{\cal S}_{X,t}\wedge{\cal T}]\left(\ln|X|+|Q_{{\cal T},t}|\right)$, Which concludes the proof of \eqref{eq:indd2}.
\end{proof}

\begin{proof}[Proof of \theoremref{thm:strong-spanning}]
The proof of \theoremref{thm:strong-spanning} follows from \lemmaref{lem:Cut Prob}, \claimref{Clm bound on the bound of cut} and equation \eqref{eq:main}. We start by proving the first assertion of  \theoremref{thm:strong-spanning}. Recall that for every hierarchical tree ${\cal T}$, $|A_{\cal T}|\le 22\log k$. Note that $|X|_k\le k$.

\begin{eqnarray*}
\E[d_T(x,y)]
&\stackrel{\eqref{eq:main}}{\le}& 8\sum_{X,i}\Pr[{\cal E}_{X,i}]\cdot\Delta_X\\
&=& 8\sum_{m=0}^{l_k-1}\sum_{i\in I_{m}}\sum_{X}\Pr[{\cal E}_{X,i}]\cdot\Delta_X\\
&\stackrel{\eqref{Cut Prob terminal}}{\le}&2^{17}L_{k}\cdot d(x,y)\sum_{m=0}^{l_{k}-1}\sum_{i\in I_{m}}\sum_{X,j\le k+1,Y}\Pr\left[\mathcal{S}_{X,i,Y,j}\wedge\overline{\mathcal{F}}_{X,i,j}\right]\ln\left(\varphi_{k}\left(X,Y,x\right)\right)\\
&\stackrel{\eqref{eq:indd}}{\le}&2^{17}d(x,y)L_{k}\sum_{m=0}^{l_{k}-1}\sum_{X,{\cal T}}\Pr[{\cal S}_{X,0}\wedge{\cal T}]\left(\ln|X|_{k}+|A_{{\cal T},0}|\right)\\
&\le& 2^{17}L_{k}\cdot d(x,y)\left(\ln k+22\log k\right)l_{k}\sum_{{\cal T}}\Pr[{\cal S}_{V,0}\wedge{\cal T}]\\
&\le&2^{29}\log k\log\log k\log\log\log k\cdot d(x,y)=\tilde{O}\left(\log k\right)d(x,y)~.
\end{eqnarray*}
Similarly, for the second assertion of \theoremref{thm:strong-spanning}, recall that for every hierarchical tree ${\cal T}$, $|Q_{\cal T}|\le 22\log n$. Note that $|X|\le n$.
\begin{eqnarray*}
\lefteqn{\E[d_T(x,y)]}\\
&\stackrel{\eqref{eq:main}}{\le}& 8\sum_{X,i}\Pr[{\cal E}_{X,i}]\cdot\Delta_X\\
&\stackrel{\eqref{Cut Prob non terminal}}{\le}&8\sum_{X,i}2^{14}d(x,y)L_{k}\sum_{j\le k}\sum_{Y}\Pr\left[\mathcal{S}_{X,i,Y,j}\wedge\overline{\mathcal{F}}_{X,i,j}\right]\ln\left(\varphi_k\left(X,Y,x\right)\right)\nonumber\\
&&+~ 8\sum_{X,i}2^{14}d(x,y)L_{n}\sum_{j}\sum_{Y}\Pr\left[\mathcal{S}_{X,i,Y,j}\wedge\overline{\mathcal{F}}_{X,i,j}\right]\ln\left(\varphi_n\left(X,Y,x\right)\right)\\
&\le&2^{17}L_{k}\cdot d(x,y)\sum_{m=0}^{l_{k}-1}\sum_{i\in I_{m}}\sum_{X,j\le k,Y}\Pr\left[\mathcal{S}_{X,i,Y,j}\wedge\overline{\mathcal{F}}_{X,i,j}\right]\ln\left(\varphi_k\left(X,Y,x\right)\right)\nonumber\\
&&+~ 2^{17}d(x,y)L_{n}\sum_{g=0}^{l_{n}-1}\sum_{i\in J_{g}}\sum_{X,j,Y}\Pr\left[\mathcal{S}_{X,i,Y,j}\wedge\overline{\mathcal{F}}_{X,i,j}\right]\ln\left(\varphi_n\left(X,Y,x\right)\right)\\
&\stackrel{\eqref{eq:indd}\wedge\eqref{eq:indd2}}{\le}&2^{17}d(x,y)L_{k}\sum_{m=0}^{l_{k}-1}\sum_{i\in I_{m}}\sum_{X,{\cal T}}\Pr[{\cal S}_{X,0}\wedge{\cal T}]\left(\ln\left|X\right|_k+\left|A_{{\cal T},0}\right|\right)\\
&&+~2^{17}d(x,y)L_{n}\sum_{g=0}^{l_{n}-1}\sum_{i\in J_{g}}\sum_{X,{\cal T}}\Pr[{\cal S}_{X,0}\wedge{\cal T}]\left(\ln|X|+\left|Q_{{\cal T},0}\right|\right)\\
&\le& \left(2^{29}\log k\log\log k\log\log\log k+2^{29}\log n\log\log n\log\log\log n\right)d(x,y)\\
&=&\tilde{O}\left(\log n\right)d(x,y)~.
\end{eqnarray*}
\end{proof}

\subsubsection{Proof of \lemmaref{lem:Cut Prob}}\label{sec:lemma-4}

Let $X$ be the vertex set of the graph in the current level $i$ of the petal decomposition, with arbitrary center $x_0$, target $t$ and radius $\Delta$ (with respect to $x_0$). Recall the vertices $x,y\in X$, and the ball $B_x=B_G(x,d(x,y))$. Set $\gamma=d(x,y)/\Delta$. 
As $X$ and $i$ are fixed, and all the probabilistic events in the statement of the lemma are contained in ${\cal S}_{X,i}$, we will implicitly condition all the probabilistic events during the proof  on ${\cal S}_{X,i}$ (i.e. our sample space is restricted to ${\cal S}_{X,i}$). We shall also omit $X$ and $i$ from the subscript of the probabilistic events (i.e., we will write ${\cal C}_{j},{\cal F}_{j},{\cal F}_{(< j)},{\cal E}_{j},{\cal E}$ and we will write ${\cal Z}_{Y,j}$ instead of ${\cal S}_{X,i,Y,j}$).

The petal decomposition algorithm returns a partition $(X_{0},\dots,X_{s},(y_{1},x_{1}),\dots,(y_{s},x_{s}),t_{0},\dots,t_{s})$. We make a small change in the numbering of the created petals:
let $t_{1},\dots,t_{k+1}$ be the terminal targets chosen in lines 6 or 12 of \texttt{petal-decomposition}, and $t_{k+2},\dots,t_{n+1}$
be the non-terminal targets chosen in line 21. Observe that in that notation we always have exactly $n+1$ petals, while there might be an index $j$ such that $X_j=\emptyset$, in that case we say that $X_j$ is an imaginary petal. Note that there are at most $k+1$ terminal targets because there are just $k$ terminals
(in addition to the first special petal whose target is not necessarily a terminal).

Recall the definitions from $\texttt{create-petal}$ procedure, some of them depend on the type of petal (terminal or non-terminal petal), and some of them are actually random variables (which depend on the previously created petals). We will write these with an index $j$ to clarify which petal they correspond to.
In the terminal petal case ($j\le k+1$), $L_k=\lceil\log\log k\rceil$, $R_j=hi-lo=\Delta/40$ (after reformation),
$a_j=lo+(q_j-1)R_j/L_k$, $b'_j=a_j+R_j/(2L_k)$, $b_j=a_j+R_j/L_k$, where $q_j$ is chosen in such a way that

 \begin{equation}\label{eq:Wa le Wb}
\frac{2|X|_{k}}{2^{\log^{1-(q-1)/L_{k}}k}}\le|W_{a_{j}}|_{k}\le|W_{b_{j}}|_{k}\le
 \frac{2|X|_{k}}{2^{\log^{1-q/L_{k}}k}}~.
\end{equation}
Also, $\hat{\chi}_{j}=\max\{\frac{|X|_{k}+1}{\left|W_{a}\right|_{k}},e\}$,
$\lambda_{j}=\frac{2\ln\hat{\chi}_{j}}{b_{j}'-a_{j}}=\frac{4L_{k}\ln\hat{\chi}_{j}}{R_{j}}=\frac{160L_{k}\ln\hat{\chi}_{j}}{\Delta}$. The radius $r_j$ of the petal $X_j=W_{r_j}$, is chosen from $[a_j,b'_j]$ with density function $f\left(r\right)=\frac{\lambda_{j}\cdot e^{-\lambda_{j}r}}
{e^{-\lambda_{j}\cdot a_{j}}-e^{-\lambda_{j}\cdot b'_{j}}}$.
Analogously, in the non-terminal case ($j>k+1$):
$L_n=\lceil\log\log n\rceil$, $R_j=hi-lo=\Delta/32$ (no reformation), $a_j=lo+(q_j-1)R_j/L_n$, $b'_j=a_j+R_j/(2L_n)$, $b_j=a_j+R_j/L_n$, where $q_j$ is chosen in such a way that $\frac{2|X|}{2^{\log^{1-(q_{j}-1)/L_{n}}n}}\le|W_{a_{j}}|\le|W_{b_{j}}|\le
\frac{2|X|}{2^{\log^{1-q_{j}/L_{n}}n}}$.
Also, $\hat\chi_j=\max\{\frac{|X|+1}{\left|W_{a}\right|},e\}$, $\lambda_j=\frac{2\ln \hat\chi_j}{b'_j-a_j}=\frac{4L_n\ln\hat\chi_j}{R_j}=\frac{128L_n\ln\hat\chi_j}{\Delta}$. The radius $r_j$ of the petal $X_j=W_{r_j}$, is chosen from $[a_j,b'_j]$ with density function $f\left(r\right)=\frac{\lambda_{j}\cdot e^{-\lambda_{j}r}}
{e^{-\lambda_{j}\cdot a_{j}}-e^{-\lambda_{j}\cdot b'_{j}}}$.

Let $\delta_j=e^{-8\lambda_j\gamma\Delta}$.
Towards the proof, we assume that $\gamma\le\frac{1}{2^{10}L_k}$ as otherwise the assertions of the lemma are trivial (as $\Pr[{\cal E}_{X,i}]\le\sum_{j}\sum_{Y\subseteq X}\Pr\left[\mathcal{S}_{X,i,Y,j}\wedge\overline{\mathcal{F}}_{X,i,j}\right]$).
\begin{claim}\label{Clm singlePetalCutProb}
For every $1\le j\le n+1$, $\Pr\left[C_{j}\mid \mathcal{Z}_{Y,j}\right]\le\left(1-\delta_j\right)\left(\Pr\left[
\overline{\mathcal{F}_{j}}\mid\mathcal{Z}_{Y,j}\right]+\frac{2}{\hat{\chi}_{j}^2}\right)$.
\end{claim}
\begin{proof}
For ease of notation we write simply $\lambda$ ,$\delta$, $\chi$, $a$ and $b'$ instead of $\lambda_j,\delta_j,\chi_j,a_j,b'_j$, as the proof is the same for both of the cases.\footnote{Note that if $X_j=\emptyset$ is imaginary petal created only for clarity in notation, then $\Pr\left[C_{j}\mid \mathcal{Z}_{Y,j}\right]=0$ and the claim is obvious.}
Let $\rho$ be the minimal number greater than $a$ such that $W_{\rho}\cap B_x\ne\emptyset$.
If $\rho\ge b'$, then trivially $\Pr\left[\overline{\mathcal{F}_{j}}\mid\mathcal{Z}_{Y,j}\right]=1$
and hence $\mathcal{P}_{r}\left[C_{j}\mid\mathcal{Z}_{Y,j}\right]=0$
and we are done. Otherwise, $\rho\in[a,b']$, the probability
that $B_x$ intersects $X_{j}$ is
\[
\Pr\left[\overline{\mathcal{F}_{j}}\mid\mathcal{Z}_{Y,j}\right]
=\intop_{\rho}^{b'}f\left(r\right)dr=\intop_{\rho}^{b'}\frac{\lambda\cdot e^{-\lambda r}}{e^{-\lambda\cdot a}-e^{-\lambda\cdot b'}}dr=\frac{e^{-\lambda \rho}-e^{-\lambda b'}}{e^{-\lambda\cdot a}-e^{-\lambda\cdot b'}}~.
\]
The ball $B_x$ is of diameter at most $2\gamma\Delta$, therefore by \claimref{ClmW_rProp}, $B_x\subseteq W_{\rho+8\gamma\Delta}$. Hence
\begin{eqnarray*}
\Pr\left[C_{j}\mid \mathcal{Z}_{Y,j}\right] & \le & \intop_{\rho}^{\rho+8\gamma\Delta}f\left(r\right)dr=
    \intop_{\rho}^{\rho+8\gamma\Delta}\frac{\lambda\cdot e^{-\lambda r}}{e^{-\lambda\cdot a}-e^{-\lambda\cdot b'}}dr=
    \frac{e^{-\lambda \rho}-e^{-\lambda\left(\rho+8\gamma\Delta\right)}}{e^{-\lambda\cdot a}-e^{-\lambda\cdot b'}}\\
 & = & \frac{e^{-\lambda \rho}}{e^{-\lambda\cdot a}-e^{-\lambda\cdot b'}}\left(1-e^{-\lambda8\gamma\Delta}\right)=
 \left(1-\delta\right)\frac{e^{-\lambda \rho}-e^{-\lambda\cdot b'}+e^{-\lambda\cdot b'}}{e^{-\lambda\cdot a}-e^{-\lambda\cdot b'}}\\
 & = & \left(1-\delta\right)\left(\frac{e^{-\lambda \rho}-e^{-\lambda\cdot b'}}{e^{-\lambda\cdot a}-e^{-\lambda\cdot b'}}+\frac{e^{-\lambda\cdot b'}}{e^{-\lambda\cdot a}-e^{-\lambda\cdot b'}}\right)\\
 & = & \left(1-\delta\right)\left(\Pr\left[\overline{\mathcal{F}_{j}}|\ \mathcal{Z}_{Y,j}\right]+\frac{1}{e^{\lambda\cdot\left(b'-a\right)}-1}\right)\\ & \le & (1-\delta)\left(\Pr\left[\overline{\mathcal{F}_{j}}|\ \mathcal{Z}_{Y,j}\right]+\frac{2}{\hat{\chi}^2}\right).
\end{eqnarray*}
Where the last inequality follows since $\hat\chi\ge2$.
\end{proof}

We are now ready to bound the probability that one of the terminal petals cuts $B_x$.
\begin{eqnarray}\label{eq:cutt}
\Pr\left[\bigvee_{j\le k+1}C_{j}\right]
& =& \sum_{j\le k+1}\sum_{Y\subseteq X}\Pr\left[C_{j}\mid\mathcal{Z}_{Y,j}\right]\cdot\Pr[\mathcal{Z}_{Y,j}]\\
 & \le &\sum_{j\le k+1}\sum_{Y}\left(1-\delta_{j}\right)\left(\Pr\left[\overline{\mathcal{F}_{j}}\mid\mathcal{Z}_{Y,j}\right]+\frac{2}{\hat{\chi}_{j}^{2}}\right)\cdot\Pr[\mathcal{Z}_{Y,j}]\nonumber\\
 & = &\sum_{j\le k+1}\sum_{Y}\left(1-\delta_{j}\right)\Pr\left[\overline{\mathcal{F}_{j}}\wedge\mathcal{Z}_{Y,j}\right]+\sum_{j\le k+1}\sum_{Y}\left(1-\delta_{j}\right)\frac{2}{\hat{\chi}_{j}^{2}}\cdot\Pr[\mathcal{Z}_{Y,j}]~.\nonumber
\end{eqnarray}
To see how the first equality follows, note that the probability that $B_x$ is cut by the first $k+1$ petals ($\Pr\left[\bigvee_{j\le k+1}C_{j}\right]$) is equal to the sum of the probabilities that $B_x$ cut at the first time by petal $j$ ($\sum_{j\le k+1}\Pr\left[\mathcal{C}_{j}\wedge\bigwedge_{l<j}\overline{C_{l}}\right]$). While for each $j$, the probability that $B_x$ is cut at the first time by petal $j$ is equal to the probability that $B_x$ is active at iteration $j$ (i.e. ${\cal F}_{(< j)}$), and is indeed cut by petal $j$ ($\sum_{Y\subseteq X}\Pr\left[\mathcal{C}_{j}\mid\mathcal{Z}_{Y,j}\right]\cdot\Pr[\mathcal{Z}_{Y,j}]$).

Note that for $j\le k+1$, $1-\delta_j=1-e^{-8\lambda_j\gamma\Delta}=
1-e^{-8\cdot 160\gamma L_k\ln\hat{\chi}_{j}}<
2^{11}\gamma L_k\ln\hat{\chi}_{j}$.
For a set $Y=Y_{j-1}$  such that $\Pr\left[\overline{\mathcal{F}_{j}}\wedge\mathcal{Z}_{Y,j}\right]\ne0$, necessarily a target $t_j$ and $q_j$ are chosen so that $B_x\cap W_{b'_j}\neq\emptyset$ (otherwise it is no possible that $X_j$ intersects $B_x$). As $B_x$ is a ball of radius $\gamma\Delta\le \Delta/(2^{10}L_k)$, increasing the petal radius by $4\gamma\Delta$ we have $x\in W_{b_{j}'+4\gamma\Delta}$, and  moreover $B\left(x,\frac{\Delta}{2^{9}L_{k}}\right)\subseteq W_{b_{j}'+4\gamma\Delta+\frac{\Delta}{2^{7}L_{k}}}\subseteq W_{b_{j}'+\frac{\Delta}{2^{8}L_{k}}+\frac{\Delta}{2^{7}L_{k}}}\subseteq W_{b_{j}'+\frac{R_j}{2L_{k}}}=W_{b_{j}}$.

We will show that for such a $Y$, (for which $\Pr\left[\overline{\mathcal{F}_{j}}\wedge\mathcal{Z}_{Y,j}\right]\ne0$ ), $\log\left(\hat{\chi}_{j}\right)\le6\log\left(\varphi_k\left(X,Y,x\right)\right)$.\footnote{
Recall that $\varphi_k\left(X,Y,x\right)=\max\left\{\frac{|X|_k}{\left|B_{Y}\left(x,
\Delta_{X}/(2^{8}L_{k})\right)\right|_{k}},e\right\}$.}
 We may assume that $\hat{\chi}_j=\frac{|X|_k+1}{\left|W_{a}\right|_{k}}$,
 as otherwise ($\hat{\chi}_j=e$) the bound is trivial. Using that
 $\log^{1/L_{k}}k\le3$ we get:

\begin{eqnarray*}
\log{\hat{\chi}_j}
&=& \log\left(\frac{|X|_{k}+1}{\left|W_{a_j}\right|_{k}}\right)\nonumber\\
&\le& \log\left(\frac{2|X|_{k}}{\left|W_{a_j}\right|_{k}}\right)\nonumber\\
&\stackrel{\eqref{eq:Wa le Wb}}{\le}&   \log^{1-(q-1)/L_{k}}k\nonumber\\
&\le& 3\log^{1-q/L_{k}}k\nonumber\\
&\stackrel{\eqref{eq:Wa le Wb}}{\le}& 3\log\left(\frac{2|X|_{k}}{\left|W_{b_j}\right|_{k}}\right)\nonumber\\
&\le& 3\log\left(\frac{2|X|_{k}}{\left|
    B\left(x,\Delta/2^{9}L_{k}\right)\right|_{k}}\right)\nonumber\\
&=& 3\left(\log\left(\varphi_k\left(X,Y,x\right)\right)+1\right)\nonumber\\
&\le& 6\log\left(\varphi_k\left(X,Y,x\right)\right)~.
\end{eqnarray*}

\sloppy Note that in particular it is true that $\ln{\hat{\chi}_j}\le 6\ln\left(\varphi_k\left(X,Y,x\right)\right)$. Hence we can bound the first component of the summation in
 $\eqref{eq:cutt}$:
 \begin{eqnarray}\label{eq:xi(X,Y,x)}
   \sum_{j\le k+1}\sum_{Y}\left(1-\delta_{j}\right)\Pr\left[\overline{\mathcal{F}_{j}}\wedge\mathcal{Z}_{Y,j}\right] &\le& 2^{11}\gamma L_{k}\sum_{j\le k+1}\sum_{Y}\ln\left(\hat{\chi}_{j}\right)\Pr\left[\overline{\mathcal{F}_{j}}\wedge\mathcal{Z}_{Y,j}\right] \\\nonumber
   &\le& 6\cdot2^{11}\gamma L_{k}\sum_{j\le k+1}\sum_{Y}\ln\left(\varphi_k\left(X,Y,x\right)\right)\Pr\left[\overline{\mathcal{F}_{j}}\wedge\mathcal{Z}_{Y,j}\right]
 \end{eqnarray}

For bounding the second component ($\sum_{j\le k+1}\sum_{Y}\left(1-\delta_{j}\right)\frac{2}{\hat{\chi}_{j}^{2}}\cdot\Pr[\mathcal{Z}_{Y,j}]$), note that $\frac{1}{\hat{\chi}_{j}}=
  \min\left\{ \frac{\left|W_{a_{j}}\right|_{k}}{\left|X\right|_{k}+1},
  \frac{1}{e}\right\} \le\frac{\left|W_{a_{j}}\right|_{k}}{\left|X\right|_{k}}$.
In addition, the probability of the event
${S}_{Y,j}$ is equal to the sum
of probabilities of sequences $X_0,\dots,X_s$,
over all sequences for which $Y=(X_j\cup\dots\cup X_s\cup X_0)$ (while we abuse notation, $\Pr[X_{0},\dots,X_{s}]$ is the probability that the petal decomposition returned the partition ${X_{0},\dots,X_{s}}$). Note also that for each such partition, $X_{0},\dots,X_{s}$, $\left|W_{a_{j}}\right|_{k}\le\left|X_{j}\right|_{k}$ and $\sum_{j\le k+1}\frac{\left|X_{j}\right|_{k}}{\left|X\right|_{k}}\le1$, because all the petals, 
$X_1,\dots,X_{k+1}$, , are pairwise disjoint. We get

\begin{eqnarray*}
\lefteqn{\sum_{j\le k+1}\sum_{Y}\frac{\Pr[{S}_{Y,j}]}{\hat{\chi_j}}}\\
&=&\sum_{j\le k+1}\sum_{Y:~B_{x}\subseteq Y}\sum_{(X_{0},\dots,X_{s}):~Y=(X_{j}\cup\dots\cup X_{s}\cup X_{0})}\Pr[X_{0},\dots,X_{s}]\frac{1}{\hat{\chi}_{j}}\nonumber\\
&\le&\sum_{j\le k+1}\sum_{Y:~B_{x}\subseteq Y}\sum_{(X_{0},\dots,X_{s})~:~Y=(X_{j}\cup\dots\cup X_{s}\cup X_{0})}\Pr[X_{0},\dots,X_{s}]\frac{\left|W_{a_{j}}\right|_{k}}{\left|X\right|_{k}}\nonumber\\
&\le&\sum_{j\le k+1}\sum_{\left(X_{0},\dots,X_{s}\right)}\sum_{Y:~B_{x}\subseteq Y,Y=(X_{j}\cup\dots\cup X_{s}\cup X_{0})}\Pr[X_{0},\dots,X_{s}]\frac{\left|X_{j}\right|_{k}}{\left|X\right|_{k}}\nonumber\\
&\le&\sum_{j\le k+1}\sum_{(X_{0},\dots,X_{s})}\Pr[X_{0},\dots,X_{s}]\frac{\left|X_{j}\right|_{k}}{\left|X\right|_{k}}\nonumber\\
&=& \sum_{(X_{0},\dots,X_{s})}\Pr[X_{0},\dots,X_{s}]\sum_{j\le k+1}\frac{\left|X_{j}\right|_{k}}{\left|X\right|_{k}}\nonumber\\
&\le&\sum_{(X_{0},\dots,X_{s})}\Pr[X_{0},\dots,X_{s}]=1~.
\end{eqnarray*}
Where the third inequalitie follows by the fact that for a fixed $j$ and a particular partition ${X_{0},\dots,X_{s}}$ of $X$ there might be only a single set $Y$ such that $B_x\subseteq Y=X_{j}\cup\dots\cup X_{s}\cup X_{0}$.
As $\frac{\ln\hat{\chi}_{j}}{\hat{\chi}_{j}}\le1$  we can bound the second component of the summation in $\eqref{eq:cutt}$:
 \begin{eqnarray}\label{secondPartBound}
   \sum_{j\le k+1}\sum_{Y}\left(1-\delta_{j}\right)\frac{2}{\hat{\chi}_{j}^{2}}\cdot\Pr[\mathcal{Z}_{Y,j}] &\le& 2^{11}\gamma L_{k}\sum_{j\le k+1}\sum_{Y}\ln\hat{\chi}_{j}\frac{2}{\hat{\chi}_{j}^{2}}\cdot\Pr[\mathcal{Z}_{Y,j}]\nonumber \\
   &\le& 2^{12}\gamma L_{k}\sum_{j\le k+1}\sum_{Y}\frac{\Pr[\mathcal{Z}_{Y,j}]}{\hat{\chi}_{j}}\nonumber\\
   &\le& 2^{12}\gamma L_{k}\nonumber\\
   &\le&2^{12}\gamma L_{k}\sum_{Y,j}\Pr\left[\overline{\mathcal{F}_{j}}\wedge\mathcal{Z}_{Y,j}\right]\ln\varphi_{k}(X,Y,x)~.
 \end{eqnarray}
Where the last inequality follows by $\sum_{Y,j}\Pr\left[\overline{\mathcal{F}_{j}}\wedge\mathcal{Z}_{Y,j}\right]\ln\varphi_{k}(X,Y,x)\ge\sum_{Y,j}
\Pr\left[\overline{\mathcal{F}_{j}}\wedge\mathcal{Z}_{Y,j}\right]=1$ (recall that all the probabilities are implicitly conditioned on ${\cal S}_{X,i}$, and ${\overline{\mathcal{F}_{j}}}$ has to hold for some $j$).
We conclude:
\begin{eqnarray*}
  \lefteqn{\Pr\left[\bigvee_{j\le k+1}C_{j}\right]}\\
  &\stackrel{\eqref{eq:cutt}}{\le}&\!\!\sum_{j\le k+1}\sum_{Y}\left(1-\delta_{j}\right)\Pr\left[\overline{\mathcal{F}_{j}}\wedge\mathcal{Z}_{Y,j}\right]+\!\!\sum_{j\le k+1}\sum_{Y}\left(1-\delta_{j}\right)\frac{2}{\hat{\chi}_{j}^{2}}\cdot\Pr[\mathcal{Z}_{Y,j}] \\
  &\stackrel{\eqref{eq:xi(X,Y,x)}\wedge\eqref{secondPartBound}}{\le}&\!\! \left(6\cdot2^{11}+2^{12}\right)\gamma L_{k}\sum_{j\le k+1}\sum_{Y}\Pr\left[\overline{\mathcal{F}_{j}}\wedge\mathcal{Z}_{Y,j}\right]\ln\left(\varphi_k\left(X,Y,v\right)\right)\\
  &=&\!\!2^{14}\gamma L_{k}\sum_{j\le k+1}\sum_{Y}\Pr\left[\overline{\mathcal{F}_{j}}\wedge\mathcal{Z}_{Y,j}\right]\ln\left(\varphi_k\left(X,Y,v\right)\right) ~.
\end{eqnarray*}
\begin{claim}
\label{Clm no trm petal no danger} If $x\in K$, then for $j>k+1$, it holds that $\Pr\left[C_{j}\mid \mathcal{Z}_{Y,j}\right]=0$.
\end{claim}
\begin{proof}
Since we condition on $\mathcal{Z}_{Y,j}$, it implies that $B_x\subseteq Y_{j-1}$. As $\gamma<1/16$, the third assertion of \claimref{Clm reform effect} implies that $X_j\cap B_x=\emptyset$.
\end{proof}
We are ready to prove the first statement of \lemmaref{lem:Cut Prob}. For terminal $v\in K$, Using \claimref{Clm no trm petal no danger},
\begin{eqnarray*}
  \Pr[{\cal E}]&=& \Pr\Big[\bigvee_{j\le n+1}C_{j}\Big]  = \Pr\Big[\bigvee_{j\le k+1}C_{j}\big]\\
  &\le& 2^{14}\gamma L_{k}\sum_{j\le k+1}\sum_{Y}\Pr\left[\overline{\mathcal{F}_{j}}\wedge\mathcal{Z}_{Y,j}\right]\ln\left(\varphi_k\left(X,Y,v\right)\right)~.
\end{eqnarray*}

Using a symmetric argument, we can also show the second assertion of \lemmaref{lem:Cut Prob}. Set $\mu=\Pr\left(\bigvee_{j\le k+1}C_{j}\right)$ the probability that $B_x$ is cut by terminal petal. Let $I_n={k+2,\dots,n+1}$ be all the indices of the non-terminal petals. Note that for $j\in I_n$, $1-\delta_{j}=1-e^{-8\lambda_{j}\gamma\Delta}=1-e^{-8\cdot128\cdot\gamma L_{n}\ln\hat{\chi}_{j}}\le2^{10}\gamma L_{n}\ln\hat{\chi}_{j}$. Hence for every vertex $x\in X$,
\begin{eqnarray*}
  \Pr[{\cal E}]  &=& \Pr\left[\bigvee_{j\le n+1}C_{j}\right] \\
  &=& \sum_{j\le n+1}\sum_{Y\subseteq X}\Pr\left[C_{j}\mid\mathcal{Z}_{Y,j}\right]\cdot\Pr[\mathcal{Z}_{Y,j}] \\
  &=& \mu + \sum_{j\in I_n}\sum_{Y\subseteq X}\Pr\left[C_{j}\mid\mathcal{Z}_{Y,j}\right]\cdot\Pr[\mathcal{Z}_{Y,j}]\\
  &\le& \mu+\sum_{j\in I_n}\sum_{Y}\left(1-\delta_{j}\right)\left(\Pr\left[
  \overline{\mathcal{F}_{j}}\mid\mathcal{Z}_{Y,j}\right]+
  \frac{2}{\hat{\chi}_{j}^{2}}\right)\cdot\Pr[\mathcal{Z}_{Y,j}]\\
  &\le& \mu+2^{10}\gamma L_{n}\left(\sum_{j\in I_{n}}\sum_{Y}\ln\hat{\chi}_{j}\cdot\Pr\left[\overline{\mathcal{F}_{j}}\wedge\mathcal{Z}_{Y,j}\right]+\sum_{j\in I_{n}}\sum_{Y}\frac{2}{\hat{\chi}_{j}}\cdot\Pr[\mathcal{Z}_{Y,j}]\right)
\end{eqnarray*}
where the first inequality is by \claimref{Clm singlePetalCutProb}. For $Y$ such that $\Pr\left[\overline{\mathcal{F}_{j}}\wedge\mathcal{Z}_{Y,j}\right]\ne0$ it holds that $B_x\cap W_{b'_j}\neq\emptyset$. As $\gamma\le\frac{1}{2^{10}L_k}\le \frac{1}{2^{10}L_n}$, by \claimref{ClmW_rProp} we have that $B_x\left(x,\frac{\Delta}{2^{9}L_{n}}\right)\subseteq W_{b_{j}'+4\gamma\Delta+\frac{\Delta}{2^{7}L_{n}}}\subseteq W_{b_{j}'+\frac{\Delta}{2^{8}L_{n}}+\frac{\Delta}{2^{7}L_{n}}}\subseteq W_{b_{j}'+\frac{R_j}{2L_{n}}}=W_{b_{j}}$.
We will show that $\ln\hat{\chi}_{j}\le6\ln\varphi_n\left(X,Y,x\right)$. If $\hat{\chi}_{j}=e$ this is trivial, hence we will assume that  $\hat{\chi}_{j}=\frac{|X|+1}{\left|W_{a}\right|}$.
By the maximality of $q_j$ ($\frac{2|X|}{2^{\log^{1-(q-1)/L_{n}}n}}\le|W_{a_{j}}|\le|W_{b_{j}}|\le\frac{2|X|}{2^{\log^{1-q/L_{n}}n}}$) we get
$\log\hat{\chi}_{j}=\log\left(\frac{|X|+1}{\left|W_{a_{j}}\right|}\right)\le\log\left(\frac{2|X|}{\left|W_{a_{j}}\right|}\right)\le\log^{1-(q-1)/L_{n}}n\le3\log^{1-q/L_{n}}n\le3\log\left(\frac{2|X|}{\left|W_{b}\right|}\right)\le3\log\left(\frac{2|X|}{\left|B\left(v,\Delta/2^{9}L_{n}\right)\right|}\right)\le3\left(\log\left(\varphi_{n}\left(X,Y,x\right)\right)+1\right)\le6\log\left(\varphi_{n}\left(X,Y,x\right)\right)$.

By similar arguments (and definitions) to \eqref{secondPartBound} we get
\begin{eqnarray*}
\sum_{j\in I_n}\sum_{Y}\frac{\Pr[{S}_{Y,j}]}{\hat{\chi_j}}
&\le& \sum_{j\in I_n}\sum_{Y:~B_{x}\subseteq Y}\sum_{(X_{0},\dots,X_{s})~:~Y=(X_{j}\cup\dots\cup X_{s}\cup X_{0})}\Pr[X_{0},\dots,X_{s}]\frac{\left|W_{a_{j}}\right|}{\left|X\right|}\\
&\le&\sum_{j\in I_n}\sum_{(X_{0},\dots,X_{s})}\Pr[X_{0},\dots,X_{s}]\frac{\left|X_{j}\right|}{\left|X\right|}\\
&=& \sum_{(X_{0},\dots,X_{s})}\Pr[X_{0},\dots,X_{s}]\sum_{j\le k+1}\frac{\left|X_{j}\right|}{\left|X\right|}~~\le~~1~.
\end{eqnarray*}
As $1\le\sum_{Y,j}\Pr\left[\overline{\mathcal{F}_{j}}\wedge\mathcal{Z}_{Y,j}\right]\ln\varphi_{n}(X,Y,x)$ we get
\begin{eqnarray*}
  \Pr[{\cal E}]
  &\le& \mu+2^{10}\gamma L_{n}\left(6\sum_{j\in I_{n}}\sum_{Y}\ln\left(\varphi_k\left(X,Y,x\right)\right)\cdot\Pr\left[\overline{\mathcal{F}_{j}}\wedge\mathcal{Z}_{Y,j}\right]+2\right)\\
  &\le& \mu+2^{13}\gamma L_{n}\sum_{j\in I_{n}}\sum_{Y}\ln\left(\varphi_n\left(X,Y,x\right)\right)\cdot\Pr\left[\overline{\mathcal{F}_{j}}\wedge\mathcal{Z}_{Y,j}\right]~.
\end{eqnarray*}
Hence the second assertion of \lemmaref{lem:Cut Prob} follows as well.

\section{Applications}\label{app:app}
In this section we illustrate several algorithmic applications of our techniques. Some of our applications are suitable for graphs with a small vertex cover.
Recall that for a graph $G=\left(V,E\right)$, a set $A\subseteq V$ is a vertex cover of
$G$, if for any edge $e\in E$, at least one of its endpoints
is in $A$. A polynomial time $2$-approximation algorithm
to this problem is folklore.

\subsection{Sparsest Cut }\label{app}

In the sparsest-cut problem we are given a graph $G=\left(V,E\right)$
with capacities on the edges $c:E\rightarrow\mathbb{R}_{+}$, and a collection of pairs $(s_1,t_1),\dots,(s_r,t_r)$ along with their demands $D_1,\dots,D_r$. The goal is to find a cut $S\subseteq V$ that minimizes the ratio between capacity and demand across the cut:
\[
\phi(S)=\frac{\sum_{\{u,v\}\in E}c(u,v)|\1_S(u)-\1_S(v)|}{\sum_{i=1}^rD_i|\1_S(s_i)-\1_S(t_i)|}~,
\]
where $\1_S(\cdot)$ is the indicator for membership in $S$.
Arora et. al. \cite{ALN08} present an $\tilde{O}\left(\sqrt{\log r}\right)$ approximation algorithm
to this problem. Our contribution is the following.
\begin{theorem}\label{thm:sparsest}
	If there exists a set $K\subseteq V$ of size $k$ such that any demand pair contains a vertex of $K$, then
	there exists a $\tilde{O}\left(\sqrt{\log k}\right)$ approximation algorithm for the sparsest-cut problem.
\end{theorem}

The key ingredient of the algorithm of \cite{ALN08}
is a non-expansive embedding from $\ell_{2}^{2}$ (negative-type metrics) into $\ell_{1}$,
which has $\tilde{O}\left(\sqrt{\log r}\right)$ contraction for all demand pairs.
We will use the strong terminal embedding for negative type metrics given in item (7) of \corollaryref{cor:implications} to improve the distortion to $\tilde{O}\left(\sqrt{\log k}\right)$.

We now elaborate on how to use the embedding of $\ell_2^2$ into $\ell_1$ to obtain an approximation algorithm for the sparsest-cut, all the details can be found in \cite{LLR95,ARV09,ALN08}, and we provide them just for completeness. First, write down the following SDP relaxation with triangle inequalities:

\begin{algorithm}[h]
	\caption{Sparsest Cut SDP Relaxation}

	$\min\sum_{\{u,v\}\in E}c(u,v)\cdot\Vert\bar{u}-\bar{v}\Vert_{2}^{2}$
	
	s.t.  $\sum_{i=1}^r D_i\cdot\Vert\bar{s_i}-\bar{t_i}\Vert_{2}^{2}=1$
	
	~~~~	For all $u,v,w\in V$, $\Vert\bar{u}-\bar{v}\Vert_{2}^{2}+\Vert\bar{v}-\bar{w}\Vert_{2}^{2}\ge\Vert\bar{u}-\bar{w}\Vert_{2}^{2}$
	
	~~~~ For all $u\in V$, $\bar{u}\in\R^n$
\end{algorithm}
Note that this is indeed a relaxation: if $S$ is the optimal cut, set $\rho=\sum_{i=1}^rD_i\cdot|\1_S(s_i)-\1_S(t_i)|$; for $u\in S$ set $\bar{u}=(\frac{1}{\sqrt{\rho}},0,...,0)$, and for $u\notin S$, set $\bar{u}=(0,...,0)$. It can be checked to be a feasible solution of value equal to that of the cut $S$.

Let $K\subseteq V$ be a vertex cover of the demand graph $(V,\{\{s_i,t_i\}_{i=1}^r\})$ of size at most $2k$ (recall that
we can find such a cover in polynomial time). Let $X=\left\{ \bar{v}\in\R^n\mid v\in V\right\} $
be an optimal solution to the SDP (it can be computed in polynomial time), which is in particular an $\ell_{2}^{2}$ (pseudo) metric. By \corollaryref{cor:implications}  there exists a non-expansive embedding $f:X\rightarrow\ell_{1}$
with terminal distortion $\tilde{O}\left(\sqrt{\log k}\right)$ (where $K$
is the terminal set).\footnote{The embedding of \corollaryref{cor:implications} is in fact into $\ell_2$, but there is an efficient randomized algorithm to embed $\ell_2$ into $\ell_1$ with constant distortion \cite{FLM77}.}  This implies that for any $u,v\in V$ and any $1\le i\le r$,

\begin{eqnarray}\label{eq:llff}
	&&\nonumber\Vert\bar{u}-\bar{v}\Vert_{2}^{2} \ge \Vert f(\bar{v})-f\left(\bar{u}\right)\Vert_{1}\\
	&&\Vert\bar{s_i}-\bar{t_i}\Vert_{2}^{2}\le \tilde{O}(\sqrt{\log k})\cdot \Vert f(\bar{s_i})-f(\bar{t_i})\Vert_1~.
\end{eqnarray}
Let $\Vert f(\bar{v})-f\left(\bar{u}\right)\Vert_{1}=\sum_{S\subseteq V}\alpha_{S}\left|\1_{S}(v)-\1_{S}(u)\right|$
be a representation of the $\ell_1$ metric as a nonnegative linear combination of cut metrics (it is well known that there is such a representation with polynomially many cuts $S$ having $\alpha_S>0$). We conclude
\begin{eqnarray*}
	{\rm opt(SDP)}&=&		\sum_{\{u,v\}\in E}c(u,v)\cdot\Vert\bar{u}-\bar{v}\Vert_{2}^{2}\\ & = & \frac{\sum_{\{u,v\}\in E}c(u,v)\cdot\Vert\bar{u}-\bar{v}\Vert_{2}^{2}}{\sum_{i=1}^r D_i\cdot\Vert\bar{s_i}-\bar{t_i}\Vert_{2}^{2}}\\
	&\stackrel{\eqref{eq:llff}}{\ge} & \frac{\sum_{\{u,v\}\in E}c(u,v)\cdot\Vert f(\bar{v})-f\left(\bar{u}\right)\Vert_{1}}{\sum_{i=1}^r D_i\cdot \tilde{O}\left(\sqrt{\log k}\right)\cdot\Vert f(\bar{s_i})-f\left(\bar{t_i}\right)\Vert_{1}}\\
	& = & \frac{1}{\tilde{O}\left(\sqrt{\log k}\right)}\cdot\frac{\sum_{\{u,v\}\in E}c(u,v)\cdot\sum_{S\subsetneq V}\alpha_{S}\left|\1_{S}(v)-\1_{S}(u)\right|}{\sum_{i=1}^r D_i\cdot\sum_{S\subsetneq V}\alpha_{S}\left|\1_{S}(s_i)-\1_{S}(t_i)\right|}\\
	& \ge & \frac{1}{\tilde{O}\left(\sqrt{\log k}\right)}\min_{S: \alpha_S>0}\frac{\sum_{\{u,v\}\in E}c(u,v)\cdot\left|\1_{S}(v)-\1_{S}(u)\right|}{\sum_{i=1}^r D_i\cdot\left|\1_{S}(s_i)-\1_{S}(t_i)\right|}\\
	&=&\min_{S: \alpha_S>0} \frac{\phi(S)}{\tilde{O}\left(\sqrt{\log k}\right)}~.
\end{eqnarray*}
In particular, among the polynomially many sets $S\subseteq V$ with $\alpha_S>0$, there exists one which has sparsity at most $\tilde{O}(\sqrt{\log k})$ times larger than the optimal one.

\subsection{Min Bisection}

In the min-bisection problem, we are given a graph $G=\left(V,E\right)$ on an even number $n$ of vertices,
with capacities $c:E\rightarrow\mathbb{R}_{+}$. The purpose
is to find a partition of $V$ into two equal parts $S\subseteq V$ and $\bar{S}=V\setminus S$, that minimizes $\sum_{e\in E(S,\bar{S})}c(e)$.
This problem is NP-hard, and the best known approximation is $O\left(\log n\right)$
by \cite{R08}. We obtain the following generalization.

\begin{theorem}
There exists a $O(\log k)$ approximation algorithm for min-bisection, where $k$ is the size of a minimal vertex cover of the input graph.
\end{theorem}
\begin{proof}
Our algorithm follows closely the algorithm of \cite{R08}, the major difference is that we use our embedding into trees with terminal congestion. Let $K\subseteq V$ be the set of terminals, which is a vertex cover of size at most $2k$, and ${\cal D}$ a distribution over trees with strong terminal congestion $(O(\log k), O(\log n))$ given by \corollaryref{cor:cong}. The algorithm will sample a tree $T=(V,E_T)$ from ${\cal D}$, find an optimal bisection in $T$ and return it. We refer the reader to \sectionref{sec:cong} for details on notation and on the definition of capacities $C_T:E_T\to\R_+$ for $T$. We note that there is polynomial time algorithm (by dynamic programming) to find a min-bisection in trees.

It remains to analyze the algorithm. Let $S\subseteq V$ be the optimal solution in $G$, and $S_T$ be the optimal bisection for the tree $T$. The expected cost of using $S_T$ in $G$ can be bounded using \lemmaref{lem:capacities inequality} as follows
\begin{eqnarray*}
\sum_{T\in\mbox{supp}\left(\mathcal{D}\right)}\Pr\left[T\right]\sum_{e\in E\left(S_{T},\bar{S}_{T}\right)}c\left(e\right)
		&\stackrel{\eqref{eq:capacities inequality}}{\le}&\sum_{T}\Pr\left[T\right]\sum_{e'\in E_{T}\left(S_{T},\bar{S}_{T}\right)}C_T\left(e'\right)\\
&\le& \sum_{T}\Pr\left[T\right]\sum_{e'\in E_{T}\left(S,\bar{S}\right)}C_{T}\left(e'\right)\\
& \stackrel{\eqref{eq:capacities inequality}}{\le} & \sum_{T}\Pr\left[T\right]\sum_{e\in E\left(S,\bar{S}\right)}\mbox{load}_{T}\left(e\right)\\
		& = & \sum_{e\in E\left(S,\bar{S}\right)}\mathbb{E}_{T\sim\mathcal{D}}\left[\mbox{load}_{T}\left(e\right)\right] \\
		& \le & \sum_{e\in E\left(S,\bar{S}\right)}O\left(\log k\right)\cdot c(e)\\
&=&O\left(\log k\right)\cdot\mbox{opt}\left(G\right)\ ~,
	\end{eqnarray*}
where the last inequality uses that every edge touches a terminal, so its expected congestion is $O(\log k)$.
	
\end{proof}

\subsection{Online Algorithms: Constrained File Migration}

We illustrate the usefulness of our probabilistic terminal embedding into ultrametric via the constrained file migration problem. This is an online problem, in which we are given a graph $G=(V,E)$ representing a network, each node $v\in V$ has a memory capacity $m_v$, and a parameter $D\ge 1$. There is some set of files that reside at the nodes, at most $m_v$ files may be stored at node $v$ in any given time. The cost of accessing a file that currently lies at $v$ from node $u$ is $d_G(u,v)$ (no copies of files are allowed). Files can also be migrated from one node to another, this costs $D$ times the distance. When a sequence of file requests arrives online, the goal is to minimize the cost of serving all requests. The {\em competitive ratio} of an online algorithm is the maximal ratio between its cost to the cost of an optimal (offline) solution. For randomized algorithms the expected cost is used.

We consider the case where there exists a small set of vertices which are allowed to store files
(i.e. $m_{v}>0$). One may think about these vertices as servers who store files, while allowing file requests from all end users.
Let $K\subseteq V$ be the set of terminal vertices that are allowed to store files, with $\left|K\right|=k$. Our result is captured by the following theorem.
\begin{theorem}
There is a randomized algorithm for the constrained file migration problem with competitive ratio $O(\log m\cdot\log k)$, where $k$ vertices can store files and $m$ is the total memory available.
\end{theorem}

This theorem generalizes a result of \cite{B96}, who showed an algorithm with competitive ratio $O(\log m\cdot\log n)$ for arbitrary graphs on $n$ nodes. Both results are based on the following theorem. (Recall that a 2-HST is an ultrametric (see \defref{def:ultra}) such that the ratio between the label of a node to any of its children's label is at least 2.)
\begin{theorem}[\cite{B96}]
	\label{thm:Distributed paging in a ultrametric} For any
	2-HST, there is a randomized algorithm with competitive ratio $O\left(\log m\right)$
	for constrained file migration with total memory $m$.
\end{theorem}

By \theoremref{thm:strong-ultra}  there is a distribution $\mathcal{D}$ over embeddings
of $G$ into ultrametrics with expected terminal distortion $O\left(\log k\right)$, but in fact every tree in that distribution is a 2-HST.
Assume that in the optimal (offline) solution there are $s_{uv}$ times a file residing on $v$ was accessed by $u$, and $t_{uv}$ files were migrated from $v$ to $u$. Let $c_{uv}=s_{uv}+D\cdot t_{uv}$ be the total cost of file traffic from $v$ to $u$. Note that as $m_{v}=0$ for any $v\notin K$, then for any $u\in V$ we have $c_{uv}=0$.
Using the fact that the terminal distortion guarantee of $\mathcal{D}$ applies to all of the relevant distances, we obtain that
	\begin{eqnarray}\label{eq:frrf}
		{\rm opt}_G  &=&\sum_{u\in V, v\in K}c_{uv}\cdot d_G(u,v)\\\nonumber
& \ge & \frac{1}{O(\log k)}\cdot\sum_{u\in V, v\in K}c_{uv}\cdot\mathbb{E}_{T\sim\mathcal{D}}[d_{T}(u,v)]\\ \label{eq:opt G KDPP}
		& = & \frac{1}{O(\log k)}\cdot\mathbb{E}_{T\sim\mathcal{D}}\Big[\sum_{u\in V, v\in K}c_{uv}\cdot d_{T}(u,v)\Big]~ .\nonumber
	\end{eqnarray}

Observe that for any tree $T\in\mbox{supp}\left(\mathcal{D}\right)$ we could have served the request sequence in the same manner as the optimal algorithm, which would have the cost $\sum_{u\in V, v\in K}c_{uv}\cdot d_{T}(u,v)$. In particular, the optimal solution ${\rm opt}_T$ for the same requests with the input graph $T$ cannot be larger than that, i.e.
	\begin{equation}
		\sum_{u\in V, v\in K}c_{uv}\cdot d_{T}(u,v)\ge{\rm opt}_T~.\label{eq:opt T KDPP}
	\end{equation}
	Our algorithm will operate as follows: Pick a random tree according
	to the distribution $\mathcal{D}$, pick a random strategy $S$ for
	transmitting files in $T$ according to the distribution $\mathcal{S}(T)$
	guaranteed to exists by \theoremref{thm:Distributed paging in a ultrametric}, and serve the requests according to $S$. Denote by ${\rm cost}_H(S)$ the cost of applying strategy $S$ with distances taken in the graph $H$. For any possible $T\in\mbox{supp}\left(\mathcal{D}\right)$ it holds that
	\begin{equation}
		{\rm opt}_T\ge\frac{\mathbb{E}_{S\sim\mathcal{S}(T)}[\mbox{cost}_{T}\left(S\right)]}{O(\log m)}\ge\frac{\mathbb{E}_{S\sim\mathcal{S}(T)}[\mbox{cost}_{G}\left(S\right)]}{O(\log m)} ,\label{eq: avg cost of stategy}
	\end{equation}
	where the last inequality holds since $T$ dominates $G$ (i.e. $d_T(u,v)\ge d_G(u,v)$ for all $u,v\in V$). Combining equations \eqref{eq:frrf}, \eqref{eq:opt G KDPP}
	and \eqref{eq: avg cost of stategy} we get that
	\[
	{\rm opt}_G\ge\frac{\mathbb{E}_{T\sim\mathcal{D}}\mathbb{E}_{S\sim\mathcal{S}(T)}[\mbox{cost}_{G}\left(S\right)]}{O(\log k\log m)}~.
	\]
	Hence our randomized algorithm has $O\left(\log m\log k\right)$ competitive
	ratio, as promised.

\section{Acknowledgements}
We would like to thank Robert Krauthgamer, Yair Bartal and Manor Mendel for fruitful discussions, and to an anonymous referee for an idea leading to \theoremref{thm:lp-strong}.

{\small
	\bibliographystyle{alpha}
	\bibliography{bib-extended}

\newcommand{\etalchar}[1]{$^{#1}$}
\begin{thebibliography}{AKPW95}

\bibitem[ABC{\etalchar{+}}05]{ABCD05}
Ittai Abraham, Yair Bartal, T-H.~Hubert Chan, Kedar~Dhamdhere Dhamdhere, Anupam
  Gupta, Jon Kleinberg, Ofer Neiman, and Aleksandrs Slivkins.
\newblock Metric embeddings with relaxed guarantees.
\newblock In {\em Proceedings of the 46th Annual IEEE Symposium on Foundations
  of Computer Science}, FOCS '05, pages 83--100, Washington, DC, USA, 2005.
  IEEE Computer Society.

\bibitem[ABN08]{ABN08}
Ittai Abraham, Yair Bartal, and Ofer Neiman.
\newblock Nearly tight low stretch spanning trees.
\newblock In {\em FOCS '08: Proceedings of the 2008 49th Annual IEEE Symposium
  on Foundations of Computer Science}, pages 781--790, Washington, DC, USA,
  2008. IEEE Computer Society.

\bibitem[ABN11]{ABN06}
Ittai Abraham, Yair Bartal, and Ofer Neiman.
\newblock Advances in metric embedding theory.
\newblock {\em Advances in Mathematics}, 228(6):3026 -- 3126, 2011.

\bibitem[ABN15]{ABN07}
Ittai Abraham, Yair Bartal, and Ofer Neiman.
\newblock Embedding metrics into ultrametrics and graphs into spanning trees
  with constant average distortion.
\newblock {\em {SIAM} J. Comput.}, 44(1):160--192, 2015.

\bibitem[ABP92]{ABP92}
B.~Awerbuch, A.~Baratz, and D.~Peleg.
\newblock Efficient broadcast and light-weight spanners.
\newblock Technical Report CS92-22, The Weizmann Institute of Science, Rehovot,
  Israel., 1992.

\bibitem[ADD{\etalchar{+}}93]{ADDJS93}
I.~Alth$\ddot{\mbox{o}}$fer, G.~Das, D.~P. Dobkin, D.~Joseph, and J.~Soares.
\newblock On sparse spanners of weighted graphs.
\newblock {\em Discrete \& Computational Geometry}, 9:81--100, 1993.

\bibitem[AF09]{AF09}
Reid Andersen and Uriel Feige.
\newblock Interchanging distance and capacity in probabilistic mappings.
\newblock Technical report, 2009.

\bibitem[AKPW95]{AKPW95}
Noga Alon, Richard~M. Karp, David Peleg, and Douglas West.
\newblock A graph-theoretic game and its application to the $k$-server problem.
\newblock {\em SIAM J. Comput.}, 24(1):78--100, 1995.

\bibitem[ALN07]{ALN07}
Sanjeev Arora, James~R. Lee, and Assaf Naor.
\newblock Fr{\'{e}}chet embeddings of negative type metrics.
\newblock {\em Discrete {\&} Computational Geometry}, 38(4):726--739, 2007.

\bibitem[ALN08]{ALN08}
S.~Arora, J.~R. Lee, and A.~Naor.
\newblock Euclidean distortion and the sparsest cut.
\newblock {\em Journal of the American Mathematical Society 21}, 1:1--21, 2008.

\bibitem[AN12]{AN12}
Ittai Abraham and Ofer Neiman.
\newblock Using petal-decompositions to build a low stretch spanning tree.
\newblock In {\em STOC}, pages 395--406, 2012.

\bibitem[AR98]{AR98}
Yonatan Aumann and Yuval Rabani.
\newblock An {$O(\log k)$} approximate min-cut max-flow theorem and
  approximation algorithm.
\newblock {\em SIAM Journal on Computing}, 27(1):291--301, 1998.

\bibitem[ARV09]{ARV09}
Sanjeev Arora, Satish Rao, and Umesh~V. Vazirani.
\newblock Expander flows, geometric embeddings and graph partitioning.
\newblock {\em J. {ACM}}, 56(2), 2009.

\bibitem[Bar96]{B96}
Yair Bartal.
\newblock Probabilistic approximations of metric spaces and its algorithmic
  applications.
\newblock In {\em FOCS}, pages 184--193, 1996.

\bibitem[Bar04]{B04}
Yair Bartal.
\newblock Graph decomposition lemmas and their role in metric embedding
  methods.
\newblock In {\em ESA}, pages 89--97, 2004.

\bibitem[BFN16]{BFN16}
Yair Bartal, Arnold Filtser, and Ofer Neiman.
\newblock Constructing almost minimum spanning trees with constant average
  distortion.
\newblock In {\em SODA}, 2016.

\bibitem[BFR95]{BFR95}
Yair Bartal, Amos Fiat, and Yuval Rabani.
\newblock Competitive algorithms for distributed data management.
\newblock {\em J. Comput. Syst. Sci.}, 51(3):341--358, 1995.

\bibitem[BLMN05]{DBLP:journals/dcg/BartalLMN05}
Yair Bartal, Nathan Linial, Manor Mendel, and Assaf Naor.
\newblock Some low distortion metric ramsey problems.
\newblock {\em Discrete {\&} Computational Geometry}, 33(1):27--41, 2005.

\bibitem[Bou85]{B85}
J.~Bourgain.
\newblock On lipschitz embedding of finite metric spaces in hilbert space.
\newblock {\em Israel Journal of Mathematics}, 52(1-2):46--52, 1985.

\bibitem[CE05]{CE05}
Don Coppersmith and Michael Elkin.
\newblock Sparse source-wise and pair-wise distance preservers.
\newblock In {\em Proceedings of the Sixteenth Annual ACM-SIAM Symposium on
  Discrete Algorithms}, SODA '05, pages 660--669, Philadelphia, PA, USA, 2005.
  Society for Industrial and Applied Mathematics.

\bibitem[CKR04]{CKR01}
Gruia C{\u{a}}linescu, Howard~J. Karloff, and Yuval Rabani.
\newblock Approximation algorithms for the 0-extension problem.
\newblock {\em SIAM Journal on Computing}, 34(2):358--372, 2004.

\bibitem[CLLM10]{CLLM10}
Moses Charikar, Tom Leighton, Shi Li, and Ankur Moitra.
\newblock Vertex sparsifiers and abstract rounding algorithms.
\newblock In {\em 51th Annual {IEEE} Symposium on Foundations of Computer
  Science, {FOCS} 2010, October 23-26, 2010, Las Vegas, Nevada, {USA}}, pages
  265--274, 2010.

\bibitem[CXKR06]{CXKR06}
T.-H. Chan, Donglin Xia, Goran Konjevod, and Andrea Richa.
\newblock A tight lower bound for the steiner point removal problem on trees.
\newblock In {\em Proceedings of the 9th International Conference on
  Approximation Algorithms for Combinatorial Optimization Problems, and 10th
  International Conference on Randomization and Computation},
  APPROX'06/RANDOM'06, pages 70--81, Berlin, Heidelberg, 2006. Springer-Verlag.

\bibitem[EEST08]{EEST05}
Michael Elkin, Yuval Emek, Daniel~A. Spielman, and Shang-Hua Teng.
\newblock Lower-stretch spanning trees.
\newblock {\em SIAM Journal on Computing}, 38(2):608--628, 2008.

\bibitem[EFN15]{EFN15}
Michael Elkin, Arnold Filtser, and Ofer Neiman.
\newblock Prioritized metric structures and embedding.
\newblock In {\em Proceedings of the Forty-Seventh Annual {ACM} on Symposium on
  Theory of Computing, {STOC} 2015, Portland, OR, USA, June 14-17, 2015}, pages
  489--498, 2015.

\bibitem[EGK{\etalchar{+}}14]{EGKRTT14}
Matthias Englert, Anupam Gupta, Robert Krauthgamer, Harald R{\"{a}}cke, Inbal
  Talgam{-}Cohen, and Kunal Talwar.
\newblock Vertex sparsifiers: New results from old techniques.
\newblock {\em {SIAM} J. Comput.}, 43(4):1239--1262, 2014.

\bibitem[ES11]{ES11}
Michael Elkin and Shay Solomon.
\newblock Steiner shallow-light trees are exponentially lighter than spanning
  ones.
\newblock In {\em {IEEE} 52nd Annual Symposium on Foundations of Computer
  Science, {FOCS} 2011, Palm Springs, CA, USA, October 22-25, 2011}, pages
  373--382, 2011.

\bibitem[FLM77]{FLM77}
T.~Figiel, J.~Lindenstrauss, and V.D. Milman.
\newblock The dimension of almost spherical sections of convex bodies.
\newblock {\em Acta Mathematica}, 139(1):53--94, 1977.

\bibitem[FRT04]{FRT03}
Jittat Fakcharoenphol, Satish Rao, and Kunal Talwar.
\newblock A tight bound on approximating arbitrary metrics by tree metrics.
\newblock {\em J. Comput. Syst. Sci.}, 69(3):485--497, 2004.

\bibitem[GNR10]{GNR10}
Anupam Gupta, Viswanath Nagarajan, and R.~Ravi.
\newblock An improved approximation algorithm for requirement cut.
\newblock {\em Oper. Res. Lett.}, 38(4):322--325, 2010.

\bibitem[GP68]{GilbertP68}
E.~N. Gilbert and H.~O. Pollak.
\newblock Steiner minimal trees.
\newblock {\em SIAM Journal on Applied Mathematics}, 16(1):1--29, Jan 1968.

\bibitem[Gup01a]{G01}
Anupam Gupta.
\newblock Steiner points in tree metrics don't (really) help.
\newblock In {\em Proceedings of the Twelfth Annual ACM-SIAM Symposium on
  Discrete Algorithms}, SODA '01, pages 220--227, Philadelphia, PA, USA, 2001.
  Society for Industrial and Applied Mathematics.

\bibitem[Gup01b]{DBLP:conf/soda/Gupta01}
Anupam Gupta.
\newblock Steiner points in tree metrics don't (really) help.
\newblock In {\em SODA}, pages 220--227, 2001.

\bibitem[HPIS13]{HIS13}
Sariel Har-Peled, Piotr Indyk, and Anastasios Sidiropoulos.
\newblock Euclidean spanners in high dimensions.
\newblock In {\em Proceedings of the Twenty-Fourth Annual ACM-SIAM Symposium on
  Discrete Algorithms}, SODA '13, pages 804--809. SIAM, 2013.

\bibitem[JL84]{JL84}
William Johnson and Joram Lindenstrauss.
\newblock Extensions of {L}ipschitz mappings into a {H}ilbert space.
\newblock {\em Contemporary Mathematics}, 26:189–206, 1984.

\bibitem[KKN14]{KKN14}
Lior Kamma, Robert Krauthgamer, and Huy~L. Nguyen.
\newblock Cutting corners cheaply, or how to remove steiner points.
\newblock In {\em SODA}, pages 1029--1040, 2014.

\bibitem[KLMN05]{KLMN04}
Robert Krauthgamer, James~R. Lee, Manor Mendel, and Assaf Naor.
\newblock Measured descent: a new embedding method for finite metrics.
\newblock {\em Geometric and Functional Analysis}, 15(4):839--858, 2005.

\bibitem[KRY95]{KRY95}
Samir Khuller, Balaji Raghavachari, and Neal~E. Young.
\newblock Balancing minimum spanning trees and shortest-path trees.
\newblock {\em Algorithmica}, 14(4):305--321, 1995.

\bibitem[KSW09]{KSW04}
Jon Kleinberg, Aleksandrs Slivkins, and Tom Wexler.
\newblock Triangulation and embedding using small sets of beacons.
\newblock {\em J. ACM}, 56(6):32:1--32:37, September 2009.

\bibitem[KV13]{KV13}
Telikepalli Kavitha and Nithin~M. Varma.
\newblock Small stretch pairwise spanners.
\newblock In {\em Proceedings of the 40th International Conference on Automata,
  Languages, and Programming - Volume Part I}, ICALP'13, pages 601--612,
  Berlin, Heidelberg, 2013. Springer-Verlag.

\bibitem[LLR95]{LLR95}
N.~Linial, E.~London, and Y.~Rabinovich.
\newblock The geometry of graphs and some of its algorithmic applications.
\newblock {\em Combinatorica}, 15(2):215--245, 1995.

\bibitem[LR99]{LR99}
Frank~Thomson Leighton and Satish Rao.
\newblock Multicommodity max-flow min-cut theorems and their use in designing
  approximation algorithms.
\newblock {\em J. ACM}, 46(6):787--832, 1999.

\bibitem[Mat96]{M96}
Jiri Matou{\v{s}}ek.
\newblock On the distortion required for embeding finite metric spaces into
  normed spaces.
\newblock volume~93, pages 333--344, 1996.

\bibitem[MM10]{MM10}
Konstantin Makarychev and Yury Makarychev.
\newblock Metric extension operators, vertex sparsifiers and lipschitz
  extendability.
\newblock In {\em 51th Annual {IEEE} Symposium on Foundations of Computer
  Science, {FOCS} 2010, October 23-26, 2010, Las Vegas, Nevada, {USA}}, pages
  255--264, 2010.

\bibitem[MN06]{MN06}
Manor Mendel and Assaf Naor.
\newblock Ramsey partitions and proximity data structures.
\newblock In {\em FOCS}, pages 109--118, 2006.

\bibitem[Moi09]{M09}
Ankur Moitra.
\newblock Approximation algorithms for multicommodity-type problems with
  guarantees independent of the graph size.
\newblock In {\em FOCS}, pages 3--12, 2009.

\bibitem[R\"02]{R02}
Harald R\"{a}cke.
\newblock Minimizing congestion in general networks.
\newblock In {\em Proceedings of the 43rd Symposium on Foundations of Computer
  Science}, FOCS '02, pages 43--52, Washington, DC, USA, 2002. IEEE Computer
  Society.

\bibitem[R{\"a}c08]{R08}
Harald R{\"a}cke.
\newblock Optimal hierarchical decompositions for congestion minimization in
  networks.
\newblock In {\em STOC}, pages 255--264, 2008.

\bibitem[RR98]{RR98}
Yuri Rabinovich and Ran Raz.
\newblock Lower bounds on the distortion of embedding finite metric spaces in
  graphs.
\newblock {\em Discrete {\&} Computational Geometry}, 19(1):79--94, 1998.

\bibitem[RTZ05]{RTZ05}
Liam Roditty, Mikkel Thorup, and Uri Zwick.
\newblock Deterministic constructions of approximate distance oracles and
  spanners.
\newblock In {\em Proceedings of the 32Nd International Conference on Automata,
  Languages and Programming}, ICALP'05, pages 261--272, Berlin, Heidelberg,
  2005. Springer-Verlag.

\end{thebibliography}
}

\clearpage
\pagenumbering{roman}
\appendix
\centerline{\LARGE\bf Appendix}
\APPENDMETRICLB

\end{document}